\documentclass[sigconf,authorversion]{acmart}

\usepackage{amsthm}
\usepackage{mathtools}
\usepackage{multicol}
\usepackage{xcolor}
\usepackage{subfig}
\usepackage{csquotes}
\usepackage{bbm}
\usepackage{pifont}
\usepackage{multirow}
\usepackage{makecell}
\usepackage{balance}
\usepackage{microtype}
\usepackage{trimclip}
\usepackage[ruled, vlined,linesnumbered]{algorithm2e}

\usepackage{tikz}
\usetikzlibrary{arrows,decorations.pathmorphing,backgrounds,positioning,fit,petri,calc,positioning}
\usepackage{xifthen}
\usepackage{relsize}

\definecolor{light blue}{HTML}{ADD8E6}
\definecolor{darker blue}{HTML}{9BC2CF}
\definecolor{ciano}{HTML}{00FFFF}
\definecolor{verde}{HTML}{007a00 }
\definecolor{arancio}{HTML}{ffa31a}
\definecolor{violetto}{HTML}{8080ff}
\definecolor{white}{HTML}{FFFFFF}
\definecolor{rossiccio}{HTML}{d64161}
\definecolor{ForestGreen}{HTML}{32CD32}

\newtheorem{problem}{Problem}
\DeclareMathOperator*{\argmax}{arg\,max}
\DeclareMathOperator*{\argmin}{arg\,min}

\newcommand{\vmark}{\ding{51}\xspace}
\newcommand{\xmark}{\ding{55}\xspace}
\newcommand{\algex}{\texttt{Exact}\xspace}
\newcommand{\algkapp}{\texttt{Greedy}\xspace}
\newcommand{\algbatch}{\texttt{Batch}\xspace}
\newcommand{\algrand}{\texttt{Prob\-Peel}\xspace}
\newcommand{\alghybrid}{\texttt{Hybrid\-Peel}\xspace}

\newcommand{\bigO}{\ensuremath{\mathcal{O}}\xspace}
\newcommand{\NP}{\ensuremath{\mathbf{NP}}\xspace}
\newcommand{\NPhard}{{\NP}{-hard}\xspace}
\newcommand{\ourproblem}{{TMDS}\xspace}
\newcommand{\OPT}{\ensuremath{\mathrm{OPT}}\xspace}
\newcommand{\fobj}{\ensuremath{\rho}\xspace}

\newcommand{\weightfunction}{\ensuremath{\tau}}

\newcommand{\temporalDegreeWeight}[2]{\ensuremath{{d}_{#1}^{\weightfunction}(#2)}\xspace}
\newcommand{\temporalDegreeWeightNetwork}[2]{\ensuremath{{d}_{#1}^{\weightfunction}(#2)}\xspace}
\newcommand{\estimateTemporalDegreeWeight}[2]{\ensuremath{\hat{d}_{#1}^{\weightfunction}(#2)}\xspace}
\newcommand{\estimateTemporalDegreeWeightPrime}[3]{\ensuremath{\hat{d}_{#1}^{\weightfunction}(#2, #3)}\xspace}

\newif\ifextended
\extendedtrue

\newcommand{\suitename}{\textsf{ALDENTE}\xspace}

\newcommand{\suiteex}{\texttt{Exact}\xspace}
\newcommand{\suitegreedy}{\texttt{Greedy}\xspace}
\newcommand{\suitebatch}{\texttt{Batch}\xspace}
\newcommand{\suiterand}{\texttt{ProbPeel}\xspace}
\newcommand{\suitehybrid}{\texttt{HybridPeel}\xspace}

\newcommand{\spara}[1]{\smallskip\noindent{\bf #1}\xspace}
\newcommand{\para}[1]{\noindent{\bf #1}\xspace}

\AtBeginDocument{%
  }

\copyrightyear{2024}
\acmYear{2024}
\setcopyright{rightsretained}
\acmConference[KDD '24]{Proceedings of the 30th ACM SIGKDD Conference on Knowledge Discovery and Data Mining}{August 25--29, 2024}{Barcelona, Spain}
\acmBooktitle{Proceedings of the 30th ACM SIGKDD Conference on Knowledge Discovery and Data Mining (KDD '24), August 25--29, 2024, Barcelona, Spain}\acmDOI{10.1145/3637528.3671889}
\acmISBN{979-8-4007-0490-1/24/08}

\begin{document}

\title{Scalable Temporal Motif Densest Subnetwork Discovery}

\author{Ilie Sarpe}
\orcid{0009-0007-5894-0774}
\affiliation{%
	\institution{KTH Royal Institute of Technology}
	\city{Stockholm}
	\country{Sweden}
}
\email{ilsarpe@kth.se}

\author{Fabio Vandin}
\orcid{0000-0003-2244-2320}
\affiliation{%
	\institution{University of Padova}
	\city{Padova}
	\country{Italy}
}
\email{fabio.vandin@unipd.it}

\author{Aristides Gionis}
\orcid{0000-0002-5211-112X}
\affiliation{%
	\institution{KTH Royal Institute of Technology}
	\city{Stockholm}
	\country{Sweden}
}
\email{argioni@kth.se}

\begin{abstract}
	Finding dense subnetworks, with density based on edges or more complex structures, 
	such as subgraphs or $k$-cliques,
	is a fundamental algorithmic problem with many applications.
	While the problem has been studied extensively in static networks,
	much remains to be explored for \emph{temporal networks}.
	
	In this work we introduce the novel problem of identifying the \emph{temporal motif densest subnetwork}, 
	i.e., the densest subnetwork with respect to \emph{temporal motifs}, 
	which are high-order patterns characterizing temporal networks.
	Identifying temporal motifs is an extremely challenging task, 
	and thus, efficient methods are required.
	To address this challenge, we design two novel randomized approximation algorithms 
	with rigorous probabilistic guarantees
	that provide high-quality solutions. 
	We perform extensive experiments showing that our methods outperform baselines. 
	Furthermore, our algorithms scale on networks with up to \emph{billions} of temporal edges, 
	while baselines cannot handle such large networks.
	We use our techniques to analyze a financial network
	and show that our formulation reveals important network structures, 
	such as bursty temporal events and communities of users with similar~interests.
\end{abstract}





\maketitle

\section{Introduction}

\begin{figure}[t]
	\centering
	\subfloat[]{
		\includegraphics[width=0.8\linewidth]{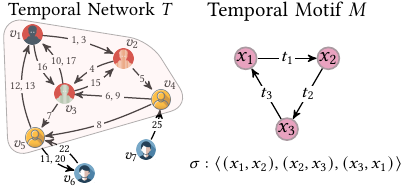}
		\label{subfig:TNandTM}
	}\\
	\subfloat[]{
		\includegraphics[width=0.9\linewidth]{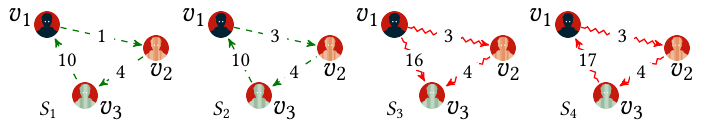}
		\label{subfig:deltaInsts}
	}
	\caption{\ref{subfig:TNandTM}:~Representation of a temporal network $T$ with $n=7$ vertices, $m=18$ edges (edge labels denote the timings, and commas denote multiple edges), and a temporal motif $M$ with its ordering $\sigma$ (i.e., $\sigma$ captures the temporal dynamics of the motif $M$). 
	\ref{subfig:deltaInsts}:~for $\delta=10$ only the green sequences ($S_1, S_2$) are $\delta$-instances of $M$ in $T$;
	the red sequences are not $\delta$-instances, 
	since $S_3$~cannot be mapped on the motif $M$ following $\sigma$, and~$S_4$ exceeds the timing constraint on $\delta$.}
	\label{fig:basicdef}
\end{figure}

Networks (or graphs) are fundamental abstractions for understanding and characterizing complex systems, 
such as social networks~\cite{Freeman2004}, 
biological systems \cite{Dorogovtsev2003}, and more~\cite{Lee2012,Huebsch2003}.
An important primitive in graph mining is the discovery of \emph{dense subgraphs}~\cite{Gionis2015}.
Dense-subgraphs find applications in areas such as visualization~\cite{Sariyuce2015}, 
anomaly detection~\cite{Chen2022}, finance~\cite{Boginski2003}, and social networks~\cite{Henry2007,Fazzone2022}.
The \emph{density} of a subnetwork is often defined as the total number of its \emph{edges} divided by the number of its vertices, and the \emph{densest subnetwork} is the one attaining maximum density. 
The definition of density can be extended to higher-order structures, 
such as $k$-\emph{cliques} or $k$-\emph{vertex motifs}~\cite{Tsourakakis2015,Fang2019,Sun2020},
leading to subnetworks that offer a more nuanced characterization of the data, 
e.g., more tightly-connected communities in social networks.

Many real-world networks are inherently \emph{temporal}~\cite{Holme2012}, 
containing information about the timing of the interactions~\cite{Masuda2016}.
Temporal networks enable the study of \emph{novel} phenomena not observable otherwise, 
such as bursty activities~\cite{Lambiotte2013}, dynamic reachability, 
temporal centrality~\cite{Santoro2022,Tang2010, Oettershagen2020}, and more~\cite{Holme2019, Oettershagen2023}.

\emph{Temporal motifs} are small-size subgraphs with 
edges having temporal information and a duration bound~\cite{Liu2019,Liu2021,Zhao2010,Kovanen2011}.
The temporal information on the edges captures the \emph{temporal dynamics} of the motif, for example,
information can spread on a temporal path $x\stackrel{t_1}\to z\stackrel{t_2}\to w$ only if $t_1 < t_2$. 
In addition, a duration constraint ensures that relevant events occur sufficiently close in time, 
e.g., $t_2-t_1$ is not too large. 
Temporal motifs are building blocks for temporal networks, 
finding novel applications, 
such as cryptocurrency network analysis~\cite{Wu2022,Arnold2024Motifs}, anomaly detection~\citep{Belth2020}, 
and more~\cite{Liu2023a, Lei2020}.

There has been growing interest in developing efficient 
algorithms to count and enumerate temporal motifs~\cite{Pashanasangi2021,Wang2020a,Sarpe2021,Sarpe2021a,Liu2019,Mackey2018,Cai2023TemporalButterfly,Pu2023TemporalButterfly}. 
Identifying temporal motifs is a challenging task, 
considerably harder than its counterpart in static networks. 
In fact, there are motifs that can be identified in polynomial time in static networks but their identification becomes \NPhard in temporal networks~\cite{Liu2019}. 
Therefore, 
little has been done in characterizing \emph{temporal subnetworks} with respect to temporal motif instances.

In this work we introduce and study the novel problem of discovering the densest temporal subnetwork with respect to a temporal motif, that is
given a large temporal network and a temporal motif, 
our goal is to find the \emph{temporal-motif densest subnetwork}, i.e., the subnetwork maximizing a density function that accounts for the number of temporal motif instances normalized by its size.

\vspace{1mm}
\para{Example.}
Consider a financial temporal network and a temporal triangle $M: x\stackrel{t_1}\to z\stackrel{t_2}\to w\stackrel{t_3}\to x$ as from Figure~\ref{subfig:TNandTM}, where $M$ captures potential money laundering activities occurring between users in a short duration~\cite{Wu2022,LiuFraud2024}.
Some users may be known to be involved in suspicious activities (e.g., they transfer large amount of money; shown in red in Fig.~\ref{subfig:TNandTM}) while others are not (e.g., no history is available for their transactions; shown in yellow or blue in Fig.~\ref{subfig:TNandTM}). By finding the marked densest sub\-network 
($\{v_1,\dots,v_5\}$) according to motif $M$ we can identify a tightly-connected subset of users potentially involved in suspicious activities. Interestingly, this subnetwork may contain new suspicious users (in yellow in Fig.~\ref{subfig:TNandTM}), not previously identified.
\vspace{1mm}

Other promising applications captured by our formulation are, for example, (i) identifying frequent destinations appearing together from user travel data, i.e., groups of attractions visited following a specified travel journey captured by a temporal motif~\cite{Lei2020}; and (ii) detecting \emph{dense} groups of users and items according to specific purchase sequences on e-commerce platforms, again captured trough temporal motifs~\cite{Pu2023TemporalButterfly, Cai2023TemporalButterfly}. In both scenarios, the \emph{temporal motif-subnetwork densest subnetwork} can provide unique insights, for (i) it can be used to  design novel public transport routes or travel passes; and for (ii) it can be used for personalized advertising or identifying users with common purchase habits.

Solving the problem we propose requires \emph{exact enumeration} of \emph{all} temporal motif instances in the input network, 
which is extremely inefficient to perform.
To address this issue we propose \suitename, 
a suite of randomized approximation \textsf{AL}gorithms for the temporal motif \textsf{DEN}sest subne\textsf{T}work probl\textsf{E}m.  \suitename contains two novel approximation algorithms leveraging a \emph{randomized sampling} procedure, which has been successfully employed for temporal motif counting~\cite{Liu2019,Sarpe2021,Sarpe2021a,Pu2023TemporalButterfly}. 
Our novel algorithms can leverage any state-of-the-art unbiased sampling algorithm for temporal motif counting. 
The first algorithm in \suitename peels the vertex set in \emph{batches} until it becomes empty (removing vertices participating in few temporal motif instances), \emph{estimating} for small $\varepsilon>0$
within $(1\pm \varepsilon)$ the different counts needed for peeling.
The second algorithm adopts a similar estimated \emph{batch-peeling} technique 
but only for a fixed number of steps. 
On the remaining vertices it applies a \emph{greedy-peeling} approach, 
where vertices that participate in few temporal motif instances are peeled one at a time.
By carefully combining estimates, our methods output a high-quality solution with controlled error probability.
Both our algorithms \emph{avoid exhaustive enumeration} of all temporal motifs in the input network.

We show that  
the problem we study cannot be solved by directly applying existing techniques tailored to static
high-order density formulations that ignore the information provided by temporal motifs~\cite{Tsourakakis2014, Wang2020, Tsourakakis2015, Fang2019, Mitzenmacher2015, Chen2023, Sun2020},
as this leads to poor quality solutions. 
We then \emph{embed temporal information} in a weighted set of subgraphs, 
which can be used to leverage existing techniques for the identification of $k$-vertex Motif Densest Subnetworks ($k$-MDS).
This results into three baselines (an exact methods and two approximate ones), 
extending previous ideas for static networks~\citep{Goldberg:CSD-84-171,Charikar2000,Bahmani2012,Epasto2015}. 

We finally show that the methods in \suitename find high-quality solutions on \emph{billion-edge temporal networks in few hundreds of seconds}, where all other baselines cannot terminate in hours.

Summarizing, our contributions are as follows:

	\vspace{1mm}
	\para{1.} We introduce the Temporal Motif Densest Subnetwork (\ourproblem) problem, asking for the densest subnetwork of a temporal network, where the density considers the number of instances of a temporal motif. The \ourproblem problem strongly differs from existing temporal cohesive subgraph formulations given its \emph{temporal motif}-based density formulation.
	We show that the \ourproblem problem is not captured by analogous formulations on static networks ($k$-MDS) that disregard temporal information. 
	We then develop three baselines extending ideas for $k$-MDS discovery embedding temporal information in a suitable weighted set of subgraphs,  unfortunately such algorithms do not scale on large data.
	
	\vspace{1mm}
	\para{2.}
	We design \suitename, a suite of two novel approximation algorithms based on randomized sampling and peeling techniques, 
	providing \emph{high-quality solutions} and \emph{scaling on massive networks}. 
	
	\vspace{1mm}
	\para{3.}
	We perform experiments on medium- and large-size temporal networks, with up to billions of edges, validating \suitename.
	We show that baselines are often inefficient on large datasets. 
	Our randomized algorithms consistently achieve speedups from $\times2$ to $\times 43$, saving hundreds of GBs of memory, and reporting high-quality solutions.
	We conclude with a case study on a financial network discovering interesting temporal motif dense subnetworks.

Details on application scenarios for \ourproblem are provided in App.~\ref{appsec:TMDSApplications}, and a summary of the notation is found in App.~\ref{appsec:notation}.
\ifextended
We report in the Appendix, detailed proofs of our results (App.~\ref{sec:missingProofs}), additional discussion on methods (App.~\ref{appsec:randAlg} and App.~\ref{appsec:baselines}) and experiments~(App.~\ref{appsec:experimentalResults}).
\fi
Our implementation is publicly available.\footnote{\url{https://github.com/iliesarpe/ALDENTE}.}


\section{Preliminaries}\label{sec:prelims}

A \emph{temporal network} is a pair $T=(V_T,E_T)$,
where $V_T=\{v_1,\dots,v_n\}$ is a set of vertices and
$E_T=\{(u,v,t) : u,v\in V_T, 
t\in \mathbb{R}^+ \}$
is a set of \emph{temporal edges}.
We let $n=|V_T|$ be the number of vertices
and $m=|E_T|$ be the number of edges.
Each edge $e=(u,v,t)\in E_T$
contains a time\-stamp~$t$, representing the time of its occurrence.%
\footnote{We will denote directed edges by ``$(\ldots)$'' and undirected edges by ``$\{\ldots\}$''.}
Without loss of generality we assume that the edges in $E_T$ are listed by increasing time\-stamp
and that time\-stamps are~distinct.\footnote{In practice priority bits are assigned to enforce a strict total ordering over $E_T$.}

Next we introduce \emph{temporal motifs}.
For concreteness,
we adopt the definition of temporal motifs proposed by~\citet{Paranjape2016} 
that is one of the most commonly used in practice, 
but our results can be extended to other definitions of temporal motifs,
e.g., the ones surveyed by~\citet{Liu2021}.\footnote{By suitably replacing some of the subroutines of the algorithms in \suitename.}
While we focus on unlabeled networks, our approach can be also extended to labeled networks.


\begin{definition}[\citet{Paranjape2016}]
	\label{def:k-l-temporal-motif}
	A \emph{$k$-vertex $\ell$-edge temporal motif} with $k,\ell \ge 2$ is a pair  
	$M= ({K},\sigma)$, where ${K}=(V_{K},E_{K})$ 
	is a directed and weakly-connected \emph{(multi-)graph},
	with $|V_{K}|=k$, $|E_{K}|=\ell$, and
	$\sigma$ is an ordering of the $\ell$ edges in~$E_{K}$.
\end{definition}

A $k$-vertex $\ell$-edge temporal motif is also identified by the 
\emph{sequence} $\langle(x_1,y_1),\dots,(x_\ell,y_\ell)\rangle$ 
of $\ell$ edges of $E_{K}$ \emph{ordered} according~to~$\sigma$.
An illustration of a temporal motif is shown in Figure~\ref{subfig:TNandTM}.

\begin{definition}[{$\delta$-instance of a temporal motif}]
	\label{def:deltainst}
	Given a temporal network $T$, a temporal motif~$M$,
	and a value $\delta\in \mathbb{R}^+$
	corresponding to a \emph{time-interval length}, a \emph{time-ordered} sequence
	$S=\langle(x_1',y_1',t_1'),\dots,(x_\ell',y_\ell',t_\ell')\rangle$
	of $\ell$ unique temporal edges of $T$ is a \emph{$\delta$-instance}
	of the temporal motif $M=\langle(x_1,y_1),\dots,\allowbreak(x_\ell,y_\ell)\rangle$ if:
		
	\vspace{1mm}
	\noindent
	(1) there exists a bijection $h$ from the vertices appearing in $S$
		to the vertices of $M$,
		with $h(x_i') = x_i $ and $h(y_i')=y_i$, and $i \in [\ell]$;
		
	\vspace{1mm}
	\noindent
	(2) the edges of $S$ occur within $\delta$ time; i.e., $t_\ell' -t_1' \le \delta$.
\end{definition}

The bijection implies that
edges in~$S$ are mapped onto edges in~$M$ in the \emph{same ordering} $\sigma$.
Additionally, a $\delta$-instance is not required to be vertex-induced. 
An example of Definition~\ref{def:deltainst} is shown in Fig.~\ref{subfig:deltaInsts}.

\vspace{1mm}
Let $\mathcal{S}_{T}=\{S : S \text{ is a } \delta\text{-instance of } M \text{ in } T\}$
be the \emph{set of $\delta$-instances of a motif $M$}
for a given time-interval length~$\delta$.
Computing the set $\mathcal{S}_T$ is extremely demanding,
first such a set can have size $\bigO\!\left({m\choose \ell}\right)$, and, in general even detecting a single $\delta$-instance of a temporal motif is an \NPhard problem~\cite{Liu2019}.

%
%
We define a \emph{weighting function} $\tau: \mathcal{S}_T\mapsto \mathbb{R}^+$
that assigns a non-negative weight to each $\delta$-instance $S\in\mathcal{S}_T$ 
based on its edges and vertices,
e.g., accounting for importance scores for vertices and/or edges, or available metadata.
Given $S=\langle(x_1,y_1,t_1),\dots,\allowbreak(x_\ell,y_\ell,t_\ell)\rangle\in \mathcal{S}_T$  
two possible realizations of weighting functions~are: 

\vspace{1mm}
\noindent
(i) a \emph{constant} function $\tau_{c}$: $\tau_{c}(S) =1$, and 

\vspace{1mm}
\noindent
(ii) a \emph{decaying} function $\tau_{d}$: $\tau_{d}(S) =\frac{1}{\ell-1} \sum_{j=2}^\ell \exp(-\lambda (t_j - t_{j-1})),$ 
	where $\lambda > 0$ controls the time-exponential decay.

In the above, $\tau_c$ is a simple weighting function where each $\delta$-instance receives a unitary weight, this can be used in most exploratory analyses where no prior information is known about the network. The decaying function $\tau_d$ is inspired from link-decay models, where the value of $\lambda$ corresponds to the inverse of the average inter-time distance of the timestamps in $E_T$. Therefore, the function $\tau_d$ accounts for decaying processes over networks, which are common in various domains, e.g., email communications, and spreading of diseases~\cite{Ahmed2021,Ahmad2021TieDec,Vazquez2007NonPois}.

\vspace{1mm}

As an example, considering the instance $S_1$ from Fig.~\ref{subfig:TNandTM}, 
then $\tau_c(S_1)=1$, 
while for $\lambda=0.1$ then $\tau_{d}(S_1) = \frac{1}{2} (e^{-3\lambda} + e^{-6\lambda}) \approx 0.64$.

\vspace{1mm}

Given $T=(V_T,E_T)$ 
and a subset of vertices $W\subseteq V_T$, 
we denote with $T[W]$ the temporal subnetwork \emph{induced} by $W$ as 
$T[W]=(W, E_T[W])$, where $E_T[W] = \{(u,v,t): u\in W,\allowbreak v\in W, \allowbreak  (u,v,t)\in E_T\}$ is the set edges from $E_T$ among the vertices in $W$. 
Given $W\subseteq V_T$, a temporal motif $M$, and $\delta\in\mathbb{R}^+$, 
for ease of notation, we denote with $\mathcal{S}_W$ the set of all $\delta$-instances 
of $M$ in $T[W]$. 
With these definitions at hand we can now assign \emph{weights}
to the set of $\delta$-instances of a temporal motif $M$ in a temporal subnetwork. 
Finally, given  $W\subseteq V_T$ and a weighting function $\tau$ we define the
the \emph{total weight} of the set of $\delta$-instances of motif $M$ in $T[W]$ as
$\mathcal{\tau}(W) = \sum_{S\in\mathcal{S}_W } \tau(S).$ 

\vspace{1mm}
As an example, for $T$ and $M$ as in Fig.~\ref{subfig:TNandTM}, 
for $W=\{v_1,v_2,v_3\}$,
the constant weight function $\tau_c$, and $\delta=20$,
we have $\tau(W)=4$. 

We are now ready to state the problem we address in this~paper.
\begin{problem}[Temporal Motif Densest Subnetwork (\ourproblem)]
	\label{probl:denseMotif}
	Given a temporal network $T=(V,E)$, a temporal motif~$M$, 
	a time-interval length $\delta\in \mathbb{R}^+$, and 
	a weighting function $\tau:\mathcal{S}_T\mapsto \mathbb{R}_{>0}$, 
	find a subset of vertices $W^*\subseteq V$ that maximizes the density function
	\[
	\fobj(W) = \frac{\tau(W)}{|W|}.
	\]
\end{problem}

The \ourproblem problem takes in input a temporal motif $M$, 
a weighting function $\tau$, 
and the time-interval length $\delta\in \mathbb{R}^+$. 
Differently from existing works in the literature of temporal-networks 
we are \emph{not} trying to find a time-window to maximize a static (or temporal)
density score. 
As an example,
	consider $T$ and $M$ from Fig.~\ref{subfig:TNandTM}, 
	and the \ourproblem\ problem with constant weighting function $\tau_{c}$.
	Fig.~\ref{subfig:TNandTM} shows the solution obtained with $\delta=10$, where
	the optimal solution is $W^*=\{v_1,v_2,v_3,v_4,v_5\}$ 
	with $\fobj(W^*)=\tau(W^*)/|W^*|= 6/5$.
	Note that the timestamps in $T[W^*]$ span an interval length greater than $\delta$ (i.e., $[1,17]$).

Throughout this paper we say that a vertex set $W\subseteq V$ 
achieves an $\alpha$-\emph{approximation ratio}, with $\alpha\le1$, 
if $\fobj(W)/\OPT \ge \alpha$, 
where $\OPT\mathbin = \fobj(W^*)$ is the value of an optimal solution to Problem~\ref{probl:denseMotif}. 
An algorithm that returns a solution with an $\alpha$-approximation ratio is 
an \emph{$\alpha$-approximation algorithm}.  

We conclude the preliminaries 
with defining the \emph{temporal motif degree} of vertex $v\in V_T$ 
by $\temporalDegreeWeight{V_T}{v}\mathbin = \sum_{S \in \mathcal{S}_{V_T}}(\tau(S) \mathbbm{1}[v\in S]$), where $\mathbbm{1}[v\in S]$ is an indicator function denoting if the vertex appears at least once over the edges in the $\delta$-instance $S$. As an example, if we consider $T$ and $M$ as in Fig.~\ref{subfig:TNandTM}, $\tau_c$, and $\delta=10$, then $\temporalDegreeWeight{V_T}{v_4} = 2$, as $v_4$ participates in two $\delta$-instances of $M$. 
\ifextended
We provide a notation table in Appendix~\ref{appsec:notation} and all the missing proofs are deferred to Appendix~\ref{sec:missingProofs}.
\fi


\section{Related work}
\label{sec:relworks}

\para{Densest subgraphs in static networks.} 
Densest-subgraph discovery (DSD) is a widely-studied problem 
asking to find a subset of vertices that maximizes \emph{edge density}~\cite{Gionis2015, Lee2010, Fang2022,Lanciano2023}.
\citet{Goldberg:CSD-84-171} proposed a polynomial-time algorithm for computing the exact solution 
using min-cut. 
Since min-cut algorithms are expensive, such an algorithm is not practical for large networks. 
\citet{Charikar2000} developed a faster greedy algorithm achieving a $\tfrac{1}{2}$-approximation ratio. 
The idea is to iteratively remove the vertex with the smallest degree, and return the vertex set 
that maximizes the edge density over all vertex sets considered. 
\citet{Bahmani2012} proposed methods trading accuracy for efficiency, 
including a $\tfrac{1}{2(1+\xi)}$-approximation algorithm, controlled by a parameter $\xi>0$.


Most related to \ourproblem
are the problems of $k$-clique, or $k$-motif densest subgraphs~\cite{Tsourakakis2014, Wang2020, Tsourakakis2015, Fang2019, Mitzenmacher2015, Chen2023, Sun2020, He2023ScalingKClique}.
The $k$-clique densest-subgraph problem 
modifies the standard density definition by replacing edges with the number of $k$-cliques induced by a vertex set.
This problem is \NPhard.
\citet{Tsourakakis2014} developed exact and approximate algorithms 
(extending Goldberg's and Charikar's algorithm, respectively) 
for the $k$-clique and triangle-densest subgraph problem ($k=3$), further studied by~\citet{Wang2020}.
High-order core decomposition has been used to approximate the $k$-clique and the $k$-motif densest-subgraph problems,  
where density is defined on $k$-vertex subgraphs~\cite{Fang2019}. 
The techniques proposed by~\citet{Fang2019} rely on exact enumeration of $k$-vertex subgraphs, 
which in the case of temporal networks is impractical, especially on large data.

For the $k$-clique densest-subgraph problem, 
\citet{Mitzenmacher2015} applied sparsification, 
which requires exhaustive enumeration of all $k$-cliques. 
\citet{Sun2020} developed convex optimization-based algorithms and a related sampling approach requiring a subroutine that performs exact enumeration of $k$-cliques. 
For \ourproblem, we cannot afford listing the set $\mathcal{S}_T$ 
as such {expensive} procedure has exponential computational cost in general.
In our work, we develop two randomized approximation algorithms that \emph{avoid} exhaustive enumeration of 
$\mathcal{S}_T$, significantly differing from existing approaches.


\vspace{1mm}
\para{Densest-subgraph discovery in temporal networks.} 
We do not discuss methods for DSD on temporal networks defined over snapshots of static graphs.
Such a model is used for representing more coarse-grained temporal data than 
the fine-grained model we adopt, and it is less powerful~\cite{Holme2012, Gionis2024Tutorial}. 

Several cohesive subgraph models have been studied~\cite{Gionis2024Tutorial}, such as,
periodic densest subgraphs~\cite{Qin2023}, 
bursty communities~\cite{Qin2022},
densest subgraphs across snapshots~\cite{Semertzidis2018}, 
correlated subgraphs~\cite{Preti2021}, and subgraphs with specific core properties~\cite{Zhong2024}.
Our work differs significantly from all these definitions, as it is based on temporal motifs.
Several problems have been proposed to identify the time frames maximizing specific density scores~\cite{Rozenshtein2019} and extensions of core-numbers to the temporal scenario~\cite{Galimberti2020, Li2018,Oettershagen2023Core} or temporal quasi-cliques~\cite{Lin2022}.  
Again, our problem strongly differs from all these works
\ifextended
we provide extensive details in Appendix~\ref{appsubsec:existingProblems}.  
\fi
We also note that, for brevity, we do not discuss algorithms for \emph{dynamic graphs}, 
i.e., static graphs where edges are added or removed over time. 
A large literature exists for this model~\cite{Epasto2015,Bhattacharya2015, Hanauer2021}, 
which, however, differs significantly from the one we consider.

\vspace{1mm}
\para{Algorithms for temporal motifs.} 
We briefly discuss relevant algorithms for finding temporal motifs, as per Def.~\ref{def:k-l-temporal-motif}.
\citet{Paranjape2016}, who introduced such a definition, developed dynamic-programming algorithms tailored to specific $3$-edge and $2$- or $3$-vertex temporal motifs. 
Subsequently, \citet{Mackey2018} developed a backtracking algorithm to enumerate all $\delta$-instances of an arbitrary temporal motif;
the complexity of such algorithm is $\bigO(m^\ell)$, i.e., exponential and not practical. 
Many other exact-counting algorithms exist, usually tailored to specific classes of temporal motifs such as triangles~\cite{Pashanasangi2021} or small motifs~\cite{Gao2022,Cai2023TemporalButterfly}. 
Since exact counting is often impractical, many sampling-based algorithms have been developed~\cite{Liu2019,Sarpe2021,Sarpe2021a,Wang2020a,Pu2023TemporalButterfly}.
In our work, when needed, we will use the state-of-the-art randomized algorithm 
by \citet{Sarpe2021} for estimating the motif count of an arbitrary temporal motif, 
which provides $(\varepsilon, \eta)$-approximation guarantees, i.e., 
it offers an unbiased estimate~$\hat{\tau}(V_T)$ for which 
$\mathbb{P}\left[|\hat{\tau}(V_T)- \tau(V_T)| \le \varepsilon \tau(V_T)\right] > 1-\eta$ when  the constant weighting function is considered, which we extend to arbitrary non-negative weighting functions $\tau$.


\section{Randomized algorithms}
\label{sec:algorithms}

We  now introduce the novel randomized approximation algorithms with tight probabilistic guarantees that we develop in \suitename.  The algorithms in \suitename will address the \ourproblem avoiding the expensive exhaustive enumeration of the set $\mathcal{S}_T$, trough which the \ourproblem problem can be solved (see discussion in Section~\ref{sec:temporalAndBaselines}).


\subsection{\suiterand}\label{subsec:probpeelalg}
The main idea behind \algrand\ is to avoid exhaustive enumeration of $\mathcal{S}_T$ by using \emph{random sampling} and peeling  vertices in ``batches'' according to an \emph{estimate} of their temporal-motif degree. 
\algrand carefully combines highly-accurate probabilistic (1$\pm \varepsilon, \varepsilon>0$) estimates of the temporal motif degrees of all vertices~(that are obtained by suitably leveraging existing algorithms for global temporal motif estimation~\cite{Sarpe2021,Wang2020a}),
and peels vertices in batches at each step of such a schema, similar \emph{non}-approximate batch-peeling techniques have been employed successfully in other settings~\cite{Bahmani2012, Epasto2015,Sun2020FDHyp}.


We briefly introduce some definitions. 
Let $T'=(V',E')\subseteq T=(V_T,E_T)$, and recall that 
$\mathcal{S}_{T'} = \{S : S \text{ is a } \delta\text{-instance of } M \text{ in } T'\}$ is the set of 
$\delta$-instances from $\mathcal{S}_T$ computed only on the sub\-network~$T'$, 
where~$T'$ is \emph{not} required to the sub\-network induced  by $V'$.
For a given vertex $v\in V_T$, extending our previous notation, let 
$\temporalDegreeWeightNetwork{T'}{v} \mathbin =  \sum_{S\in \mathcal{S}_{T'}} \tau(S)$
be the temporal motif degree of vertex $v$ 
in the sub\-network $T'$. 
Clearly, $\sum_{v\in V_T} \temporalDegreeWeightNetwork{T}{v} = k \tau(V_T)$.

We assume a sampling algorithm $\mathcal{A}$, which
outputs a sample $T'=(V',E')\subseteq T[W]$
of a sub\-network $T[W]$.
Given a subnetwork $T[W]\subseteq T$ we will execute the sampling algorithm $\mathcal{A}$ 
to obtain $r\in \mathbb{N}$ i.i.d.\ subnetworks 
$\mathcal{D}=\{{T}_1, \dots, {T}_r : {T}_i =\mathcal{A}(T[W]) \subseteq T[W], i\in[r] \}$.
We use the subnetworks in $\mathcal{D}$ to compute an \emph{unbiased} estimator 
$\estimateTemporalDegreeWeight{W}{w}$ of the temporal-motif degree for $w \in W$, 
and for the correct choice of $r$ it holds
$\estimateTemporalDegreeWeight{W}{w} \in (1\pm \varepsilon)\temporalDegreeWeight{W}{w}$, 
and combine these estimates to obtain an estimator $\hat{\tau}(W) \in(1\pm \varepsilon) \tau(W)$. 
\ifextended
More details on the estimators are given in Appendix~\ref{appsubsubsec:buildingEstimate}.
\fi

\begin{algorithm}[t]
	\caption{\suiterand}\label{alg:kepsilonapproxsampling}
	\KwIn{$T=(V_T,E_T), M , \delta\in \mathbb{R}^+, \tau, \xi>0, \varepsilon >0, \eta \in (0,1)$.}
	\KwOut{$W$.}
	$W_1 \gets V_T$; $i\gets 1$\label{algsampl:inititeration}\;
	\While{$W_i \neq \emptyset$\label{algsampl:mainLoop}}{
		$r_i \gets \mathtt{GetBound}(W_{i}, \varepsilon, \eta/2^{i})$\label{algsampl:boundSamples}\;
		$\mathcal{D}=\{T_1,\dots,T_{r_i}\} \gets \mathcal{A}^{r_i}(T[W_i])$\label{algsampl:getSamples}\;
		$\estimateTemporalDegreeWeight{W_i}{w_1},\dots,\estimateTemporalDegreeWeight{W_i}{w_{|W_i|}}\gets \mathtt{GetEstimates}(W_i, \mathcal{D})$\label{algsampl:computeEstimates}\;
		$\hat{\tau}(W_i)\gets 1/k \sum_{w\in W_i} \estimateTemporalDegreeWeight{W_i}{w}$\label{algsampl:estimateTotalCount}\;
		$L(W_i)\gets\left\{w\in W_i : \estimateTemporalDegreeWeight{W_i}{w}\le k(1+\xi) \hat{\tau}(W_i)/|W_i|\right\}$\label{algsampl:peelNodes}\;
		$W_{i+1} \gets W_i\setminus L(W_i)$; $i\gets i+1$\label{algsampl:setNextIt}\;
	}
	\KwRet{$W=\argmax_{j=1,\dots,i}\{\hat{\tau}(W_j)/|W_j|\}$;\label{algsampl:return}}
\end{algorithm}

\vspace{1mm}
\para{Algorithm description.}  We now present \suiterand\ (Algorithm~\ref{alg:kepsilonapproxsampling}), 
our probabilistic batch-peeling algorithm. 
The algorithm has three parameters, 
$\varepsilon>0$, $\eta\in(0,1)$ and $\xi>0$. 
The parameter~$\varepsilon$ controls the relative $\varepsilon$-approximation ratio of $\hat{\tau}(W)$ 
and $\estimateTemporalDegreeWeight{W}{w}$,
and $\eta$ controls the failure probability of the algorithm, 
i.e., the guarantees of \suiterand hold with probability $>1-\eta$. 
Last, $\xi>0$ controls the number of vertices to be peeled at each iteration, with higher values of $\xi$ more vertices are peeled, so \algrand performs less iterations, at the price of lower quality solutions (line~\ref{algsampl:mainLoop}).

First, \algrand instantiates the  vertex set to be peeled (line~\ref{algsampl:inititeration}).
Then it starts the main loop that ends when the vertex set $V_T$ becomes empty (line~\ref{algsampl:mainLoop}), 
and at each iteration $i>0$ it computes a bound $r_i$ ensuring that with $r_i$ it holds $\estimateTemporalDegreeWeight{W_i}{w} \in (1\pm\varepsilon) \temporalDegreeWeight{W_i}{w}$ with probability $>1-\eta/2^i$ \emph{simultaneously} for each $w\in W_i$, through the function $\mathtt{GetBound}$ (line~\ref{algsampl:boundSamples}).
It then samples~$r_i$ subnetworks from $T_{[W_i]}$ using the sampling algorithm $\mathcal{A}$, 
denoted by $ \mathcal{A}^{r_i}(T[W_i])$, the samples are stored in $\mathcal{D}$ (line~\ref{algsampl:getSamples}). 
Then it computes the estimates $\estimateTemporalDegreeWeight{W_i}{w}$ for each $w\in W_i$ (line~\ref{algsampl:computeEstimates}).
The algorithm proceeds by peeling vertices in batches according to a threshold that combines the estimator $\hat{\tau}(W_i)$ (computed in line~\ref{algsampl:estimateTotalCount}) and $\estimateTemporalDegreeWeight{W_i}{w}, w\in W_i$ and $\xi>0$ (line~\ref{algsampl:peelNodes}). 
\algrand then updates the vertex set and counter for the next iteration (line~\ref{algsampl:setNextIt}). 
Finally, the algorithm returns the vertex set maximizing the estimated objective function $\hat{\tau}(W_j)/|W_j|$ 
over all the iterations $j=1,\dots,\bigO(\log_{1+\xi}(n))$ (line~\ref{algsampl:return}).

\vspace{1mm}
\para{Building the estimator.} We briefly describe how to compute the estimates of the temporal motif degrees $\temporalDegreeWeightNetwork{W}{w}$ of vertices $w\in W$ trough the function $\mathtt{GetEstimates}$. The general schema we introduce captures many existing estimators, i.e.,~\cite{Sarpe2021, Liu2019, Wang2020}.  Invoking $\mathcal{A}$
on a sub\-network $T[W]$, $W\subseteq V_T$, 
we get a sample $T'=(V',E')\subseteq T[W]$, therefore using such a sample $T'$ we compute an estimate of the degree $\temporalDegreeWeight{W}{w}$,
for each $w\in W$ as 
\begin{equation}
	\estimateTemporalDegreeWeightPrime{W}{w}{T'} = \sum\nolimits_{S \in \mathcal{S}_{{T}'}:w\in S} \tau(S)/p_S,\label{eq:prestoEstH}
\end{equation}
where $p_S$ is the probability of observing $S\in \mathcal{S}_{{T}'}$ 
in the sampled sub\-network ${T}'=({V}', {E}')$.
\begin{lemma}
	For any $w\in W$, the count $\estimateTemporalDegreeWeightPrime{W}{w}{T'}$ computed on a sampled subnetwork $T'\subseteq T[W]$ is an unbiased estimate of $\temporalDegreeWeight{W}{w}$.
\end{lemma}
The final estimators $\estimateTemporalDegreeWeight{W}{w}$ for $w \in W$ are then computed as the sample average over all the $r$ sampled subnetworks (in line~\ref{algsampl:getSamples}) of Eq.~\eqref{eq:prestoEstH}, that is $\estimateTemporalDegreeWeight{W}{w} \mathbin=  1/r\sum_{i=1}^r \estimateTemporalDegreeWeightPrime{W}{w}{T_i}$ for all $w \in W$.
Note also that $\hat{\tau}(W)$ (obtained in line~\ref{algsampl:estimateTotalCount}) is an unbiased estimator of $\tau(W)$.

\vspace{1mm}
\para{Bounding the sample size.} Algorithm~\ref{alg:kepsilonapproxsampling} can employ any unbiased sampling algorithm $\mathcal{A}:T\mapsto 2^T$ as subroutine~\cite{Liu2019,Wang2020, Sarpe2021}. We used \texttt{PRESTO-A}~\cite{Sarpe2021} for which the authors provide bounds on the number of samples $r$ for event ``$\hat{\tau}(W_i)\in (1\pm \varepsilon)\tau(W_i)$'' to hold with arbitrary probability $>1-\eta$ for $\eta\in(0,1)$. In our function $\mathtt{GetBound}$ we need a stronger result as we need $(1\pm \varepsilon)$-approximation for \emph{all} temporal motif degree of vertices $w\in W_i$. Hence we bound individually, for each vertex, the probability that $\estimateTemporalDegreeWeight{W_i}{w} \notin (1\pm \varepsilon)\temporalDegreeWeight{W_i}{w}$ and combine such result with a union bound over all vertices. 
\ifextended
(more details are in Appendix~\ref{appsubsubsec:embeddingSamplingAlg}).
\fi

\vspace{1mm}
\para{Guarantees.} The next result establishes the quality guarantees of \suiterand.
\begin{theorem}\label{lemma:ratioSampling}
With probability at least $1-\eta$, \suiterand achieves a 
$\frac{(1-\varepsilon)^2}{k(1+\xi)(1+\varepsilon)^2}$-approximation ratio for the \ourproblem~problem. 
\end{theorem}
Note that if the parameter $\varepsilon$ 
is small (e.g., $\varepsilon \!\approx\! 10^{-3}$), 
then \suiterand outputs a solution with approximation ratio close to~$\tfrac{1}{k(1+\xi)}$.
\ifextended
In Appendix~\ref{appsubsubsec:embeddingSamplingAlg} we discuss in details how to leverage existing state-of-the-art sampling algorithms, and their guarantees, to obtain a $(1\pm \varepsilon)$-approximation of the temporal motif degrees to be used by \suiterand (in line~\ref{algsampl:boundSamples}), where we also report all the missing details of the functions \texttt{GetBound} and \texttt{GetEstimates}.
\else
\fi

\vspace{1mm}
\para{Time complexity.}
\suiterand does not require an \emph{exhaustive} enumeration of all $\delta$-instances in the network $T$. In fact, it only estimates the temporal motif degrees of the vertices in $W_i$ at each iteration $i$ of the algorithm. The running time depends on the number of samples $r_i$ collected at each iteration and the time required to process each sample, which is $\bigO(\hat{m}^{\ell})$  where $\hat{m}$ denotes the maximum number of edges of $E_T$ in a time-window of length $\delta$. 

Thus, the time complexity is
$\bigO\!\left(r_{\max}\left(\hat{m}^{\ell} \log_{1+\xi}(n) + \frac{(1+\xi) n}{\xi} \right)\right)$,
where 
$r_{\max}\mathbin=  \max_i \{r_i\}$, which depends on $\varepsilon$ and $\eta$. 
\ifextended
Such result comes from an analysis similar to the one we perform for a baseline in Appendix~\ref{subsec:batch}.
\else
\fi
Since the algorithm does not require exhaustive enumeration of $\mathcal{S}_T$ we expect it to be particularly efficient. In addition, since we only store the estimates of the various vertices at each iteration we expect also the algorithm to be memory efficient.

\subsection{\suitehybrid}\label{subsec:hybridpeel}
In this section we develop a novel probabilistic hybrid-peeling algorithm, 
denoted by~\suitehybrid, 
which combines the randomized batch-peeling technique of \algrand (to overcome full enumeration of $\mathcal{S}_T$)
with a more refined peeling approach that peels vertices one at a time, to obtain higher quality solutions (as in \cite{Charikar2000}). 

\vspace{1mm}
\para{Algorithm description.} The algorithm works as follows: 
(i) given temporal network $T=(V,E)$, it starts by removing nodes in a similar fashion to~\algrand. 
That is, at each iteration it computes a $(1\pm\varepsilon)$-approximation of the temporal motif degree of each vertex,
and vertices are removed in ``batches'' according to a threshold controlled by $\xi>0$.
Instead of iterating the process until the vertex set becomes empty, 
the algorithm executes the randomized peeling phase for a fixed number of iterations $J\ge 1$ provided in input. 
After such $J$ iterations,
\alghybrid considers the vertex set $W_{J+1}$ and its induced temporal network $T[W_{J+1}]$. 
Then it enumerates all temporal motifs on $T[W_{J+1}]$ (computing $\mathcal{S}_{W_{J+1}}$) and 
peels $W_{J+1}$ one vertex at a time according to the minimum temporal motif degree, storing the vertex set maximizing the objective function $\fobj(\cdot)$, among the ones observed. 
Hence, this second procedure returns a vertex set $W'$ maximizing $\fobj(\cdot)$ among the $\bigO(|W_{J+1}| - k)$ vertex sets explored, and \alghybrid then returns $W \mathbin=  \argmax\{\hat{\tau}(W_1)/|W_1|, \dots, \hat{\tau}(W_J)/|W_J|, \tau(W')/|W'|\}$. 
\ifextended
The pseudocode is reported in Appendix~\ref{appsubsec:pseudocodehybrid}.
\fi

The design of \alghybrid is motivated by a drawback of \algrand: 
even for very accurate estimates of temporal motif degrees it may provide a sub-optimal solution to \ourproblem, since even for $\varepsilon$ close to 0, \algrand only converges to the approximation ratio of $\frac{1}{k(1+\xi)}$, while the peeling procedure we apply on $W_{J+1}$ provides denser solutions (close to $1/k$).
The intuition behind \alghybrid is that, the first $J$ iterations are used to prune the network sufficiently, for the algorithm to maintain in $W_{J+1}$ a dense subnetwork $W'\subseteq W_{J+1}$ attaining at least a $1/k$-approximation ratio. Therefore: (i) the algorithm is scalable, as it avoids the computation of $\mathcal{S}_T$, and (ii) is expected to output high-quality solutions. 

Next we show the approximation ratio of \alghybrid.
\begin{theorem}\label{lemma:ratioHybrid}
	With probability at least $1-\eta$, \suitehybrid is a 
	$\frac{(1-\varepsilon)^2}{k(1+\xi)(1+\varepsilon)^2}$-approximation algorithm for the \ourproblem~problem when $J>1$,
	while for $J=1$ it achieves a $\frac{(1-\varepsilon)^2}{k(1+\xi)(1+\varepsilon)}$-approximation ratio.
\end{theorem}

\vspace{1mm}
\para{Time complexity.}
The time complexity of the algorithm can be bounded by $\bigO({RT}_{\mathtt{ProbPeel}-J}(T) + {RT}_{\mathtt{Greedy}}(T[W_{J+1}]))$, i.e., the sum of the runtime ($RT$) of executing \algrand for $J$ iterations, and then the peeling procedure that we denote with \algkapp on $T[W_{J+1}]$. 
\ifextended
The analysis follows from~\algrand and our analysis of a baselines in Appendix~\ref{subsec:greedy}.
\fi
By noting that there exists a value $\gamma> 1$ such that  $|W_{J+1}| = \bigO(n^{1/\gamma})$ the complexity is bounded by
$\bigO\!(J r_{\max} \hat{m}^{\ell} + \gamma^{-1} k^3 \,  n^{k/\gamma} \log(n))$.
This highlights that first peeling the network with the randomized batch peeling technique enables \suitehybrid to work on a significantly small network, over which to execute a more refined peeling algorithm, affording an exact enumeration of temporal motif instances on $T[W_{J+1}]$, leading to an overall scalable and practical algorithm.

\begin{table*}[t]
	\centering
	\caption[Table on algorithms]{Baselines and algorithms in \suitename. For each algorithm we report: its name, the approximation ratio, its parameters, the guaranteed probability (``1'' means deterministic), ``Avoid Enum.'': if it avoids computing the set $\mathcal{H}$ (and therefore $\mathcal{S}_T$), and its time complexity. The first three algorithms are baselines, while the latter two algorithms are part of \suitename (see Section~\ref{sec:algorithms}).}
	\label{tab:ALDENTEalgs}
	\scalebox{0.85}{
		\begin{tabular}{lrrrcr}
			\toprule
			Name & Approximation & Parameters & Probability & Avoid Enum. & Time Complexity\\
			\midrule
			\suiteex & 1 & - & 1 &  \textcolor{red}{\xmark} &$\bigO\left(m\hat{m}^{\ell-1} +kn^{2k} \log(n\tau(V))\right)$\\
			\suitegreedy & $\frac{1}{k}$ & - & 1 &   \textcolor{red}{\xmark}&$\bigO\!\left(m\hat{m}^{(\ell-1)} + k^3n^k\log(n)\right)$\\
			\suitebatch & $\frac{1}{k(1+\xi)}$ & $\xi > 0$ & 1 &   \textcolor{red}{\xmark} &$\bigO\!\left(m\hat{m}^{(\ell-1)} + k^3n^k\log(n)\right)$\\
			\midrule
			\suiterand & $\frac{(1-\varepsilon)^2}{k(1+\xi)(1+\varepsilon)^2}$& $\xi>0, \varepsilon, \eta \in (0,1)$ & $>1-\eta$ & \textcolor{ForestGreen}{\vmark} &$\bigO\!\left(r(\varepsilon, \eta) \left(\hat{m}^{\ell} \log_{1+\xi}(n) + \frac{(1+\xi) n}{\xi} \right)\right)$ \\
			\suitehybrid & $\frac{(1-\varepsilon)^2}{k(1+\xi)(1+\varepsilon)^2}$& $\xi>0, \varepsilon, \eta \in (0,1), J>0$ & $>1-\eta$ &\textcolor{ForestGreen}{\vmark} & $\bigO\!\left(J r(\varepsilon, \eta) \hat{m}^{\ell} +  \gamma^{-1} k^3 \,  n^{k/\gamma} \log(n)\right)$\\
			\bottomrule
		\end{tabular}
	}
\end{table*}

\section{Temporal information impact and  baselines}\label{sec:temporalAndBaselines}

 A natural question is if
the temporal information in the temporal network $T$ is necessary, 
or if the directed static network \emph{associated to $T$} already captures the \ourproblem problem formulation.
Given a temporal network $T=(V_T,E_T)$ 
the \emph{associated static network} of $T$ is $G_{T}= (V_T,E_G)$, 
where $E_G=\{\{u,v\} \mid \text{there exists } (u,v,t) \in E_T \text{ or } (v,u,t) \in E_T\}$. 
If we keep the edge directions, a static directed network is denoted by $D_{T}$.\footnote{The static network of a temporal network $T$ simply ignores the timestamps of $T$, collapsing multiple temporal edges on the same static edge.}
For a subset of vertices $W\subseteq V$ we denote by $G_T[W]$ its associated static network.
 
We first show that even for a very simple temporal motif,
the optimal solution on $D_{T}$ can be arbitrarily bad 
when evaluated on the temporal network $T$.
This highlights that Problem~\ref{probl:denseMotif} 
\emph{cannot} be addressed by existing algorithms for $k$-MDS on static networks or aggregations of the input temporal network (i.e., disregarding temporal information).
We then show how to embed temporal information in a weighted set of subgraphs to solve the~\ourproblem problem, to obtain our baselines. 
Computing such a set requires \emph{full enumeration} of temporal motif instances on $T$, which we recall to be extremely demanding, especially on large temporal networks.

\vspace{1mm}
\para{Temporal vs.\ static.}
We start by some definitions. 
For a temporal motif $M=({K},\sigma)$, we use $G[M]$ to denote the
undirected graph associated to ${K}$, i.e.,
$G[M]=(V_{{K}}, \{\{u,v\} : \text{there exists } \allowbreak (u,v)\allowbreak 
\text{or } (v,u) \text{ in } E_{K}\})$, ignoring directions and multiple edges in $K$.
We say that
a temporal motif $M$ is a \emph{2-path} if  $M=\langle(u,v),(v,w)\rangle$, 
with $u\neq v\neq w$.
Given a directed static graph $D_T=(V_T, E_D)$ we define the number of \emph{static 2-paths} 
(i.e., directed paths of length 2) induced by a subset of vertices $H\subseteq V_T$ as $|P_2(H)|$. 
Therefore, a static version of the~\ourproblem problem on temporal networks with $M$ being a 2-path, is to consider the aggregated network of $T$ (i.e., $D_T$) with the goal of identifying $H^*\subseteq V_T$ maximizing $\fobj_2(H) \mathbin=  \frac{|P _2(H)|}{|H|}$.
We refer to this problem as S2DS (Static 2-path Densest Subnetwork).

We can show that an optimal solution to S2DS, 
computed on~$D_T$, the static network  associated to~$T$, 
can be arbitrarily bad when evaluated for~\ourproblem on~$T$ with $M$ being a 2-path.
Without loss of generality, we assume the constant weighting function $\tau_{c}$.

\begin{lemma}\label{lemma:staticvstemporal}
	Given a temporal network $T=(V_T,E_T)$, and its associated directed network $D_{T}$, let $H^*$ be a solution to S2DS on~$D_{T}$. 
	Let $W^*$ on~$T$ be the optimal solution of~\ourproblem, for a 2-path motif, with fixed $\delta\ge 1$.
	Then there exists a temporal network $T$ such that $\fobj(H^*)=0$ while $\fobj(W^*) = \bigO(n^2)$.
	Furthermore, the two solution sets~$H^*$ and $W^*$ are~disjoint.
\end{lemma}

\vspace{1mm}
\para{Embedding temporal information.}
Since we cannot simply disregard temporal information, 
we investigate how to compute a suitable weighted set of static subgraphs embedding temporal information. 
Such a set can be used to solve the~\ourproblem problem by leveraging existing techniques.
Given the input to \ourproblem, we define $\mathcal{H} \mathbin=  \{H :H=(V_H, E_H) \subseteq G_{T}, \text{ for some } H' \subseteq H \text{ it holds } H' \!\sim\! G[M],~\tau(V_H)>0, \text{ and } |V_H|=k\}$  be the set of 
$k$-connected induced subgraphs ($k$-CISs) from $G_{T}$,%
\footnote{Symbol ``$\sim$'' denotes standard graph isomorphism.} 
where each $k$-CIS, $H\in \mathcal{H}$, 
encodes a subgraph containing at least a $\delta$-instance $S\in \mathcal{S}_T$ with $\tau(S)>0$,
and $H$ is induced in~$G_{T}$,
i.e., has all the edges among the vertices~$V_H$. 
The above conditions are ensured by requiring the existence of $H'\subseteq H$ 
such that $H'\mathbin\sim\mathbin G[M]$ and $\tau(V_H)>0$, 
where $G[M]$ is the undirected graph associated to the temporal motif~$M$. 
For example, considering the motif $M$, and the temporal network $T$ from Fig.~\ref{fig:basicdef}, and $\delta=10$, 
then the corresponding set $\mathcal{H}$ is $\mathcal{H} = \{(H_1 =  G_T[\{v_1,v_2,v_3\}],~H_2 =  G_T[\{v_2,v_3,v_4\}],~H_3 =  G_T[\{v_1,v_3,v_5\}],~H_4 =  G_T[\{v_3,v_4,v_5\}]  \} $.

Next, we need to define the weight of each subgraph $H=(V_H,E_H)$ in  $\mathcal{H}$. 
This is done by $\tau(V_H) = \sum_{S\in \mathcal{S}_{V_H}} \tau(S)$, 
i.e., the weight of the subgraph $H$ is the sum of the weights of the $\delta$-instances 
that occur among the nodes of $V_H$ in $T$.
As an example, consider $H_1 = (V_{H_1}=\{v_1,v_2,v_3\}, \{ \{v_1,v_2\},\{v_2,v_3\},$ $ \{v_3,v_1\} \}) \in \mathcal{H}$ from Fig.~\ref{fig:basicdef}, then $\tau(V_{H_1}) = 2$ under weight $\tau_c$, and $\delta=10$, as there are two $\delta$-instances among such nodes.
Note that such a construction is in accordance with Lemma~\ref{lemma:staticvstemporal}, 
as to build, and weight the set $\mathcal{H}$ we are exploiting full information about the temporal motif $\delta$-instances of $M$ in~$T$. 
In Section~\ref{appsubsec:baselineOverview} we provide a summary on how to leverage the set $\mathcal{H}$ to adapt existing techniques for high-order subgraph discovery to solve \ourproblem, 
\ifextended
while a more detailed description can be found in Appendix~\ref{appsec:baselines}.
\fi
Such algorithms will be used as baselines for comparison against \suitename in our experimental evaluation. 
All resulting baselines, and the algorithms in \suitename that we developed are finally summarized in Table~\ref{tab:ALDENTEalgs}.
	
	

\section{Experimental evaluation}\label{sec:experiments}

In this section we evaluate the algorithms in \suitename against the baselines. We  describe the experimental setup in Section~\ref{subsec:setup}, 
and we compare the solution quality, and runtime of all the algorithms in Section~\ref{subsec:appxration}. 
Finally, in Section~\ref{subsec:casestudy} we conduct a case study 
on a real-world transaction network from the Venmo platform to support the usefulness of solving the \ourproblem problem. 
\ifextended
We defer to appendices for additional results, e.g., on parameter sensitivity (in Section~\ref{subsec:sensitivity}), memory usage (in Section~\ref{subsubsec:memory}) and results under decaying weighting function $\tau_{d}$ (in Section~\ref{appsubsec:expdecay}). 
\fi

\subsection{Baseline overview}\label{appsubsec:baselineOverview}

We give a brief, description of our baselines, that leverage the construction of the set $\mathcal{H}$ (as from Section~\ref{sec:temporalAndBaselines}) and that are based on known ideas in the field, 
\ifextended
see and Appendix~\ref{appsec:baselines} 
\fi
for more details.

\suiteex: Exact algorithm embedding the set $\mathcal{H}$ is a properly weighted flow network, adapting ideas from~\cite{Goldberg:CSD-84-171,Mitzenmacher2015,Fang2019}. It computes multiple min-cut solutions on the flow network, the algorithm identifies an optimal solution to the~\ourproblem problem.

\suitegreedy: A $1/k$-approximation greedy algorithm that extend the ideas in~\cite{Charikar2000,Fang2019}. The algorithm performs $\bigO(n-k)$ iterations where at each iteration a vertex with minimum temporal motif degree is removed and the set $\mathcal{H}$ is updated accordingly. The algorithm returns the vertex set with maximum density observed.

\suitebatch: A $1/(k(1+\xi)), \xi>0$ greedy approximation algorithm extending ideas from~\cite{Bahmani2012,Epasto2015,Tsourakakis2015}. At each step the algorithm removes vertices with small temporal motif degree in ``batches'' and updates the set $\mathcal{H}$ returning the maximum density vertex set observed, the overall number of iterations is bounded by $\bigO(\log n/\xi)$.

All such baselines require the computation of the set $\mathcal{H}$, and therefore $\mathcal{S}_T$.\footnote{This holds for most of other existing techniques even based on sampling, see Section~\ref{sec:relworks}.} Unfortunately, this is extremely inefficient and does not scale on large data, as we will show next.

\begin{table}[t]
	\centering
	\caption{Networks used in our experiments. $n$: vertices, $m$: temporal edges, $|E_G|$: edges in the undirected static network, time-interval length $\delta$: value used for counting $\delta$-instances.}
	\label{tab:datasets}
	\scalebox{0.7}{
		\begin{tabular}{clrrrrrr}
			\toprule
			&Network& $n$&$m$& $|E_G|$ & Precision & Timespan & $\delta$\\
			\midrule
			\multirow{4}{*}{\rotatebox[origin=c]{90}{Medium}} & Sms & 44\,K & 545\,K & 52\,K & sec & 338 (days) &172.8\,K\\
			& Facebook & 45.8\,K & 856\,K & 183\,K & sec& 1561 (days)&86.4\,K\\
			& Askubuntu & 157\,K & 727\,K & 455\,K &  sec & 2614 (days)&172.8\,K\\
			& Wikitalk & 1\,100\,K & 6\,100\,K &  2\,800\,K & sec & 2277 (days)&43.2\,K\\
			\midrule
			\multirow{4}{*}{\rotatebox[origin=c]{90}{Large}}&Stackoverflow & 2.6\,M & 47.9\,M & 28.1\,M  & sec & 2774 (days)&172.8\,K\\
			&Bitcoin & 48.1\,M & 113\,M & 84.3\,M  & sec & 2585 (days)&7.2\,K\\
			&Reddit & 8.4\,M & 636\,M & 435.3\,M  & sec &3687 (days)&14.4\,K\\
			&EquinixChicago & 11.2\,M & 2\,300\,M & 66.8\,M  & $\mu$-sec &62.0 (mins)&50\,K\\
			\midrule
			&Venmo & 19.1\,K & 131\,K & 18.5\,K & sec & 2091 (days)&-\\
			\bottomrule
		\end{tabular}
	}
\end{table}

\begin{figure}[t]
\centering
\includegraphics[width=0.9\linewidth]{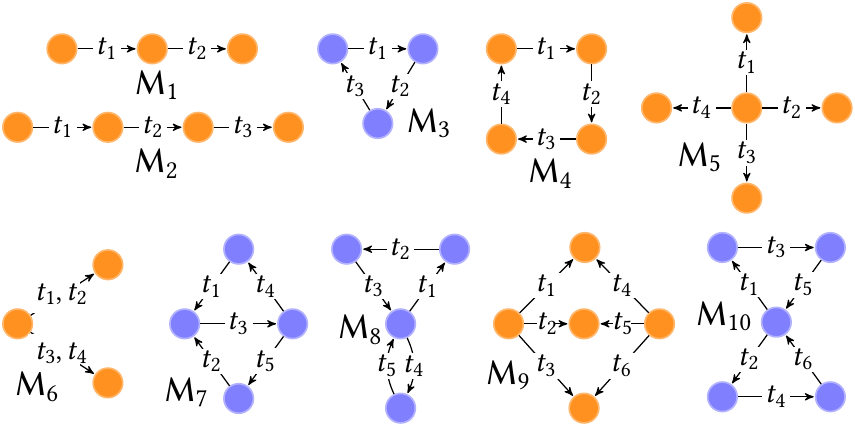} 
\caption{Temporal motifs ($\mathsf{M}_i,i\in[10]$) used in the experimental evaluation. Motifs with blue vertices are not used on EquinixChicago since this network is bipartite. For each motif $t_i, i=1,\ldots$ denotes the ordering of its edges in $\sigma$.}
\label{fig:gridMotifs}
\end{figure}

\begin{figure*}[t]
	\centering
	\captionsetup[subfigure]{labelformat=empty}
	\begin{tabular}{l}
		\subfloat{\includegraphics[width=0.93\textwidth]{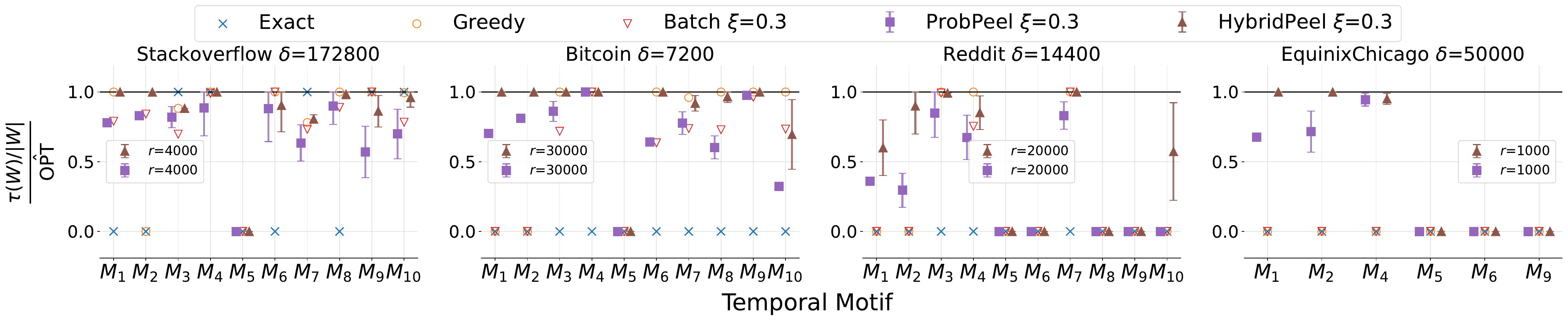}}\\[-4ex]
		\subfloat{\includegraphics[width=0.93\textwidth]{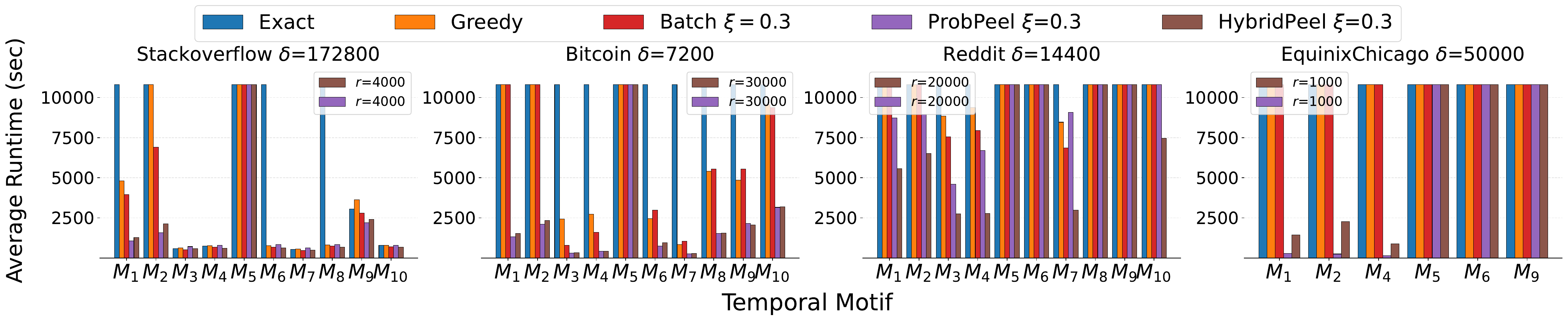}}\\[-4ex]
	\end{tabular}
	\caption{For each configuration we report (top): the quality of the solution compared to the best empirical solution (i.e., $\hat{\OPT}$), and (bottom): the average running times to achieve such solution. } 
	\label{fig:testApproxandRT}
\end{figure*}

\subsection{Setup}\label{subsec:setup}
\para{Implementation details and hardware.} All the algorithms are implemented in C++20 and compiled under Ubuntu 20.04 
with gcc 9.4.0 with optimization flags. 
The experiments are executed sequentially on a 72-core Intel Xeon Gold 5520 @2.2\,GHz machine 
with 1008\,GB of RAM memory available. 
For \suiteex, we use the algorithm by~\citet{Boykov2004} 
to compute the $(s,z)$-min cut on the flow network. 
For the implementation of the min-heap 
we use the data structures provided by the \texttt{Boost} library.\footnote{\url{https://www.boost.org/}} 
We use the algorithm of~\citet{Mackey2018} for exact enumeration of temporal motif $\delta$-instances.
We use \texttt{PRESTO-A}~\cite{Sarpe2021}, 
\ifextended
with parameter $q=1.25$, 
\else
\fi
as sampling algorithm in \suiterand and \suitehybrid. 
At each iteration of \suiterand (and \suitehybrid) we set $r_i=r$, 
(i.e., a fixed value for all iterations) and we set this parameter accordingly for each dataset;
we discuss the sensitivity of the solution to parameter $r$ 
\ifextended
in Section~\ref{subsec:sensitivity}. 
\fi
We consider as weighting function~$\tau$  the constant function $\tau_{c}$. Our code is available online.%
\footnote{\url{https://github.com/iliesarpe/ALDENTE}.}

\vspace{1mm}
\para{Datasets and time-interval length.}
The datasets considered in this work span medium and large sizes 
and are reported in Table~\ref{tab:datasets}. For each dataset we set $\delta$ to be some multiple of the respecting time unit (e.g., for datasets with precision in seconds setting $\delta=7.2$K corresponds to two hours). We also select a value of $\delta$ that is consistent with previous studies and application scenarios~\cite{Paranjape2016, Liu2023a,Pashanasangi2021, LiuFraud2024}, and large enough to be computationally challenging. 
\ifextended
See more details on datasets in Appendix~\ref{appsec:implDetandData}.
\fi

\vspace{1mm}
\para{Temporal motifs.} The temporal motifs we use are good representative of a general input to \ourproblem in most applications, and 
are reported in Fig.~\ref{fig:gridMotifs}.
They represent different topologies, e.g., triangles, squares and more complicated patterns, spanning different values of $k$ and $\ell$, and the largest temporal motifs correspond to particularly challenging inputs. 

\subsection{Solution quality and runtime}
\label{subsec:appxration}

We now compare the runtime and quality of the solution reported by the algorithms in \suitename (\algrand, and \alghybrid) against the baselines (\algex, \algkapp, and \algbatch). 

\vspace{1mm}
\para{Setup and metrics.} For each configuration (dataset, motif, value of $\delta$) we run each algorithm five times with a time limit of three hours and maximum RAM of 150\,GB on all datasets but EquinixChicago, on which we set the memory limit to 200\,GB. 
For each algorithm we compute the average running time over the five runs. 
To assess the quality of the solution of the randomized algorithms, we compute the actual value of the solution in the original temporal network on the returned vertex set. 
For each configuration over all the algorithms that terminate we compute $\hat{\OPT}$,%
\footnote{Such value is guaranteed to be the actual optimum \OPT only when~\algex\ terminates.} 
i.e., the best empirical solution obtained across all algorithms, 
and we use this value as reference for comparison of the different algorithms. 
For all the deterministic algorithms we show the approximation factor as 
$\frac{\tau(W)/|W|}{\hat{\OPT}}$, 
where $W$ is the solution returned by a given algorithm; 
for~\algrand\ and \alghybrid we show the average approximation factor 
over the five runs, and we also report the standard deviation. 

We set $\xi=0.3$ for \algbatch over all experiments we performed, as it provides sufficiently good solutions in most cases
\ifextended
(see Section~\ref{subsec:sensitivity} for a detailed empirical evaluation of such parameter), 
\fi
and use the same values of $\xi$ for~\algrand and~\alghybrid.  For~\alghybrid we set $J=3$ on large datasets, and $J=2$ otherwise.
When an algorithm {does not terminate} within the time limit 
we set its time to three hours, and its approximation factor to 0.
Since our focus is on scalability we place particular emphasis on \emph{large} datasets, so we defer results on medium-size datasets to 
\ifextended
Appendix~\ref{appsubsec:mediumdata} 
\fi
as they follow similar trends to the ones we will discuss. We also discuss results concerning memory usage in App.~\ref{subsubsec:memory}.

\vspace{1mm}
\para{Results.} The results are reported in Fig.~\ref{fig:testApproxandRT}. 
Concerning the runtime, we observe different trends on the various datasets. First on Stackoverflow, most of the motifs require a small runtime, of few hundreds of seconds, to be counted (even by \algex), and on such motifs most algorithms achieve a comparable runtime. While on some hard motifs, our proposed randomized algorithms achieve a significant speed-up of more than $\times 4$ and up to $\times 9$ over the baselines (e.g., on motifs $\mathsf{M}_1$ and $\mathsf{M}_2$). On Bitcoin we observe that the runtime for counting most motifs is prohibitive, and \algex is not able to complete the execution on any motif on such dataset. Remarkably, our proposed randomized algorithms achieve a consistent speed-up of at least $\times 3$ up to $\times 8$ over \algkapp and \algbatch, and more importantly our algorithms are able to complete their execution even when the baselines do not scale their computation (e.g., motifs $\mathsf{M}_1$ and $\mathsf{M}_2$). The scalability aspect becomes more clear on the biggest datasets that we considered, in fact on Reddit, despite of being consistently more efficient than the baselines, our \alghybrid is the only one that is able to complete its execution on $\mathsf{M}_{10}$ over all techniques considered. Finally, on EquinixChicago (with more than two billions of temporal edges), our randomized algorithms (\algrand and \alghybrid) are the only ones that terminate their execution in a small amount of time. We observe that \algrand is significantly more efficient than \alghybrid, at the expense of a slightly less accurate solution, which we will discuss next. On some of the configurations all algorithms do not terminate, this is because some of the motifs are extremely challenging, therefore the timelimit we set is too strict as a constraint even for randomized algorithms. Overall these experiments suggest that our randomized algorithms successfully enable the discovery of temporal motif densest subnetworks on large temporal networks. This closely matches our theoretical insights, capturing the superior efficiency and scalability of our randomized algorithms against baselines.

Concerning the solution quality we observe similar trends on most datasets. \algbatch provides solutions with smaller density than \algkapp, and \algrand closely matches \algbatch's results (as captured by our analysis). 
\algkapp is the deterministic approximation algorithm that mostly often provides the solution with highest density when it terminates, and our \alghybrid closely matches such results, so our \alghybrid results accurate, scalable and efficient (as it is consistently faster than \algkapp). Interestingly on the EquinixChicago dataset, where only our proposed randomized algorithms are able to terminate in less than three hours (\algrand takes few hundreds of seconds), the difference in the quality of the solutions provided by \algrand and \alghybrid is not significantly large, suggesting that on massive data \algrand may be a good candidate to obtain high-quality solutions in a very short amount of time. 

\vspace{1mm}
\para{Summary.} To summarize, our experiments show that techniques based on existing ideas do not scale on challenging motifs and large datasets for \ourproblem. In contrast, our proposed algorithms in \suitename provide high-quality solutions in a short amount of time and scale their computation on large datasets, enabling the practical discovery of temporal motif dense subnetworks.

\subsection{Case study}
\label{subsec:casestudy}

Recently, \citet{Liu2023} released a dataset containing a small number of transactions from the Venmo money-transfer platform, where each transaction is accompanied by a message. The corresponding temporal network $T=(V_T,E_T)$ is as follows: each temporal edge is a tuple $(u,v,t,\phi,\gamma)\in E_T$ where $(u,v,t)$ is the temporal edge as considered up to now and $\phi, \gamma$ are metadata: $\phi\in\{0,1\}$ denotes if $(u,v)$ are friends in the social network 
and $\gamma$ is a text message. Since the dataset is very small (see Table~\ref{tab:datasets}), we computed \emph{exact solutions} to Problem~\ref{probl:denseMotif}. We investigated the following question.

\indent \textbf{Q}: What insights about the Venmo platform are captured by optimal solutions to the \ourproblem problem according to a temporal motif $M$, and what subnetworks are captured by varying  $\delta\in \mathbb{R}^+$?

To answer \textbf{Q} we select motif $\mathsf{M}_5$ in Fig.~\ref{fig:gridMotifs}, a temporal star with four temporal edges, corresponding to finding groups of users (i.e., the vertex at the center of $\mathsf{M}_5$) sending many transactions to their neighbors (i.e., the vertices with no out-edges in $\mathsf{M}_5$) in a time-scale controlled by $\delta\in \mathbb{R}^+$.  In addition to $\mathsf{M}_5$, we provide as input to \ourproblem the constant weighting function $\tau_c$ and $\delta_1=7\,200$ (i.e., 2-hours), to capture short-time scale patterns. 

The optimal solution $T[W^*_{\delta_1}]$ has 14 vertices, and we report its directed static network and its temporal support in Fig.~\ref{fig:venmoStarSols} (Left). Interestingly, this is a star shaped network, with only the central vertex ($v_1$) exchanging money with all the other vertices (not reciprocated), and vertices are not friends.
The message associated to each transaction is identical for all transactions: \emph{``Sorry! We’re already sold out for tonight! Feel free to join us this even in the regular line and pay cover when you get there. Thanks!''}. Even more interestingly, all events occur really close in time (see Fig.~\ref{fig:venmoStarSols} (Left)). In fact, this corresponds to a  \emph{bursty event} with merchant $v_1$ overbooking for a specific event, which is identified by the combination  of $\mathsf{M}_5$ and small $\delta\in \mathbb{R}^+$. We also observe that such a subnetwork \emph{cannot be captured by the existing formulations} for dense temporal subnetworks (see Section~\ref{sec:relworks}) as both $T[W^*_{\delta_1}]$ and its directed static network have very \emph{small edge-density}, i.e., $13/14< 1$. 

We then consider $\mathsf{M}_5$ but analyze a much larger time-scale, that is we solve \ourproblem with  $\delta_2=172\,800$.
 Under this settings of parameters we expect the optimal solution to contain instances of $\mathsf{M}_5$ with a longer duration (accounting for more historical user activities). The optimal subnetwork has 16 vertices and it is shown again in Fig.~\ref{fig:venmoStarSols} (Right). As expected the temporal support of such subnetwork is significantly long (spanning from Oct.\ 2018 to Feb. 2021). The messages over the transactions of such subnetwork are usually related to food, theatre and social activities,  denoting that users share similar interests. But we also identify transactions that are likely related to sport gambling. Such suspicious transactions report terms such as \emph{``bracket season''} or emojis of basketballs (in fact, $\mathsf{M}_5$ captures such patterns as once a user loses a gamble with its friends, it usually sends the money using a pattern similar to $\mathsf{M}_5$).\footnote{
 \ifextended
 An histogram of the words associated to the transactions is in Fig.~\ref{fig:vemoWordCloud}.
 \fi} 
 Our findings also support the insights by~\citet{Liu2023}, 
who use temporal motifs to identify poker gamblers on Venmo. 

In summary,  by solving the \ourproblem problem for temporal motifs of interest and different values of the time-window $\delta$ we can gain precious insights on the network being analyzed not captured otherwise by previous formulations.


\begin{figure}[t]
\centering
\begin{tabular}{l@{}l}
	\includegraphics[width=0.48\linewidth]{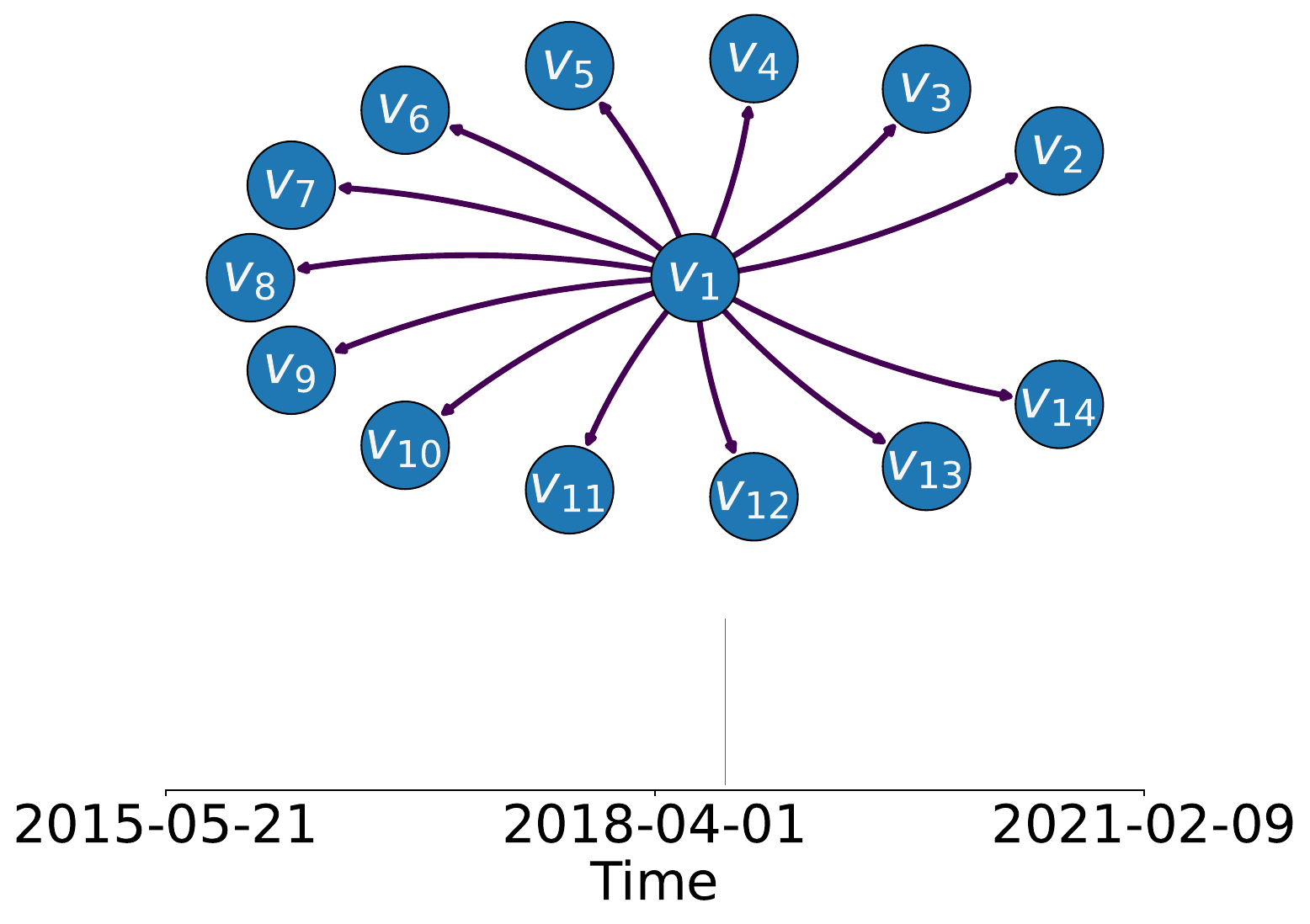} &  \includegraphics[width = 0.45\linewidth]{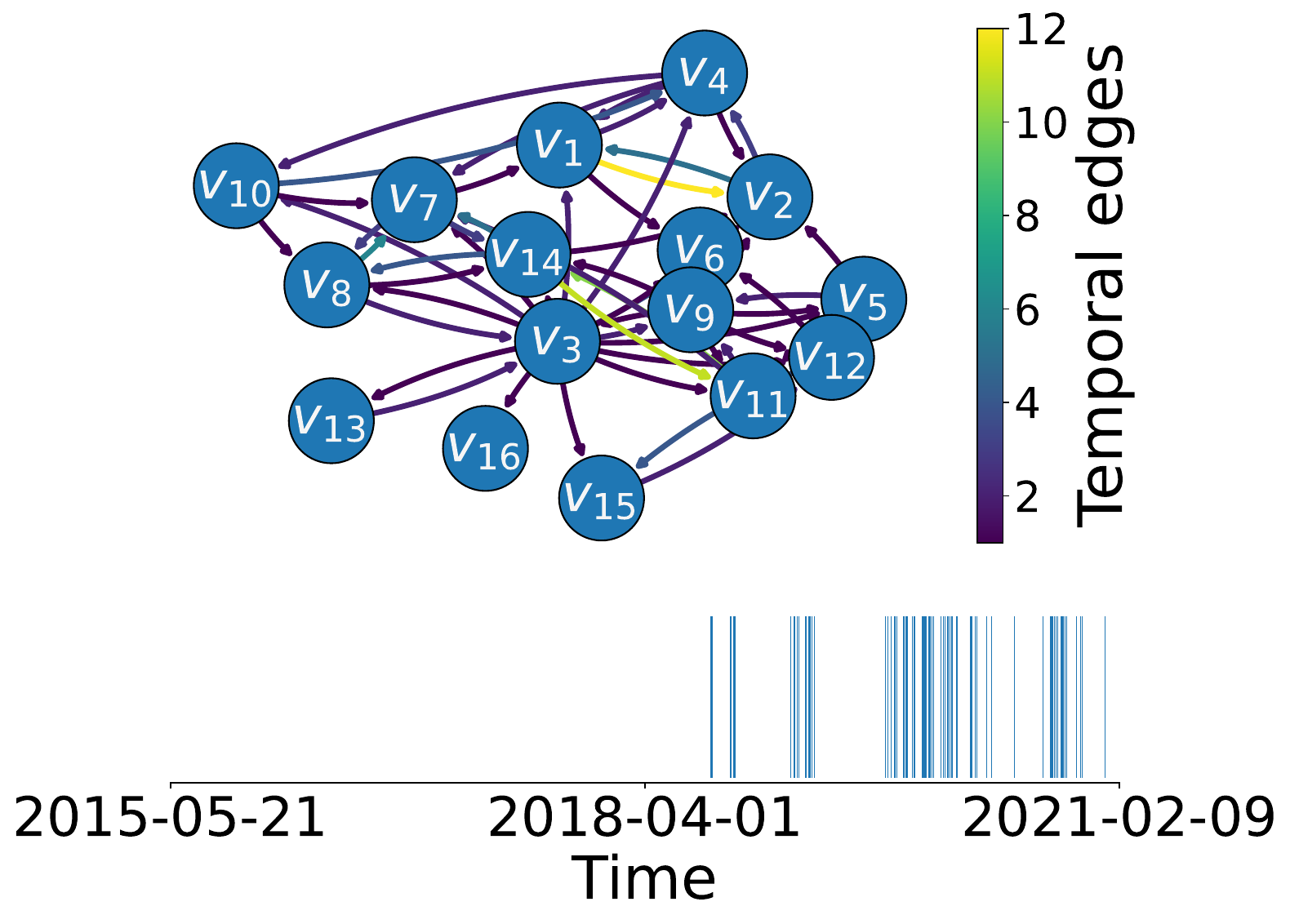}\\[-3ex]
\end{tabular}
\caption{Directed static network of $T[W^*_{\delta}]$ according to different values of $\delta$ on $\mathsf{M}_5$, we color edges according to the number temporal edges that map on each static edge. Below each  static network we report the temporal support of $T[W^*_{\delta}]$, i.e., we place a bar in correspondence of the timings of the events in $T[W_\delta^*]$ over the timespan of observation of the network. (Left): $\delta_1=7\,200$. (Right): $\delta_2=172\,800$.}
\label{fig:venmoStarSols}
\end{figure}

\section{Conclusions}
We introduced a new problem, requiring to identify the temporal motif densest subnetwork (\ourproblem) of a large temporal network. We developed two novel algorithms based on randomized sampling, for which we proved a probabilistic approximation ratio and show experimentally that they are efficient and scalable over large data. The techniques developed in this work may be useful in other problems, such as the $k$-clique problem~\cite{Tsourakakis2015,Mitzenmacher2015,Fang2019} given the availability of many sampling algorithms with tight guarantees for estimating $k$-clique counts~\cite{Jain2017, Bressan2021}. 

There are many possible directions for future work, such as improving the theoretical guarantees offered by our randomized algorithm through motif-dependent approximation ratios, and understanding if randomization can be coupled with recent ideas in the field of densest-subgraph discovery, such as techniques in~\cite{Boob2020, Chekuri2022}.

\begin{acks}
We thank Matteo Ceccarello for helpful comments on an earlier version of the current work. This work is supported, in part, by MUR of Italy, under project PRIN n. 2022TS4Y3N ``EXPAND: scalable algorithms for EXPloratory Analyses of heterogeneous and dynamic Networked Data", and project ``National Centre for HPC, Big Data and Quantum Computing'' (CN00000013).
This work is also supported by 
the ERC Advanced Grant REBOUND (834862), 
the EC H2020 RIA project SoBigData++ (871042), and 
the Wallenberg AI, Autonomous Systems and Software Program (WASP) funded by the Knut and Alice Wallenberg Foundation. 
\end{acks}

\newpage
\ifextended
\bibliographystyle{ACM-Reference-Format}
\bibliography{sample}
\fi

\appendix
\section{Application scenarios for the \ourproblem problem}\label{appsec:TMDSApplications}
While we study the \ourproblem problem in its general form, 
in this section we discuss more in depth two possible applications of our novel problem formulation, 
showing that the \ourproblem is a really versatile and powerful tool for temporal network analysis.
We will discuss how the Temporal Motif Densest Subnetwork (\ourproblem) can be used to discover important insights from 
(i) travel networks and (ii) e-commerce networks capturing online platforms.

(i) A travel network can be modeled as a temporal network $T=(V,E)$ with $V$ being the set of vertices corresponding to Points Of Interest (POIs) 
or particular geographic areas of a city~\cite{Lei2020}, and edges of the form $(u,v,t)\in E$ represent trips by user that travel from point 
$u\in V$ to point $v\in V$ at time $t$.  
A temporal motif $M$ on such network captures therefore travel patterns and their dynamics. 
As an example a temporal motif $M: x\stackrel{t_1}\to y \stackrel{t_2}\to x$ occurring within one day 
often corresponds to a round trip by a user from $x$ being the home location, to $y$ being the work location. 
The \ourproblem problem formulation, in such a scenario can provide unique insights on the POIs appearing frequently together 
in the various travel patterns (isomorphic to $M$) of the various users. 
Furthermore, analyzing different time-scales as captured by the temporal motif duration, 
can yield unique insights about daily vs.\ weekly vs.\ monthly patterns.
 As for our Figure~\ref{fig:basicdef} in the \ourproblem several POIs often coexist, and such information can be used to improve 
 connections between POIs that are not well connected by public transport, 
 or for the design of travel passes for specific areas of a city with specific time-duration, 
 based on POIs that are often visited together in the various trips.
 
(ii) E-commerce online platforms can be modeled as temporal bipartite networks $T=(V,E)$, 
where the set of vertices is partitioned in two layers $V=U \cup P$, 
with $U$ being the set of users (i.e., customers that purchase online products), 
and $P$ the set of products available for purchase. 
A temporal edge $e=(u,p,t)$ on such a network captures 
that user $u \in U$ purchased product $p \in P$ at a given time $t$. 
As an example a temporal motif $M': x\stackrel{t_1}\to y, x \stackrel{t_2}\to z$ occurring within a time-limit
corresponds to a pair of products i.e.,  $y$ and $z$ that user $x$ buys at distance of at most the provided time-limit, note also that $z$ is purchased after $y$. 
A temporal motif $M$ therefore captures specific purchase habits of users within a limited time-limit 
(as again controlled by the duration parameter of the temporal motif). 
The optimal \ourproblem can be used to collect a set of users (and items) that frequently co-occur. 
In particular, the \ourproblem can contain several vertices (i.e., users) with similar purchase sequences or buying habits, 
this enables the design of personalized advertisement for those users in the \ourproblem 
(e.g., by leveraging the history of other users in the \ourproblem), furthermore this can be studied at different time-scales. 
Note also that this is more powerful than only considering products that are frequently co-purchased together 
(e.g., consider $M'$, then the products $y$ and $z$ are purchased in different moments), for which many techniques already exist.

\begin{table}[t]
	\centering
	\caption{Notation table.}
	\label{tab:notationTable}
	\scalebox{0.9}{
		\begin{tabular}{ll}
			\toprule
			Symbol & Description\\
			\midrule
			$T=(V_T,E_T)$ & Temporal network\\
			$T[W], W\subseteq V_T$ & Induced temporal sub-network by $W$\\
			$n,m$ & Number of nodes and temporal edges of $T$\\
			$D_T, G_T$ & Directed and undirected static network of $T$\\
			$M=(K, \sigma)$ & $k$-vertex $\ell$-edge temporal motif\\
			$K=(V_K,E_K)$ & Multi-graph of the temporal motif $M$\\
			$\sigma$ & Ordering of the edges of $E_K$ in a motif\\
			$\delta$ & Time window-length of a temporal motif instance\\
			$\mathcal{S}_{T'}$ & Set of $\delta$-instances of $M$ in the network $T'$\\ 
			$\mathcal{S}_{W}, W\subseteq V_T$ & Set of $\delta$-instances of $M$ in the network $T[W]$\\
			$\tau$ & Weighting function over $\mathcal{S}_T$\\
			$\tau(W), W\subseteq V_T$ & Total weight of the $\delta$-instances in $\mathcal{S}_W$\\
			$\fobj(\cdot)$ & \ourproblem objective function\\
			$G[M]$ & Undirected graph associated to $K$\\ 
			$\mathcal{H}$ &  Set of $k$-CIS over $G_T$\\
			$\temporalDegreeWeight{W}{v}, W\subseteq V_T$ & Temporal motif degree of vertex $v$ in $T[W]$\\
			$\temporalDegreeWeightNetwork{T'}{v}, T'\subseteq T$ & Temporal motif degree of vertex $v$ in $T'$\\
			$\xi>0$ & Threshold for batch peeling\\
			$\varepsilon,\eta$ & Accuracy and confidence parameters\\
			$r$ & Sample size in \algrand, \alghybrid\\
			$\estimateTemporalDegreeWeight{W}{v}$ & Estimate of the temporal motif degree of $v$ in $T[W]$\\
			$\hat{m}$ & \shortstack[r]{Maximum edges from $E_T$ in a window length of $\delta$}\\
			$q$ & Parameter of the subroutine \texttt{PRESTO-A}~\cite{Sarpe2021}\\
			\bottomrule
		\end{tabular}
	}
\end{table}

\section{Notation}\label{appsec:notation}
A summary of the notation used throughout the paper is reported in Table~\ref{tab:notationTable}.

\ifextended

\begin{table*}[t]
	\centering
	\caption{\emph{Directed} denotes if the edges of the network are directed or not. \emph{Aggregation} means: (i) the objective is based on purely static quantities over snapshots (i.e., it discards temporal information) (ii) the edges are assumed to be in the form $(u,v,[t_1,t_2])$, that is each edge is associated with an interval instead of a timestamp, in such interval the graph is static (iii) the optimal scores are computed on specific time-intervals of the network. \emph{High-order structures} denotes if the objective function of the formulation considers more than node or edge degrees, accounting for cliques or larger patterns. $\mathsf{D}$ vs $\mathsf{C}$ denotes if the objective is to compute cores ($\mathsf{C}$) or dense subnetworks ($\mathsf{D}$). 
		\emph{Segmentation} denotes that the formulation controls the temporal properties of the reported solution, e.g., burstiness-persistence or a specific temporal support duration.}
	\label{tab:algsComparison}
	\scalebox{0.95}{
		\begin{tabular}{cccccc}
			\toprule
			Formulation & Directed & Aggregation & High-order structures& Density ($\mathsf{D}$) vs Cores ($\mathsf{C}$) & 
			Segmentation\\
			\midrule
			\ourproblem (ours) & \vmark & \xmark & \vmark & $\mathsf{D}$ & \xmark \\
			\cite{Galimberti2020} & \xmark & \vmark & \xmark & $\mathsf{C}$ & \vmark \\
			\cite{Li2018} & \xmark & \vmark & \xmark & $\mathsf{C}$ &  \vmark\\
			\cite{Oettershagen2023Core} & \xmark & \xmark & \xmark & $\mathsf{C}$ & \xmark\\
			\cite{Qin2022} & \xmark & \vmark & \xmark & $\mathsf{C}$ & \vmark\\
			\cite{Zhong2024} & \xmark & \vmark & \xmark & $\mathsf{C}$ & \vmark\\
			\cite{Lin2022} & \xmark & \vmark & \vmark & $\mathsf{D}$ &  \vmark\\
			\cite{Qin2023} & \xmark & \vmark & \xmark & $\mathsf{D}$ &  \vmark\\
			\cite{Rozenshtein2019} & \xmark & \vmark & \xmark & $\mathsf{D}$ & \vmark\\
			\bottomrule
	\end{tabular}}
\end{table*}

\section{\ourproblem vs existing formulations}~\label{appsubsec:existingProblems}
In this section we present an overview of the existing formulations for cohesive subnetwork discovery in temporal networks comparing them with our \ourproblem formulation. The extensive comparison is reported in Table~\ref{tab:algsComparison}.
As a summary of such table we highlight the following, our formulation is the unique of its kind combining \emph{high-order structures} with \emph{purely temporal informal information} avoiding therefore temporal aggregation, and additionally it can be used on directed temporal networks. 

First we observe that most of the existing formulations work on \emph{undirected} temporal networks, and it is not straightforward to adapt (or extend) them on directed temporal networks, which are more general models. 
Our formulation naturally captures directed temporal networks, and in fact fully exploits such information, as an example,  with the definition of temporal motif that we assume (Definition~\ref{def:k-l-temporal-motif}) there are eight possible triangles on directed temporal networks, but \emph{only one} such triangle exists on undirected temporal networks, reducing significantly the information that we can gain by analyzing such patterns. 
Then our formulation avoids a temporal network model leading to aggregation, which is known to cause information loss~\cite{Holme2019}, as computing snapshots of temporal networks (or constraining optimal scores to specific fixed intervals) may obfuscate the general temporal properties in the data. 
Additionally, only one of the surveyed formulations (i.e.,~\cite{Lin2022}) accounts for high-order structures in the form of \emph{static quasi-cliques} which is a completely different formulation compared to \ourproblem. Our formulation is in fact much more flexible as the temporal motif can by \emph{any arbitrary temporal motif} of interest. Finally we also avoid segmentation as the constrain on the time-interval length $\delta$ well correlates with the overall timespan of the reported solution, and with more complex weighting functions $\tau$ defined over $\mathcal{S}_T$  it is also possible to control for the temporal properties of the reported solution (e.g., a persistence score from~\cite{Belth2020}).

\section{Randomized algorithms}~\label{appsec:randAlg}
In this section we present some missing results about our \suitename algorithms based on randomized sampling, that is in Section~\ref{sec:embeddingSamplingAlg} we discuss how to estimate temporal motif degrees of the vertices and leverage such results with an existing sampling algorithm to prove concentration of the various estimates, these results are used both for~\algrand and for~\alghybrid. Lastly, in Section~\ref{appsubsec:pseudocodehybrid} we illustrate the pseudocode of~\alghybrid.
\subsection{Estimating temporal motif-degrees}\label{sec:embeddingSamplingAlg}
\subsubsection{Building the estimator}\label{appsubsubsec:buildingEstimate} We briefly recall from Section~\ref{sec:algorithms} how to compute a $(1\pm\varepsilon)$ estimate of the temporal motif degree $\temporalDegreeWeightNetwork{W}{w}$ of a vertex $w\in W$. First recall that we assume access to a sampling algorithm $\mathcal{A}$, which, 
given as input a sub\-network $T[W]$, $W\subseteq V_T$, 
it outputs a sample $T'=(V',E')\subseteq T[W]$. Using such a sample $T'$ we can compute an estimate   $\estimateTemporalDegreeWeightPrime{W}{w}{T'}$ of the degree $\temporalDegreeWeight{W}{w}$,
for each $w\in W$ as 
\begin{equation}\label{eq:prestoEst}
	\estimateTemporalDegreeWeightPrime{W}{w}{T'} = \sum_{S \in \mathcal{S}_{{T}'}:w\in S}\frac{\tau(S)}{p_S},
\end{equation}
where $p_S$ is the probability of observing $S\in \mathcal{S}_{{T}'}$ 
in the sampled sub\-network ${T}'=({V}', {E}')$.
\begin{lemma}\label{lemma:unbiasedness}
	For any $w\in W$, the count $\estimateTemporalDegreeWeightPrime{W}{w}{T'}$ computed on a sampled subnetwork $T'\subseteq T[W]$ is an unbiased estimate of $\temporalDegreeWeight{W}{w}$.
\end{lemma}
\begin{proof}
	We can rewrite the estimator as
	\[
	\estimateTemporalDegreeWeightPrime{W}{w}{T'} = \sum_{S \in \mathcal{S}_{T[W]} :w\in S}\frac{\tau(S)}{p_S}\mathbbm{1}[ \{S\in \mathcal{S}_{{T}'} \} ],
	\]
	where $\mathbbm{1}[\cdot]$ is an indicator random variable taking value 1 if the predicate holds, 
	and 0 otherwise. 
	By taking  expectation on both sides, combined with the linearity of expectation, 
	and the fact that $\mathbb{E}[\mathbbm{1}[ \{S\in \mathcal{S}_{{T}'} \} ]] = p_S$ 
	the claim follows by the definition of $\temporalDegreeWeight{W}{w}$.
\end{proof}
Given a subnetwork $T[W]\subseteq T, W\subseteq V_T$ we execute the sampling algorithm $\mathcal{A}$ to obtain $r\in \mathbb{N}$ i.i.d.\ sampled subnetworks $\mathcal{D}=\{{T}_1, \dots, {T}_r : {T}_i =\mathcal{A}(T[W]) \subseteq T[W], i\in[r] \}$ (in line~\ref{algsampl:getSamples}). We can therefore compute the sample average of the estimates obtained on each sampled subnetwork obtaining the estimators $\estimateTemporalDegreeWeight{W}{w}$ for $w \in W$, where
\[
\estimateTemporalDegreeWeight{W}{w} = \frac{1}{r}\sum_{i=1}^r \estimateTemporalDegreeWeightPrime{W}{w}{T_i}, \quad \text{for all } w \in W,
\]
and $\estimateTemporalDegreeWeightPrime{W}{w}{T_i}$ is computed with Equation~\eqref{eq:prestoEst} on the $i$-th sampled subnetwork $T_i$ from the randomized sampling algorithm $\mathcal{A}$. 
As discussed in Section~\ref{sec:algorithms}, it also holds that $\hat{\tau}(W) = 1/k \sum_{w\in W} \estimateTemporalDegreeWeight{W}{w}$ is an unbiased estimate of $\tau(W)$, as $\mathbb{E}[\estimateTemporalDegreeWeight{W}{w}] = \temporalDegreeWeight{W}{w}$ by Lemma~\ref{lemma:unbiasedness}  combined with the linearity of expectation and the fact that $\tau(W)=1/k\sum_{w\in W} \temporalDegreeWeight{W}{w}$, which is an important property for the final convergence guarantees of our algorithms in \suitename.

\subsubsection{Leveraging a sampling algorithm}\label{appsubsubsec:embeddingSamplingAlg}  Algorithm~\ref{alg:kepsilonapproxsampling} can employ any sampling algorithm $\mathcal{A}:T\mapsto 2^T$ as subroutine. In practice we used the state-of-the-art sampling algorithm \texttt{PRESTO-A}~\cite{Sarpe2021} for estimating temporal motif counts of arbitrary shape.\footnote{Other specialized sampling algorithms for triangles or butterflies can be adopted in our~\suitename algorithms speeding up our suite~\cite{Wang2020a, Pu2023TemporalButterfly}.} For \texttt{PRESTO-A} the authors provide bounds on the number of samples $r$ for event ``$\hat{\tau}(W_i)\in (1\pm \varepsilon)\tau(W_i)$'' to hold with arbitrary probability $>1-\eta$ for $\eta\in(0,1)$, which needs to be adapted to work in our scenario since we require \emph{stricter} guarantees. We adapt such algorithm to general weighting functions $\tau$ over the set of $\delta$-instances, and to compute $\estimateTemporalDegreeWeight{W_i}{w}, w\in W_i$ at each iteration $i$ of Algorithm~\ref{alg:kepsilonapproxsampling}. We now briefly describe how \texttt{PRESTO-A} works and proceed to illustrate how to adapt the bound in~\cite{Sarpe2021} to the function  $\mathtt{GetBound}$, since a tighter concentration result is needed.

Given a subnetwork $T[W_i]$, \texttt{PRESTO-A} samples a small-size subnetwork ${T}'=({V}',{E}')\subseteq T[W_i]$ ensuring that all its edges are contained in a window of length $q\delta, q>1$, i.e., let $t_1,\dots, t_{|{E}'|}$ be the sorted timestamps of the edges of $T'$ then $t_{|{E}'|}-t_1\le q\delta$. 
On such subnetwork \texttt{PRESTO-A} computes Equation~\eqref{eq:prestoEst} for each $w\in W_i$, where $p_S = \frac{q\delta - \Delta_S}{\Delta_{T[W_i]}}$, where $\Delta_S$ is the length of the $\delta$-instance $S\in\mathcal{S}_{{T}'}$ and $\Delta_{T[W_i]} = \bigO(t_{|E_T[W_i]|} - t_1)$, i.e., it is simply proportional to the length of the timespan of $T[W_i]$. 
Such algorithm extracts $r_i$ sampled subnetworks and computes the unbiased estimates $\estimateTemporalDegreeWeight{W_i}{w}, w\in W_i$ and $\hat{\tau}(W_i) = 1/k \sum_{w\in W_i} \estimateTemporalDegreeWeight{W_i}{w}$. 
Next we show 
that leveraging the techniques in~\cite{Sarpe2021} we obtain a bound on the number of samples $r_i, i\in [I]$ computed through $\mathtt{GetBound}$.
In particular, by performing a similar analysis to the one in~\cite{Sarpe2021} we obtain the following corollary, which specifies the output of the function $\mathtt{GetBound}$.

\begin{corollary}\label{corol:sampleSize}
	Let $T[W_i]$ be the induced temporal network by $W_i$ at an arbitrary iteration  $i$ of Algorithm~\ref{alg:kepsilonapproxsampling}, given $\varepsilon>0, \eta/2^i>0,q>1$ if $r_i\ge \left(\frac{\Delta_{T[W_i]}}{(q-1)\delta}-1\right) \frac{1}{(1+\varepsilon)\ln(1+\varepsilon)-\varepsilon}\ln\left(\frac{2|W_i|}{\eta/2^i}\right)$ then $\estimateTemporalDegreeWeight{W_i}{w}\in(1\pm\varepsilon)\temporalDegreeWeight{W_i}{w}$ simultaneously for each $w\in W_i$ 
	with probability  $>1-\eta/2^i$.
\end{corollary}
\begin{proof}[Proof (skecth).]
\sloppy{
The proof is identical and directly follows from the guarantees obtained in~\cite[Theorem 3.4., extended arXiv version]{Sarpe2021} combined with a union bound. Fix an arbitrary vertex $w\in W_i$ and a sample $T'\subseteq W_i$, then from Lemma~\ref{lemma:unbiasedness} the estimator $\estimateTemporalDegreeWeightPrime{W_i}{w}{T'}$ is unbiased where under the sampling algorithm \texttt{PRESTO-A} we set $p_S = \tfrac{q\delta - \Delta_S}{\Delta_{T[W_i]}}$ where $\Delta_S$ is the duration of each $\delta$-instance $S\in \mathcal{S}_{W_i}$ and $\Delta_{T[W_i]} = t_{|E_T[W_i]|-\ell}-t_\ell + c\delta$ where $t_i$ denotes the timestamp of the $i$-th edge (of $E_T[W_i]$) assuming that such edges are sorted by increasing  timestamps. This implies that the bound on the variance of such estimator is the same as~\cite[Lemma 3.2., extended arXiv version]{Sarpe2021} where the global motif count is replaced with the temporal motif degree, i.e., $\mathbb{E}[\estimateTemporalDegreeWeightPrime{W_i}{w}{T'}^2] \le \frac{\temporalDegreeWeight{W_i}{w}^2 \Delta_{T[W_i]} }{ (c-1)\delta}$. Using~\cite[Theorem 3.4., extended arXiv version]{Sarpe2021} for an arbitrary vertex $w\in W_i$ it holds, that if $r_i \ge \left(\frac{\Delta_{T[W_i]}}{(q-1)\delta}-1\right) \frac{1}{(1+\varepsilon)\ln(1+\varepsilon)-\varepsilon}\ln\left(\frac{2}{\eta/2^i}\right)$ then
$\mathbb{P}[|\estimateTemporalDegreeWeight{W}{w}-\temporalDegreeWeight{W}{w}| \ge \varepsilon\temporalDegreeWeight{W}{w}] \le \eta/2^i$.
The final claims from a direct application of the union bound over all vertices $w\in W_i$.}
\end{proof}

\subsection{ \suitehybrid --- pseudocode}\label{appsubsec:pseudocodehybrid}
\begin{algorithm}[t]
	\caption{\suitehybrid}\label{alg:hybridpeel}
	\KwIn{$T=(V_T,E_T), M , \delta\in \mathbb{R}^+, \tau, \xi>0, \varepsilon >0, \eta \in (0,1), J\ge1$.}
	\KwOut{$W$.}
	$W_1 \gets V_T$\;
	\For{$(i\gets 1,\dots,J) \land (W_i \neq \emptyset)$\label{alghyb:mainLoop}}{
		$r_i \gets \mathtt{GetBound}(W_{i}, \varepsilon, \eta/J)$\label{alghyb:boundSamples}\;
		$\mathcal{D}=\{T_1,\dots,T_{r_i}\} \gets \mathcal{A}^{r_i}(T[W_i])$\label{alghyb:getSamples}\;
		$\estimateTemporalDegreeWeight{W_i}{w_1},\dots,\estimateTemporalDegreeWeight{W_i}{w_{|W_i|}}\gets \mathtt{GetEstimates}(W_i, \mathcal{D})$\label{alghyb:computeEstimates}\;
		$\hat{\tau}(W_i)\gets 1/k \sum_{w\in W_i} \estimateTemporalDegreeWeight{W_i}{w}$\label{alghyb:estimateTotalCount}\;
		$L(W_i)\gets\left\{w\in W_i : \estimateTemporalDegreeWeight{W_i}{w}\le k(1+\xi) \hat{\tau}(W_i)/|W_i|\right\}$\label{alghyb:peelNodes}\;
		$W_{i+1} \gets W_i\setminus L(W_i)$; $i\gets i+1$\label{alghyb:setNextIt}\;
	}
	$W' \gets $ \suitegreedy($T[W_{J+1}], M, \delta, \tau$)\label{alghyb:callGreedyPeel}\;
	\KwRet{$W=\argmax_{j=1,\dots,J}\{\hat{\tau}(W_j)/|W_j|, \fobj(W')\}$.\label{alghyb:return}}
\end{algorithm}

We report the missing pseudocode of~\suitehybrid in Algorithm~\ref{alg:hybridpeel}. The algorithm proceeds as follows, it runs for $J$ iterations the randomized batch peeling approach (lines~\ref{alghyb:mainLoop}-\ref{alghyb:setNextIt}), obtaining therefore the sets $W_1,\dots,W_J$ and $W_{J+1}$. For the vertex sets $W_i, i\in[J]$ the algorithm has only approximate density $\hat{\tau}(W_i)/|W_i|$ where the approximation factor is controlled by the parameter $\varepsilon$ (i.e., $\hat{\tau}(W_i)\in (1\pm\varepsilon)\tau(W_i)$), instead on $W_{J+1}$ the algorithm can be summarized by applying~\suitegreedy on $T[W_{J+1}]$ (that we discuss in details in Section~\ref{subsec:greedy}). Therefore greedy peeling is applied on $T[W_{J+1}]$, which proceeds by removing vertices one at a time according to their exact temporal motif degree in $T[W_{J+1}]$ (line~\ref{alghyb:callGreedyPeel}). That is, the algorithm computes all temporal motif-instances in  $T[W_{J+1}]$ and it greedily removes, one at a time, the vertices with minimum temporal motif-degree, \algkapp therefore returns the solution with maximum density among the $\bigO(|W_{J+1}|-k)$ vertex sets observed (i.e., it starts from $W_{J+1}$ and peels it one vertex at a time until it becomes empty). The solution with maximum density is finally returned by \alghybrid, where the maximum is taken over the $J$ solutions obtained in the randomized batch peeling loop and the best solution identified by~\algkapp when executed on $T[W_{J+1}]$ (line~\ref{alghyb:return}). The approximation ratio~\alghybrid is shown in Theorem~\ref{lemma:ratioHybrid}, and it is presented in the next section. 

We recall that intuitively, \alghybrid achieves very good performances and scales on large data since, at the beginning of its execution during the first $J$ iterations (when the temporal network is large) it relies on the approximate estimation of temporal motif degrees for the batch peeling step, while when the network size is small (after the $J$ iterations)~\alghybrid peels vertices one at a time according to their actual temporal motif degrees in the remaining subnetwork, which can be computed efficiently.

\section{Missing proofs}\label{sec:missingProofs}

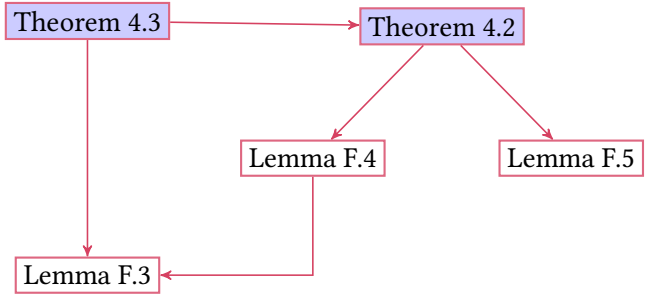
\begin{figure}[t]
	\centering
		\begin{tikzpicture}
			[green place/.style={circle,draw=blue!60,fill=ciano!20,thick,
				inner sep=0pt,minimum size=6mm}, 
			bend angle=22,
			pre/.style={<-,>=stealth',semithick},
			post/.style={->,>=stealth',semithick},
			arancio place/.style={circle,draw=arancio!60,fill=arancio!90,thick,
				inner sep=0pt,minimum size=6mm},
			verde place/.style={circle,draw=verde!60,fill=verde!90,thick,
				inner sep=0pt,minimum size=6mm},
			ciano place/.style={circle,draw=ciano!60,fill=ciano!90,thick,
				inner sep=0pt,minimum size=6mm},
			violetto place/.style={circle,draw=violetto!60,fill=violetto!90,thick,
				inner sep=0pt,minimum size=6mm},
			rossiccio place/.style={circle,draw=rossiccio!60,fill=rossiccio!80,thick,
				inner sep=0pt,minimum size=6mm},
			white place/.style={circle,draw=white!100,fill=white!100,thick,
				inner sep=0pt,minimum size=1mm},
			mylabel/.style={midway,text width=3mm, inner sep=1.5pt, fill=white, align=center}]
			\node[rectangle,thick, draw=rossiccio!80] (D3) {\Large Lemma~\ref{lemma:kapproxratio}};
			
			\node[rectangle,thick, draw=rossiccio!80] (D4) [above right=1.5of D3] {\Large Lemma~\ref{lemma:appxbatchpeel}};
			\node[rectangle,thick, draw=rossiccio!80] (D5) [right=1.5of D4] {\Large Lemma~\ref{lemma:batchPeelIt}};
			
			\node[rectangle,fill=blue!20,thick, draw=rossiccio!80] (41) [above=1.5of $(D4)!0.5!(D5)$] {\Large Theorem~\ref{lemma:ratioSampling}};
			
			\node[rectangle,fill=blue!20,thick, draw=rossiccio!80] (42) [above=2.88of D3] {\Large Theorem~\ref{lemma:ratioHybrid}};
			
			\draw[rossiccio,>=stealth',semithick, ->] (D4.south) |- (D3.east);
			\path[post,rossiccio] (41) edge [bend left=0] (D4);
			\path[post,rossiccio] (41) edge [bend left=0] (D5);
			
			\path[post,rossiccio] (42) edge [bend left=0] (41);
			\path[post,rossiccio] (42) edge [bend left=0] (D3);			
	\end{tikzpicture}
	\caption{Dependency graph for the main results we obtain, i.e., Theorem~\ref{lemma:ratioSampling} and Theorem~\ref{lemma:ratioHybrid}.}
	\label{fig:proofStruct}
\end{figure}
 In this section we present the missing proofs for Theorem~\ref{lemma:ratioSampling}, Theorem~\ref{lemma:ratioHybrid}, and Lemma~\ref{lemma:staticvstemporal}. We show in Figure~\ref{fig:proofStruct} the dependency graph for the proofs of Theorem~\ref{lemma:ratioSampling} and Theorem~\ref{lemma:ratioHybrid} as they depend on results that we obtain in Section~\ref{appsec:baselines}, for a better understanding, we suggest to follow the proofs in reverse order of dependence.

\begin{proof}[Proof of Theorem~\ref{lemma:ratioSampling}]
	We first fix the random variable corresponding to the number of iterations of Algorithm~\ref{alg:kepsilonapproxsampling} to $I\le \bigO(\log_{1+\xi}(n))$ since an analogous result to Lemma~\ref{lemma:batchPeelIt} holds also for Algorithm~\ref{alg:kepsilonapproxsampling}, i.e., we are conditioning on the realization of $I$ in its range. In the following analysis we assume that the events $``\estimateTemporalDegreeWeight{W_i}{w}\in (1\pm\varepsilon) \temporalDegreeWeight{W_i}{w}, w \in W_i$'' hold at each iteration  $i\in[I]$ of our Algorithm~\ref{alg:kepsilonapproxsampling} with probability $> 1-\eta/2^i$, note that this is ensured by the function \texttt{GetBound}.
	
	First we note that there should exist at least one vertex in $L(W_i), i\in [I]$, since for contradiction let $\estimateTemporalDegreeWeight{W_i}{w} > k(1+\xi)\hat{\tau}(W_i)/|W_i|$ for each $w\in W_i$, then we can prove the following $k\hat{\tau}(W_i)= \sum_{w\in W_i} \estimateTemporalDegreeWeight{W_i}{w} > k(1+\xi)|W_i|\hat{\tau}(W_i)/|W_i|$ hence $1>(1+\xi), \xi>0$, leading to the contradiction.
	The analysis is similar to Lemma~\ref{lemma:appxbatchpeel} but  complicated by the fact that the quantities involved are estimates. Consider the iteration in which the first vertex $v\in W^*$ is removed, let $W'$ be the set of vertices in such iteration and notice that $W^*\subseteq W'$. 
	Connecting $\hat{\tau}(W')$ and $\tau(W')$ through $\estimateTemporalDegreeWeight{W'}{w},w\in W'$.
	\[
	\hat{\tau}(W') = \frac{1}{k}\sum_{w\in W'} \estimateTemporalDegreeWeight{W'}{w} \ge \frac{1}{k}\sum_{w\in W'} (1-\varepsilon) \temporalDegreeWeight{W'}{w} \ge (1-\varepsilon)\tau(W')
	\]
	And similarly $(1+\varepsilon)\tau(W') \ge \hat{\tau}(W')$. We recall that for $v\in W^*$ it holds that $\temporalDegreeWeight{W^*}{v}\ge \tau(W^*)/|W^*|$. And that $\estimateTemporalDegreeWeight{W'}{v}\ge (1-\varepsilon) \temporalDegreeWeight{W'}{v} \ge (1-\varepsilon) \temporalDegreeWeight{W^*}{v}$, since $W^*\subseteq W'$. Combining everything we get
	\begin{align}
		&k(1+\xi)(1+\varepsilon) \frac{\tau(W')}{|W'|} \ge k(1+\xi)\frac{\hat{\tau}(W')}{|W'|} \ge \estimateTemporalDegreeWeight{W'}{v} \ge (1-\varepsilon)\temporalDegreeWeight{W^*}{v}\nonumber\\
		& \ge (1-\varepsilon) \frac{\tau(W^*)}{|W^*|} \Leftrightarrow \frac{\tau(W')}{|W'|} \ge \frac{(1-\varepsilon)}{k(1+\xi)(1+\varepsilon)} \frac{\tau(W^*)}{|W^*|}. \label{eq:boundWprime}
	\end{align}
	Let $W$ be the vertex set returned by Algorithm~\ref{alg:kepsilonapproxsampling} and $W'$ the solution achieving the approximation in Equation~\eqref{eq:boundWprime}, note that $\hat{\tau}(W')/|W'|\ge (1-\varepsilon)\fobj(W')$ and the output $W$ may be such that $\hat{\tau}(W)/|W|\ge (1-\varepsilon)\fobj(W')$ but with $\fobj(W)\le\fobj(W')$, but noting that $\hat{\tau}(W)\in(1\pm\varepsilon)\tau(W)$ then $ (1-\varepsilon)\fobj(W')\le \hat{\tau}(W)/|W| \le (1+ \varepsilon)\fobj(W)$ therefore $\fobj(W)\ge \frac{(1-\varepsilon)}{(1+\varepsilon)}\fobj(W')$ 
	hence the claim follows by combining the above with Equation~\eqref{eq:boundWprime} and the fact that Algorithm~\ref{alg:kepsilonapproxsampling} will return a vertex set achieving at least the solution attained by $W$, which proves the  bound on the approximation factor. Finally, to show that such guarantees hold with probability $>1-\eta$, let $F_i$ be the event that the probabilistic guarantees of Algorithm~\ref{alg:kepsilonapproxsampling} fail at iteration $i\in [I]$, this  occurs only if
	the following bad event occurs:
	$F_i \mathbin{\color{blue}{=}}  \text{``there exist } w\in W_i \text{ such that } \estimateTemporalDegreeWeight{W_i}{w}\notin (1\pm \varepsilon) \temporalDegreeWeight{W_i}{w}$''. 
	Recall that at each iteration $i\in [1,\bigO(\log n)]$ the sample size obtained by \texttt{GetBound} guarantees that $\mathbb{P}[F_i]\le \eta/2^i$, then the probability of the event $F_R$ that the algorithm fails under random iterations $R$ during its execution is bounded by
	\[
	\mathbb{P}\left(F_R\right) < \sum_{i=1}^{+\infty} (\mathbb{P}\left[ F_R | R=i\right]\mathbb{P}\left[I=i\right])
	\le \sum_{i=1}^{+\infty} \mathbb{P}[F_i] \le \sum_{i=1}^{+\infty} \frac{\eta}{2^{i}} \le \eta,
	\]
	combining the union bound and the geometric series, concluding the proof.
\end{proof}

\begin{figure}[t]
	\centering
	\begin{tabular}{c}
		\includegraphics[width=0.8\linewidth]{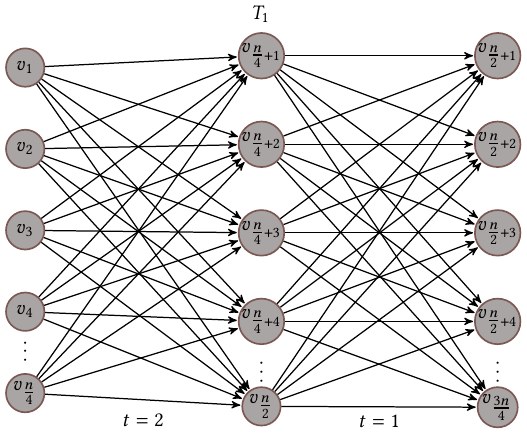}\\
		\includegraphics[width =0.8\linewidth]{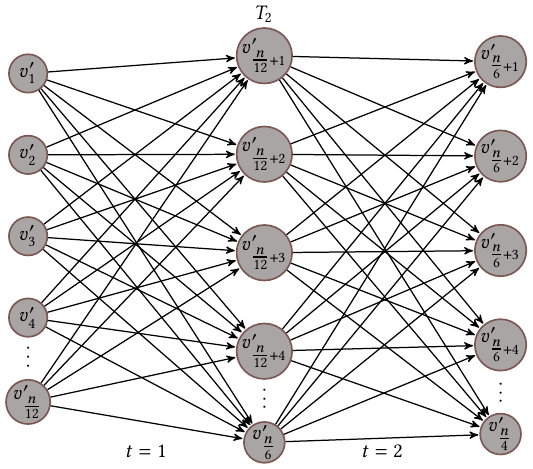}
	\end{tabular}
	\caption{Temporal network $T\mathbin{\color{blue}{=}}  (T_1=(V_1,E_1) \cup T_2=(V_2,E_2) )$ as constructed in the proof of Lemma~\ref{lemma:staticvstemporal}. The timestamps of the temporal edges are placed below each layer, for simplicity of visualization. $V_1$ is the optimal solution for S2DS but it achieves a solution that is 0 for \ourproblem problem, for which the vertices in $V_2$ achieve the optimal solution \bigO($n^2$).}
	\label{fig:lemmastaticvstemporal}
\end{figure}

\begin{proof}[Proof of Theorem~\ref{lemma:ratioHybrid} (Sketch)]
	The proof is similar to the one of Lemma~\ref{lemma:ratioSampling}. We only need to account for some sub-cases that may occur when considering the solution $W'$ obtained in Line~\ref{alghyb:callGreedyPeel} that is obtained by running \algkapp on $W_{J+1}$. As for the proof of Lemma~\ref{lemma:ratioSampling} we assume that at each iteration of the $J$ iterations by using approximate temporal-motif degrees it holds that ``$\estimateTemporalDegreeWeight{W_i}{w_j}\in (1\pm \varepsilon) \temporalDegreeWeight{W_i}{w_j}$'' for each $w_j\in W_i$ and $i=1,\dots,J$. Which also implies that $\hat{\tau}(W_j)/|W_j| \in (1\pm \varepsilon) \tau(W_j)/|W_j| $. Now let us consider the first node $v\in W^*$ that gets removed across the iterations of the algorithm, then there may be two distinct cases (I). $v\in W_i$ for some $i=1,\dots, J$; and (II) or $v\in W_{J+1}$. (I) In this case the algorithm returns a solution which attains the same approximation quality as in Lemma~\ref{lemma:ratioSampling} since it is easy to see that $W'$ (i.e., the subnetwork obtain by \suitegreedy)  cannot provide a much worse solution than the best solution observed by the algorithm during the $J$ iterations using approximate temporal motif degrees. Let $W_a$ be the solution achieving the $\tfrac{(1-\varepsilon)}{k(1+\xi)(1+\varepsilon)}$-approximation factor (recall that $v\in W^*$ and $v\in W_i$, for $i\in [J]$ so such a solution must exists). Let $W_b$ be the best estimated solution by the algorithm i.e., $W_b = \argmax\{\hat{\tau}(W_i)/|W_i|, \fobj(W')\}$ returned by \suitehybrid, then it is easy to see that $\fobj(W')$ cannot be less than $(1-\varepsilon)\fobj(W_a)$ since otherwise the algorithm will pick $\hat{\tau}(W_a)/|W_a|$ as optimal solution, therefore the worst case still remains the one as in Lemma~\ref{lemma:ratioSampling}  leading to the  $\tfrac{(1-\varepsilon)^2}{k(1+\xi)(1+\varepsilon)^2}$-approximation factor. (II) If for $v\in W^*$, it holds $v\in W_{J+1}$ then it follows that $\fobj(W')\ge 1/k \fobj(W^*)$ from Lemma~\ref{lemma:kapproxratio}, note that the algorithm may return a solution $W_j$ for some $j\in [J]$ such that $\fobj(W_j)\le \fobj(W')$ but $\hat{\tau}(W_j)/|W_j| \ge \fobj(W')$, clearly for such a solution it holds $\fobj(W_j) \ge \tfrac{\rho(W^*)}{k(1+\varepsilon)}$. 
	Therefore the approximation-ratio of the algorithm is bounded by $\min\left\{\tfrac{(1-\varepsilon)^2}{k(1+\xi)(1+\varepsilon)^2}, \tfrac{1}{k(1+\varepsilon)} \right\}$. Finally we need to take the union bound over the $J$ iterations to guarantee that the failure probability when running the randomized batch-peeling is bounded by $\eta\in(0,1)$. Note that when $J=1$ for case (I) the approximation ratio is bounded by $\tfrac{(1-\varepsilon)^2}{k(1+\xi)(1+\varepsilon)}$ since~\suitegreedy evaluates $\fobj(\cdot)$ exactly on $W'$ and not approximately, concluding therefore the claim.
\end{proof}

\begin{proof}[Proof of Lemma~\ref{lemma:staticvstemporal}]
	First fix $n$ arbitrary large and also assume $n/12\in \mathbb{N}$ to keep the notation simple. Let $T = T_1 \cup T_2$ be the union of two temporal networks with similar structure, the first one $T_1=(V_1,E_1)$ where $V_1=\{v_1,\dots,v_{3n/4}\}$ and $E_1=\{(v_i,v_j,2): i\in[1,n/4], j\in[n/4+1,n/2]\}\cup\{(v_i,v_j,1):i\in[n/4+1,n/2], j\in[n/2+1,3n/4]\}$ and the other network is $T_2 =(V_2,E_2)$ where $V_2=\{v'_{1},\dots, v'_{n/4}\}$ and $E_2=\{(v'_i,v'_j,1): i\in[1,n/12], j\in[n/12+1,n/6]\}\cup\{(v'_i,v'_j,2): i\in[n/12+1,n/6], j\in[n/6+1,n/4]\}$. See Fig.~\ref{fig:lemmastaticvstemporal} for a visualization of the resulting network $T$.
	Let $D_{T}$ be the directed static network associated to $T$, then the solutions of S2DS on $V_1$ and $V_2$ are respectively $(n/4)^3/(3n/4)$ and $(n/12)^3/(n/4)$ hence the optimal solution is achieved by $H^*=V_1$. But when we consider  \ourproblem problem on $T$ on a 2-path temporal motif for $\delta$ fixed at least to 1, $V_1$ has a solution whose value is $0$ (i.e., $\fobj(H^*)=0$), since none of its temporal paths are time-respecting (there are no instances of 2-paths in $T_1$), while $W^*=V_2$ achieves a solution $\fobj(W^*)= \bigO(n^2)$. Finally note that $H^*\cap W^* = \emptyset$.
\end{proof}

\section{Baseline algorithms}\label{appsec:baselines}
In the next sections we describe, prove the guarantees and give bounds on the time complexity of each baseline algorithm considered, we recall that a summary of such baselines is given in Table~\ref{tab:ALDENTEalgs}.

\subsection{\suiteex}\label{subsec:exact}
We first present \suiteex, an exact algorithm for the \ourproblem problem, 
which solves multiple $(s,z)$-min-cut instances 
on a flow network encoding weights of $\delta$-instances in $\mathcal{S}_T$, leveraging $\mathcal{H}$ (which we briefly mentioned in Section~\ref{sec:temporalAndBaselines}). 
This algorithm extends ideas developed for the exact solution of $k$-vertex motif and triangle (or $k$-clique) 
densest-subgraph problems in static networks~\cite{Goldberg:CSD-84-171, Tsourakakis2014, Wang2020, Tsourakakis2015, Mitzenmacher2015, Fang2019}.

A flow network $Z=(V_Z,E_Z)$ is a directed network with two special vertices $s\in V_Z$ (source) and $z\in V_Z$ (sink) and a set of intermediate vertices, such that the source has edges of the form $(s,\cdot)$ and the sink is connected by edges of $(\cdot, z)$ where ``$\cdot$'' represents an intermediate vertex. Additionally, each edge $e\in E_Z$ over the network is assigned with a weight $w(e)$, and a $(s,z)$-min-cut on such network is a partition $(\mathsf{S},\mathsf{Z})$ of the vertices $V_Z$ such that $s\in \mathsf{S}$ and $z\in \mathsf{Z}$, and the cost of the cut (defined as sum of the weights of the edges going from vertices in $\mathsf{S}$ to vertices in $\mathsf{Z}$) is minimized.

We first recall some definitions, starting from graph isomorphism:
we say that $G_1=(V_1,E_1)$ and $G_2=(V_2,E_2)$ are \emph{isomorphic} (denoted by $G_1\!\sim\!G_2$) 
if there exists a bijection $g$ on the vertices such that $e_1=(x,y)\in E_1$ 
if and only if $e_2=(g(x),g(y))\in E_2$. 

Next we recall the construction of the set $\mathcal{H}$, that is given a parameter $k\in\mathbb{N}$, we defined $\mathcal{H} = \{H :H=(V_H, E_H) \subseteq G_{T}, \text{ for some } H' \subseteq H \text{ it holds } H' \!\sim\! G[M], \tau(V_H)>0, \text{ and } |V_H|=k\}$ to be the set of 
$k$-connected induced subgraphs ($k$-CIS) from $G_{T}$, where each $k$-CIS $H$ encodes a subgraph containing at least a $\delta$-instance $S$ with $\tau(S)>0$,
and is induced in~$G_{T}$,
i.e., it contains all edges among the vertices~$V_H$. 
Note that this condition is ensured by requiring the existence of $H'\subseteq H$ such that $H'\!\sim\! G[M]$ and $\tau(V_H)>0$, where $G[M]$ is the undirected graph associated to the temporal motif $M$. 
Given a $k$-vertex $\ell$-edge temporal motif $M$ and a temporal network $T=(V_T,E_T)$, let us consider the weighted flow network $Z = (V_Z = \{s\}\cup V_\mathcal{H} \cup V_T \cup \{z\}, E_Z)$, 
with $w_{Z,\zeta}: E_Z \mapsto \mathbb{R}_0^+$,~where:
\begin{itemize}
	\item 
	each vertex $v_H\in V_\mathcal{H}$ is associated to a  $k$-CIS $H\in\mathcal{H}$;
	\item $s$ and $z$ are the source and sink vertices of the flow network;
	\item $E_Z=(E_1\cup E_2 \cup E_3)$ where $E_1=\{(s,v_H): v_H\in V_\mathcal{H}\}$, 
	$E_2=\{(v_{H}, v): v_H \in V_\mathcal{H}, v\in V_T, v\in V_H\}$, 
	$E_3=\{(v,z): v\in V_T\}$, i.e., 
	$E_1$ connects the source $s$ to all vertices in $V_\mathcal{H}$, 
	$E_2$ connects each vertex $v_H\in V_\mathcal{H}$ representing a $k$-CIS $H=(V_H,E_H)\in \mathcal{H}$ to its vertices in $V_H\subseteq V_T$, and 
	$E_3$ connects each vertex of $V_T$ to the sink $z$;
	\item the weighting function is defined as follows, given $\zeta \ge 0$:
	\[
	w_{Z,\zeta}(e) =
	\begin{cases}
		\tau(V_H)= \sum_{S\in \mathcal{S}_{H}}\tau(S) & \text{if }e=(s,v_{H})\in E_1, \\
		+\infty & \text{if } e\in E_2, \text{ and}\\
		\zeta & \text{if } e=(v, z)\in E_3.
	\end{cases}
	\]
\end{itemize}

We show how the optimal solution $W^*$ to the \ourproblem~problem is computed in Algorithm~\ref{alg:MinCutExact}. 
The algorithm first enumerates all $\delta$-instances in the temporal network $T$ (line~\ref{alglineExact:enumerateInst}) and builds the flow network $Z$ described above (line~\ref{alglineExact:buildZ}).
It then instantiates $a,b$, i.e., 
the range in which we will seek for the optimal values of the solution to Problem~\ref{probl:denseMotif}, 
$\tau_{\min}$ and $\zeta$ used for finding the optimal values through binary search 
(lines~\ref{alglineExact:initab}-\ref{alglineExact:initZeta}). 
Then the algorithm starts its main loop (line~\ref{alglineExact:mainwhile}) seeking at each iteration for the optimal solution $W^*$ by repeatedly solving the $(s,z)$-min cut on the flow network $Z$ with weights $w_{Z,\zeta}$, and updating $\zeta$ accordingly at each iteration, while maintaining the candidate optimal solution (line~\ref{alglineExact:updateSol}). 

We show in the next lemma that Algorithm~\ref{alg:MinCutExact} correctly identifies the optimal density value $\fobj^*\mathbin= \fobj(W^*)=\tau(W^*)/|W^*|$ to Problem~\ref{probl:denseMotif} therefore returning the optimal solution $W^*$.

The proof follows similar arguments as in previous work~\cite{Wang2020, Tsourakakis2015, Mitzenmacher2015}. 
Some of details are not immediately transferable from previous proofs, 
such as the connection to the solution on the temporal network $T$, 
therefore, we present in detail the proof below. 
\begin{lemma}\label{lemma:correctExact}
	Given a temporal network $T=(V_T,E_T)$, a temporal motif $M$, and $\delta>0$, let $\fobj^*$ be the optimal value for the solution of \ourproblem problem with weighting function $\tau$. Let $(\mathsf{S}, \mathsf{Z})$ be the $(s,z)$-min cut on $Z$ with weighting function $w_{Z,\zeta}, \zeta \ge 0$ at an arbitrary iteration of Algorithm~\ref{alg:MinCutExact}. If $V'=\mathsf{S}\cap V_T \neq \emptyset$ then $\zeta \le \rho^*$ else $\zeta \ge \fobj^*$.
\end{lemma}

\begin{algorithm}[t]
	\caption{\suiteex}\label{alg:MinCutExact}
	\KwIn{$T=(V_T,E_T), M , \delta\in \mathbb{R}^+, \tau$.}
	\KwOut{$W^*$ optimal solution to Problem~\ref{probl:denseMotif}.}
	$\mathcal{S} \gets \{S: S \text{ is } \delta\text{-instance of } M \text{ in } T\}$\label{alglineExact:enumerateInst}\;
	\lIf{$\mathcal{S} = \emptyset$}{\KwRet{$\emptyset$}}
	$Z \gets \mathtt{BuildGraph}(T,\mathcal{S})$\label{alglineExact:buildZ}\;
	$a\gets 0; b \gets \sum_{S\in \mathcal{S}} \tau(S);\tau_{\min}\gets\min_{S\in\mathcal{S}}\{\tau(S)\}$\label{alglineExact:initab}\;
	$\zeta \gets (b+a)/2$\label{alglineExact:initZeta}\;
	\While{$b-a \ge \frac{\tau_{\min}}{n(n-1)}$\label{alglineExact:mainwhile}}{
		$(\mathsf{S},\mathsf{Z}) \gets \mathtt{MinCut}(s,z,Z,w_{Z,\zeta})$\label{alglineExact:mincut}\;
		\lIf{$\mathsf{S} = \{s\}$}{$b\gets \zeta$\label{alglineExact:onlyS}}
		\Else{
			$a\gets \zeta$;
			$V' \gets \mathsf{S} \cap V_T$\label{alglineExact:updateSol}\;
		}
		$\zeta \gets (b+a)/2$\label{alglineExact:biseczeta}\;
	}
	\KwRet{$W^* \gets  V'$\label{alglineExact:returnOPT}}\;
\end{algorithm}

\begin{proof}
	First of all we note that every min cut $(\mathsf{S},\mathsf{Z})$ should be finite, since trivially the cut $(\{s\}, V_Z\setminus \{s\})$ has a value of the cut $\tau(V_T)$\footnote{Recall that for a directed and weighted flow network, the cost of the cut $(\mathsf{S},\mathsf{Z})$ is the sum of the weights on the outgoing edges from $\mathsf{S}$ to $\mathsf{Z}$.} and such value is finite. Fix a $(s,z)$-min cut $(\mathsf{S},\mathsf{Z})$ then let $V'=\mathsf{S}\cap V_T$. First of all note that if $v_{H=(V_H,E_H)} \in \mathsf{S}, v_H\in V_{\mathcal{H}}$ then $v\in V' $, for all $v \in V_H\subseteq V_T$ and therefore $\mathcal{S}_{H} = \mathcal{S}_{V_H} \subseteq\mathcal{S}_{V'}$ in $T$, since every minimum cut is finite.
	
	Let us consider a finite min-cut $(\mathsf{S}',\mathsf{Z}')$. Let $V' = \mathsf{S}' \cap V_T$ and let $G[{V'}]$ be the induced subgraph of $V'$ in $G_{T}$ and $T[V']$ the induced subnetwork of $V'$ in $T$. Let $C$ denote the cost of the cut, i.e., sum of the outgoing weighted edges from $\mathsf{S}'$ to $\mathsf{Z}'$.  Since the minimum cut can only be finite, we have the following:
	\begin{align}
		C(\mathsf{S}',\mathsf{Z}') 
		& = C(\{s\}, V_\mathcal{H}\cap \mathsf{Z}')+ C(\mathsf{S}'\cap V_T, \{t\})\nonumber\\
		& = \sum_{v_H\in V_\mathcal{H}\cap \mathsf{Z}'} \tau(V_H) + 
		\sum_{v\in \mathsf{S}'\cap V_T}  \zeta  \label{eq:zetaIneq-a} \\
		& = \sum_{v_H\in V_\mathcal{H}} \tau(V_H) - 
		\sum_{v_H\in V_\mathcal{H} \cap \mathsf{S}'} \left(\tau(V_H) - \zeta |V'|\right)\label{eq:zetaIneq-b}\\
		& = \tau(V_T) - \left(\tau(V') - \zeta |V'|\right) \label{eq:zetaIneq}
	\end{align}
	In the above derivation,
	Equation~(\ref{eq:zetaIneq-a}) follows from the definition of the cost of a cut, and the fact that if $v_H\in \mathsf{S}'$ then each vertex $v\in V_T$ in $V_H$ is also in $\mathsf{S}'$ since the min-cut is finite.
	Equation~(\ref{eq:zetaIneq-b}) follows from the definition of  $V'$ and the fact that 
	$V_\mathcal{H} = (V_\mathcal{H}\cap \mathsf{Z}') \cup (V_\mathcal{H} \cap \mathsf{S}')$. 
	To see why Equation~(\ref{eq:zetaIneq}) holds, 
	let $V_\mathcal{H}' \subseteq V_\mathcal{H}$ and we have
	\begin{equation}\label{eq:countHInst}
		\!\!\!\sum_{v_{H =(V_H,E_H)}\in V_\mathcal{H}'}\!\!\! \tau(V_H) = 
		\sum_{v_H \in V_\mathcal{H}'} \sum_{S\in \mathcal{S}_{V_H}}\tau(S) = 
		\!\!\!\sum_{S\in \mathcal{S}_{\cup_{v_H\in V_\mathcal{H}'} V_H}} \!\!\!\tau(S).
	\end{equation}
	Note that the above holds since each vertex in $V_\mathcal{H}$ represents a distinct $k$-CIS, i.e., the same vertex cannot represent two identical vertex sets in the original network $G_T$. 
	Given Equation~\eqref{eq:countHInst}, 
	we now obtain Equation~\eqref{eq:zetaIneq} by combining
	$\mathcal{S}_{\cup_{v_H\in V_\mathcal{H}} V_H} = \mathcal{S}_{V_T}$, 
	the definition of $V'$, and the fact that given $V_H\subseteq V'$ there exists $v_H\in V_\mathcal{H}$ 
	such that $v_H$ is in $V_\mathcal{H}\cap \mathsf{S}'$ since the cut is optimal, and conversely, 
	if a given vertex $v_H$ is in $V_\mathcal{H}\cap \mathsf{S}'$ then $V_H \subseteq V'$ 
	by the finiteness of the minimum cut.
	
	We are now ready to show that if $V'\neq \emptyset$ then $\tau(V') - \zeta |V'| \ge 0$.
	If $\tau(V') - \zeta |V'| < 0$, then by Equation~\eqref{eq:zetaIneq} 
	this min-cut would be larger than the cut $C(\{s\}, V_Z\setminus \{s\})$ contradicting its optimality. 
	Hence it holds that $\zeta \le \frac{\tau(V') }{|V'|} \le \fobj^*$.
	
	Conversely, if $V'= \emptyset$ then we show that $\tau(W) - \zeta |W| \le 0$ for each $W\subseteq V_T$.
	Note that this is trivial if $V_\mathcal{H}\cap \mathsf{S}' = \emptyset$, so we will assume that the set $V_\mathcal{H}\cap \mathsf{S}'$ is not empty. 
	Let for contradiction $\tau(W) - \zeta |W| > 0$ for  a certain $W\subseteq V_T$.
	Let $V_\mathcal{H}[W]$ be the set of vertices in $V_\mathcal{H}$ corresponding to $k$-CIS 
	between vertices of $W$ in $G_T$ with at least one $\delta$-instance $S\in \mathcal{S}_W$, $\tau(S) >0$, 
	for which it holds $V_\mathcal{H} \cap \mathsf{S}' = V_\mathcal{H}[W]$.
	Consider also the cut $(\mathsf{S},\mathsf{Z})= (\{s\}\cup W\cup V_\mathcal{H}[W], V_Z\setminus\{\{s\}\cup W \cup V_\mathcal{H}[W]\} )$.
	Since this cut is finite, by applying Equation~\eqref{eq:zetaIneq} we get
	\[
	C(\mathsf{S},\mathsf{Z}) = \tau(V_T) - \left(\tau(W) - \zeta |W|\right) < \tau(V_T),
	\]
	obtaining a value of the cut that is smaller than the cut $(\mathsf{S}',\mathsf{Z}')$, 
	that is $\tau(V_T)$ when $V'=\emptyset$.
	This contradicts the initial assumption that the min-cut ($\mathsf{S}',\mathsf{Z}'$) is optimal. 
	Hence, we have that $\zeta \ge \frac{\tau(W)}{|W|}$ for all $W\subseteq V_T$, 
	and hence $\zeta \ge \fobj^*$.
\end{proof}
Lemma~\ref{lemma:correctExact} proves that Algorithm~\ref{alg:MinCutExact} correctly identifies the optimal solution $W^*$ by finding the optimal density $\fobj^*$ through binary search. Next we show how to determine when to stop the search.

\begin{lemma}
	Given a temporal network $T=(V_T,E_T)$, let $\tau_{\min} = \min_{S\in\mathcal{S}}\{\tau(S)\}$. Given two distinct and non-empty vertex sets $V_1,V_2\subseteq V_T$ such that $\fobj(V_1)\neq \fobj(V_2)$, then it holds that
	\[
	\left| \frac{\tau(V_1)}{|V_1|} - \frac{\tau(V_2)}{|V_2|} \right| \ge \frac{\tau_{\min}}{n(n-1)}.
	\]
	\label{lemma:termination}
\end{lemma}
\begin{proof}
	If $|V_1|=|V_2|$, 
	it easy to see that $|\tau(V_1)/|V_1| - \tau(V_2)/|V_2 | |\ge \tau_{\min}/n\ge \tau_{\min}/(n(n-1))$, for $n\ge2$, otherwise it holds $|\tau(V_1)/|V_1| - \tau(V_2)/|V_2| |\ge \tau_{\min}/(|V_1||V_2|)\ge \tau_{\min}/(n(n-1))$, for $n\ge2$. Therefore, the claim follows.
\end{proof}

By Lemma \ref{lemma:termination}, Algorithm~\ref{alg:MinCutExact} terminates in at most $\bigO(\log{n}+\log(\tau(V_T)/\tau_{\min}))$ steps of binary search, 
since it explores the range $[0,\zeta =\tau(V_T)]$, 
which is halved at each step.\footnote{Some optimization according to~\citet{Fang2019} can be used to reduce this range but it cannot improve the overall worst-case complexity.}
\ifextended 
In fact, let $x$ be the number of steps until termination, the claimed bound is achieved by solving for $x$ the following equation
\[
\frac{\tau(V_T)}{2^x} = \bigO\!\left(\frac{\tau_{\min}}{n^2}\right). 
\]
\fi

\vspace{1mm}
\para{Time complexity.} 
The execution time of Algorithm~\ref{alg:MinCutExact} 
is bounded by 
$\bigO(T_{\mathrm{enum}}+ T_{\mathrm{flow}} (\log(n) + \log(\tau(V)/\tau_{\min})))$, 
where $T_{\mathrm{enum}}$ is the running time of enumerating all the $\delta$-instances of a temporal motif $M$.\footnote{We assume that $\tau(S)$ can be computed in at most $\bigO(\ell)$ for each given $\delta$-instance $S$.} 
By using Mackey's algorithm~\cite{Mackey2018} it holds that 
$T_{\mathrm{enum}}  = \bigO(m\hat{m}^{\ell-1})$, where $\hat{m}$ is the maximum number of temporal edges of the temporal network $T$ in a time-interval length $\delta$. 
The cost of each min-cut solution can be computed in time $\bigO(n_Zm_Z)$, 
where $Z$ is the flow network on which the solution is found.\footnote{More advanced algorithms may require smaller running time.} 
Therefore, $n_Z = \bigO(n+|\mathcal{H}|)$ and $m_Z = 
\bigO(k|\mathcal{H}| + n)$; recall also that $|\mathcal{H}|=\bigO(n^k)$. 
Hence the cost of computing  a single min cut is $\bigO(k|\mathcal{H}|(n+|\mathcal{H}|) + n^2 )$. 
Therefore, the final complexity of the algorithm is 
$\bigO(m\hat{m}^{\ell-1} +(k|\mathcal{H}|(n+|\mathcal{H}|) + n^2) (\log(n) + \log(\tau(V_T)/\tau_{\min})))$, 
which in general is prohibitive in most applications.

\subsection{\suitegreedy}\label{subsec:greedy}



In this section we present \suitegreedy, a $1/k$-approximation algorithm for the  \ourproblem problem. In particular,  Algorithm~\ref{alg:kapprox}  is based on an extension of Charikar's algorithm for the edge densest-subgraph problem~\cite{Charikar2000}. \algkapp removes vertices one at a time according to the minimum temporal motif degree in the subset of vertices being examined. This will be done by first computing the set $\mathcal{H}$, and then maintaining such set updated over the peeling procedure, as we show next.

\vspace{1mm}
\para{Algorithm description.} Given a temporal network $T=(V_T,E_T)$ a temporal motif $M$, a value of $\delta$ and the weighting function $\tau$, let $v\in V_T$ let $\mathcal{S}_{V_T}(v)\mathbin= \{S: S \in \mathcal{S}_{V_T} \text{ and } v\in S \}$, that is the set of $\delta$-instances in the subnetwork induced by $V_T$ in which $v\in V$ participates. Therefore, the temporal motif degree (under weighting function $\tau$) of vertex $v\in V$ is $\temporalDegreeWeight{V_T}{v} = \sum_{S \in \mathcal{S}_{V_T}(v)}\tau(S)$, and we write $\temporalDegreeWeight{V_T}{v}$ when $M$ and $\delta$ are clear from the context.
We are ready to describe Algorithm~\ref{alg:kapprox}, which first enumerates all $\delta$-instances of the temporal motif $M$ in $T$ (line~\ref{alglinekapp:enumerateInst}) and then initializes two sets $W_n$ and $\mathcal{S}_n$, 
maintaining respectively the vertex set of the network and its corresponding set of 
$\delta$-instances (line \ref{alglinekapp:initSol}). 
Entering the main loop (line~\ref{alglinekapp:mainLoop}), 
at iteration $i$ the algorithm removes the vertex with the minimum temporal motif degree in the subnetwork induced by $W_i$, updating the corresponding sets of remaining vertices and $\delta$-instances (lines~\ref{alglinekapp:getMinNode}-\ref{alglinekapp:updateNodeSet}), and this is done by properly encoding the set $\mathcal{H}$ (see the time complexity analysis for more details). The algorithm returns as solution the set that maximizes the \ourproblem objective function among all $n-k$ vertex sets examined. 
The following result establishes the approximation ratio of Algorithm~\ref{alg:kapprox}.

\begin{algorithm}[t]
	\caption{\suitegreedy}\label{alg:kapprox}
	\KwIn{$T=(V_T,E_T), M , \delta\in \mathbb{R}^+, \tau$.}
	\KwOut{$W$.}
	$\mathcal{S} \gets \{S: S \text{ is } \delta\text{-instance of } M \text{ in } T\}$\label{alglinekapp:enumerateInst}\;
	\lIf{$\mathcal{S} = \emptyset$}{\KwRet{$\emptyset$}}
	$\mathcal{S}_n \gets \mathcal{S}$; $W_n \gets V$\label{alglinekapp:initSol}\;
	\For{$i\gets n,\dots,k+1$\label{alglinekapp:mainLoop}}{
		$v_{\min}\gets \argmin_{v\in W_i} \{\temporalDegreeWeight{W_i}{v}\}$\label{alglinekapp:getMinNode}\;
		$\mathcal{S}_{i-1} \gets \mathcal{S}_i \setminus \{S : v_{\min}\in S, S \in \mathcal{S}_i\}$\label{alglinekapp:updateInst}\;
		$W_{i-1}\gets W_i\setminus \{v_{\min}\}$\label{alglinekapp:updateNodeSet}\;
	}
	\KwRet{$W \gets  \argmax_{i=k,\dots,n} \{ \fobj(W_i)\}$\label{alglinekapp:retMaxSol}}\;
\end{algorithm}

\begin{lemma}\label{lemma:kapproxratio}
	Algorithm~\ref{alg:kapprox} is a $1/k$-approximation algorithm for the \ourproblem problem.
\end{lemma}
\begin{proof}
	Let $W^*$ be the optimal solution for a given instance to the \ourproblem problem.\footnote{Recall that we write $v\in S$ to indicate that $v$ participates in at least one edge 
		of the sequence of edges representing $S$.}  
	It is easy to see that 
	$\tau(W^*)/|W^*| = 1/(k|W^*|) \allowbreak \sum_{v\in W^*} \temporalDegreeWeight{W^*}{v}$ 
	since each $\delta$-instance is weighted exactly $k$ times. 
	For a vertex $v\in W^*$ it is
	\begin{equation}
		\frac{\tau(W^*)}{|W^*|} = \frac{\sum_{v'\in W^*} \temporalDegreeWeight{W^*}{v'}}{k|W^*|} \ge \frac{\sum_{v'\in W^*\setminus \{v\}} \temporalDegreeWeight{W^*\setminus\{v\}}{v'}}{k(|W^*|-1)}.
	\end{equation}
	The above inequality follows from the optimality of $W^*$. 
	Combining the above with the fact that $\sum_{v'\in W^*\setminus \{v\}} \temporalDegreeWeight{W^*\setminus\{v\}}{v'} + k\temporalDegreeWeight{W^*}{v} = k\tau(W^*) $ we get that $\temporalDegreeWeight{W^*}{v}\ge \tau(W^*)/|W^*|$. 
	Let us consider the iteration before Algorithm~\ref{alg:kapprox} removes the first vertex $v\in W^*$.
	Let $W$ denote the vertex set in such an iteration.
	Then $W^*\subseteq W$ and additionally we know that given $w\in W$ it holds  
	$\temporalDegreeWeight{W}{w}\ge \temporalDegreeWeight{W}{v}\ge\temporalDegreeWeight{W^*}{v}\ge \tau(W^*)/|W^*|$. 
	Hence,
	\[
	\tau(W)= \frac{1}{k} \sum_{w\in W} \temporalDegreeWeight{W}{w} \ge \frac{|W|\tau(W^*)}{k|W^*|} \Rightarrow \frac{\tau(W)}{|W|} \ge \frac{\tau(W^*)}{k|W^*|}.
	\]
	Last, note that the algorithm will return a solution $W'$ for which it holds 
	$\frac{\tau(W')}{|W'|} \ge \frac{\tau(W)}{|W|} $, which concludes the proof.
\end{proof}

\vspace{1mm}
\para{Time complexity.} 
Note that Algorithm~\ref{alg:kapprox} requires an enumeration of \emph{all} $\delta$-instances in the temporal network $T$, which takes exponential time.
To provide an efficient implementation 
we compute for each vertex $v\in V_T$ the 
temporal motif degree $\temporalDegreeWeight{V_T}{v}$ 
and we build a Fibonacci min-heap over these counts (taking $\bigO(n\log n)$ for build up), 
with amortized update cost of $\bigO(1)$, 
which we use to retrieve the vertex to be removed at each iteration in $\bigO(\log n)$ time. 
Let us define $\mathcal{H}(v) \mathbin=  \{H=(V_H,E_H): v\in V_H, H\in \mathcal{H}\}\subseteq \mathcal{H}$. Additionally, for each vertex $v\in V_T$ we maintain the set $\mathcal{H}(v)$ of the static $k$-CIS in $G_T$, i.e., subgraphs $H_1\dots, H_{|\mathcal{H}(v)|}$ for which there exists at least one $\delta$-instance with non-zero weight between vertices in $V_{H_i}$, where $v\in V_{H_i}$ for each $i=1,\dots,|\mathcal{H}(v)|$. 

By using the Fibonacci min-heap data structure, vertices can be removed by avoiding the re-calculation of $\delta$-instances, and we only update the vertices alive and their instances, i.e., removing a vertex $v$ requires for each $k$-CIS $H_i$, for $i\in [1,|\mathcal{H}(v)|]$, to be removed from the lists of all the other $k-1$ vertices involved in each $k$-CIS, requiring $\bigO(k|\mathcal{H}(v)|\log(|\mathcal{H}_{\max}|))$ time where 
$|\mathcal{H}_{\max}| = \max_{v\in V_T}\{|\mathcal{H}(v)|\}$. 
Therefore, the final time complexity is
\begin{align*}
	\ifextended
	&\bigO\!\left(m\hat{m}^{(\ell-1)} +   n\right. \left.\log n +  \sum_{i=1}^n (\log n + k|\mathcal{H}(v)|\log(|\mathcal{H}_{\max}|))\right)=\\
	&\bigO\!\left(m\hat{m}^{(\ell-1)} + n\log n + k^2|\mathcal{H}|\log(|\mathcal{H}_{\max}|)\right).
	\else
	\bigO\!\left(m\hat{m}^{(\ell-1)} + n\log n + k^2|\mathcal{H}|\log(|\mathcal{H}_{\max}|)\right).
	\fi
\end{align*}
Note that the time complexity is significantly smaller than the one of Algorithm~\ref{alg:MinCutExact},
but may still be impractical on large datasets, 
due to the exact enumeration of $\delta$-instances, and since 
$|\mathcal{H}| = \bigO(n^k)$.

\subsection{\suitebatch}\label{subsec:batch}

\begin{algorithm}[t]
	\caption{\suitebatch}\label{alg:batchpeel}
	\KwIn{$T=(V_T,E_T), M , \delta\in \mathbb{R}^+, \tau, \xi>0$.}
	\KwOut{$W$.}
	$\mathcal{S} \gets \{S: S \text{ is } \delta\text{-instance of } M \text{ in } T\}$\label{algbatchPeel:enumInst}\;
	\lIf{$\mathcal{S} = \emptyset$}{\KwRet{$\emptyset$}}
	$W,W' \gets V$\label{algbatchPeel:initNodeSets}\;
	\While{$W' \neq \emptyset$\label{algbatchPeel:mainLoop}}{
		$L(W')\gets\{v\in W' : \temporalDegreeWeight{W'}{v} \le k(1+\xi) \tau(W')/|W'|\}$\label{algbatchPeel:smallNodes}\;
		$W' \gets W'\setminus L(W')$\label{algbatchPeel:updateNodeSet}\;
		\If{$\fobj(W') \ge \fobj(W)$\label{algbatchPeel:updateBestSol}}
		{
			$W\gets W'$\;
		}
	}
	\KwRet{$W$.}
\end{algorithm}

In this section we show how to speed up Algorithm~\ref{alg:kapprox}
for the price of a slightly weaker approximation guarantee, to obtain our last baseline. 
This technique is an important building block for both \algrand and \alghybrid. 
The main idea is similar to existing techniques  in literature~\cite{Bahmani2012, Epasto2015, Tsourakakis2014}:
instead of removing one vertex at a time, 
remove them in \emph{batches} according to their temporal motif degree, 
so that the algorithm terminates in fewer iterations with respect to \algkapp. Additionally, the algorithms leverages a similar approach to \algkapp for using the set $\mathcal{H}$, without the need of storing the temporal motif degrees in a min heap. 
The number of iterations, as well as the approximation ratio, are controlled by a parameter $\xi>0$. 


\vspace{1mm}
\para{Algorithm description.} \suitebatch is shown in detail in Algorithm~\ref{alg:batchpeel}. 
It takes as input all the parameters to the \ourproblem problem and an additional parameter $\xi>0$, and after enumerating all $\delta$-instances of $M$ in $T$ 
(line~\ref{algbatchPeel:enumInst}), 
at each round it peels the set $L(W')$ where $W'$ is the vertex set active in that iteration.
The vertices $v\in L(W')$ are such that their temporal motif degree $\temporalDegreeWeight{W'}{v}$ (in the subnetwork induced by $W'$) is small compared to a threshold controlled by $\xi$ 
(lines~\ref{algbatchPeel:smallNodes}-\ref{algbatchPeel:updateNodeSet}). 
The algorithm returns the best solution encountered, 
according to the objective function $\fobj(\cdot)$ of Problem~\ref{probl:denseMotif} 
(line~\ref{algbatchPeel:updateBestSol}). 
The following result establishes the approximation ratio of Algorithm~\ref{alg:batchpeel}.

\begin{lemma}\label{lemma:appxbatchpeel}
	Algorithm~\ref{alg:batchpeel} is a $\frac{1}{k(1+\xi)}$-approximation algorithm for the \ourproblem~problem.
\end{lemma}
\begin{proof}
	First observe that at each iteration the set $L(W')$ is non-empty.
	Indeed, if we assume $\temporalDegreeWeight{W'}{v} > k(1+\xi)\tau(W')/|W'|$, 
	for each  $v \in W'$, 
	then $k\tau(W') = \sum_{v\in W'}\temporalDegreeWeight{W'}{v} > |W'|k(1+\xi)\allowbreak\tau(W')/|W'|$,
	hence $k>k(1+\xi)$, with $\xi > 0$, leading to a contradiction. 
	Let us consider the  optimal solution $W^*$, and the iteration when the first vertex $v\in W^*$ is added to $L(W')$. Note that $W^*\subseteq W'$. We know that for each $w \in L(W')$ it holds $\temporalDegreeWeight{W'}{w} \le k(1+\xi)\tau(W')/|W'|$ hence
	\[
	k(1+\xi)\frac{\tau(W')}{|W'|} \ge \temporalDegreeWeight{W'}{v} \ge  \temporalDegreeWeight{W^*}{v} \ge \frac{\tau(W^*)}{|W^*|}.
	\]
	Where the last step is obtained from the first part of the proof of Lemma~\ref{lemma:kapproxratio}. The proof concludes by the fact that Algorithm~\ref{alg:batchpeel} will output a solution $W$ for which $\fobj(W) \ge \fobj(W')$.
\end{proof}

\vspace{1mm}
\para{Time complexity.} Regarding its time complexity, we have:
\begin{lemma}\label{lemma:batchPeelIt}
	Algorithm~\ref{alg:batchpeel} performs at most $\bigO(\log_{1+\xi}(n))$ iterations.
\end{lemma}
\begin{proof}
	We need to bound the number of vertices for which $W'$ is decreased at each iteration. 
	Observe that
	\[
	k\tau(W') \ge \!\!\!\sum_{v\in W'\setminus L(W')}\!\!\! \temporalDegreeWeight{W'}{v} \ge k(1+\xi) \frac{\tau(W')}{|W'|} (|W'| - |L(W')| ),
	\]
	and therefore,
	\[
	|L(W')| \ge \frac{\xi|W'|}{(1+\xi)} \Leftrightarrow |W'\setminus L(W')| \le \frac{|W'|}{(1+\xi)}.
	\]
	This means that the set $W'$ is decreased by a factor of at least $1/(1+\xi)$ at each iteration.
	Let $x$ denote the number of iterations until the termination condition is reached and note that $W'$ starts from $V$, hence $|W'|=n$ at the beginning of the first iteration.
	Solving for $x$ in the equation
	\[
	n \frac{1}{(1+\xi)^x} = \bigO(1),
	\]
	proves the claim.
\end{proof}
We  now provide a bound on the time complexity of Algorithm~\ref{alg:batchpeel}. 
First observe that we do not need to use a Fibonacci heap to store weighted counts of the different instances for the graph vertices, since we are peeling the vertices in batches. 
Instead, we need to keep track of the $k$-CIS (and hence the set $\mathcal{H}$, and $\mathcal{H}(v), v\in V_T$) in $G_T$ containing each vertex, 
as in Algorithm~\ref{alg:kapprox}. Therefore, the running time can be bounded with a similar analysis to~\suitegreedy by
\begin{align*}
	&\bigO\!\left(m\hat{m}^{(\ell-1)} + \frac{(1+\xi) n}{\xi} + k^2|\mathcal{H}|\log(|\mathcal{H}_{\max}|)\right).
\end{align*}

As with Algorithms~\ref{alg:MinCutExact} and~\ref{alg:kapprox} the time complexity is prohibitive in many applications, due to the exact enumeration of $\delta$-instances and the dependence on the size of the set $|\mathcal{H}|=\bigO(n^k)$ but the number of iterations performed by the algorithm are now significantly reduced from $\bigO(n)$ to $\bigO(\log_{1+\xi}(n))$.


\section{Datasets}\label{appsec:implDetandData}
All the datasets can be publicly downloaded online, 
and the descriptions of the datasets can be found 
in previous papers. We briefly report a list of URLs and the works that introduced each dataset.
\begin{itemize}
\item the Bitcoin and Reddit datasets can be downloaded from \url{https://www.cs.cornell.edu/~arb/data/}~\cite{Liu2019},
\item the EquinixChicago dataset from:~\url{https://github.com/VandinLab/PRESTO}~\cite{Sarpe2021}. 
\item The Venmo dataset from:~\url{https://anonymous.4open.science/r/Financial_motif-7E8F/}~\cite{Liu2023} . As a note, such dataset is a small subset of a biggest set of financial transactions leaked from the Venmo platform and available publicly online for research purposes.
\item All the other missing datasets can be downloaded from:~\url{https://snap.stanford.edu/temporal-motifs/data.html}~\cite{Paranjape2016}.
\end{itemize}

\section{Missing experimental results}\label{appsec:experimentalResults}

\subsection{Results on medium sized datasets}\label{appsubsec:mediumdata}
\begin{figure*}[t]
	\centering
	\captionsetup[subfigure]{labelformat=empty}
	\begin{tabular}{l}
		\subfloat{\includegraphics[width=1\textwidth]{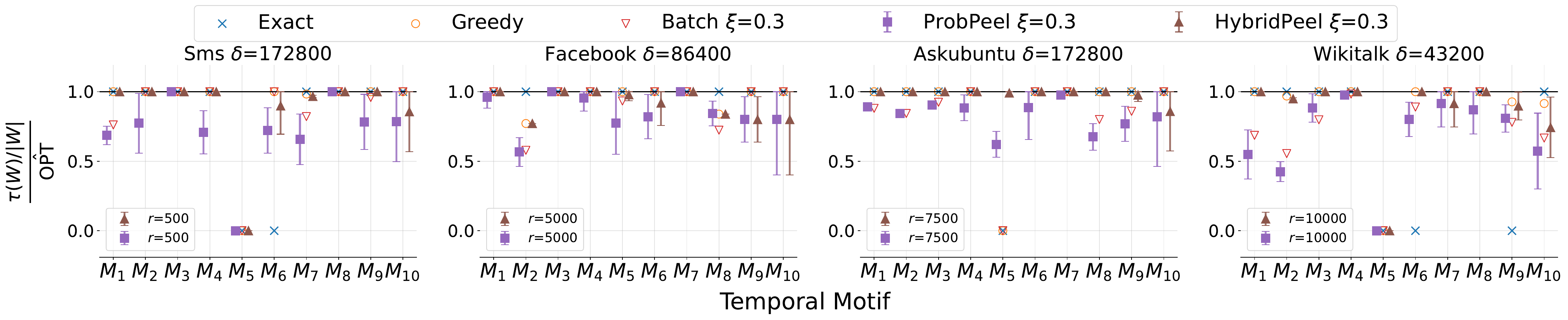}}\\[-4ex]
		\subfloat{\includegraphics[width=1\textwidth]{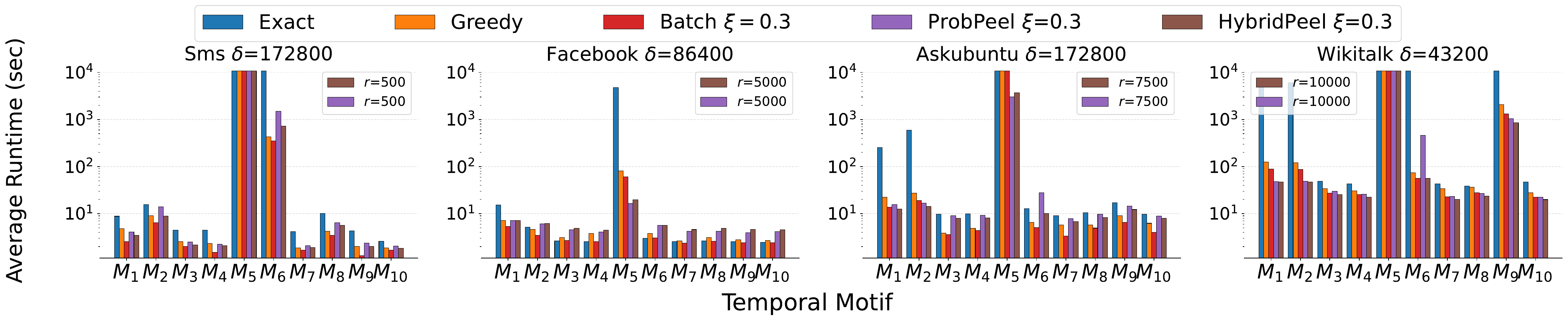}}\\[-4ex]
	\end{tabular}
	\caption{For each configuration we report (top): the quality of the solution compared to the best empirical solution (i.e., $\hat{\OPT}$), and (bottom): the running times to achieve such solutions. Running times are in \emph{log}-scale.}
	\label{fig:testApproxandRTMedium}
\end{figure*}

In this section we comment the missing results from Section~\ref{subsec:appxration}, i.e., the results relative to the solution quality and running time on the medium-size datasets from Table~\ref{tab:datasets}. The settings of such result is the same as from Section~\ref{subsec:appxration}, where we recall we set $J=2$ on medium-size datasets for \alghybrid. We observe similar trends to the large datasets we considered, in particular \algkapp generally outputs solutions with similar or better quality than \algbatch on most instances over all motifs, and both such algorithms have comparable running times. Interestingly, on such medium-size networks, for all motifs over most configurations the optimal densest temporal subnetworks can be identified by \algex, as in general this algorithm takes few tens of seconds to conclude its execution (except for specific configurations, such as the Wikitalk dataset or motif $\mathsf{M}_5$). Additionally, regarding the algorithms in \suitename, we also observe the same trends as for our large datasets (discussed in Section~\ref{subsec:appxration}), that is \algrand provides a solution with close approximation-ratio to the one provided by \algbatch, and \alghybrid often improves \algrand, again \alghybrid almost always matches the solution quality provided by~\algkapp. Interestingly, on Askubuntu, our randomized algorithms are the \emph{only} algorithms that are able to complete their execution under three hours on $\mathsf{M}_5$ where all baselines fail, which again highlights the scalability and efficiency of \suitename in hard settings, such as challenging temporal motifs.

\begin{figure}[t]
	\centering
	\begin{tabular}{c}
		\includegraphics[width=0.95\linewidth]{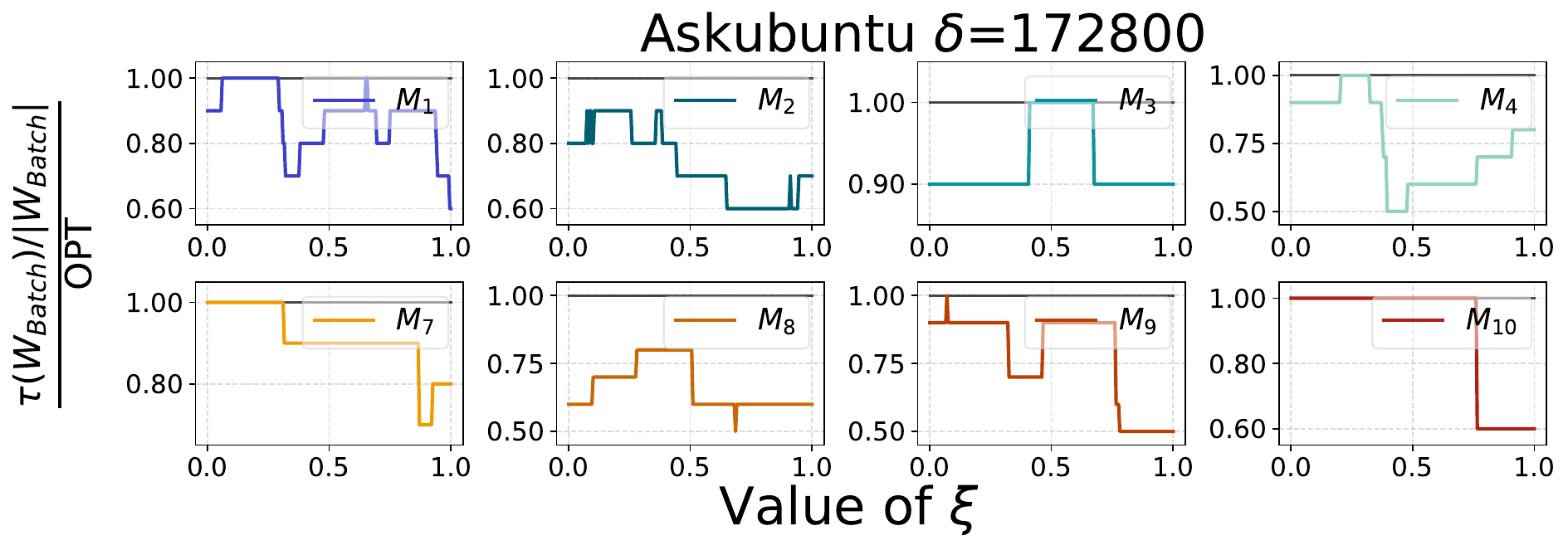}\\
		\includegraphics[width = 0.95\linewidth]{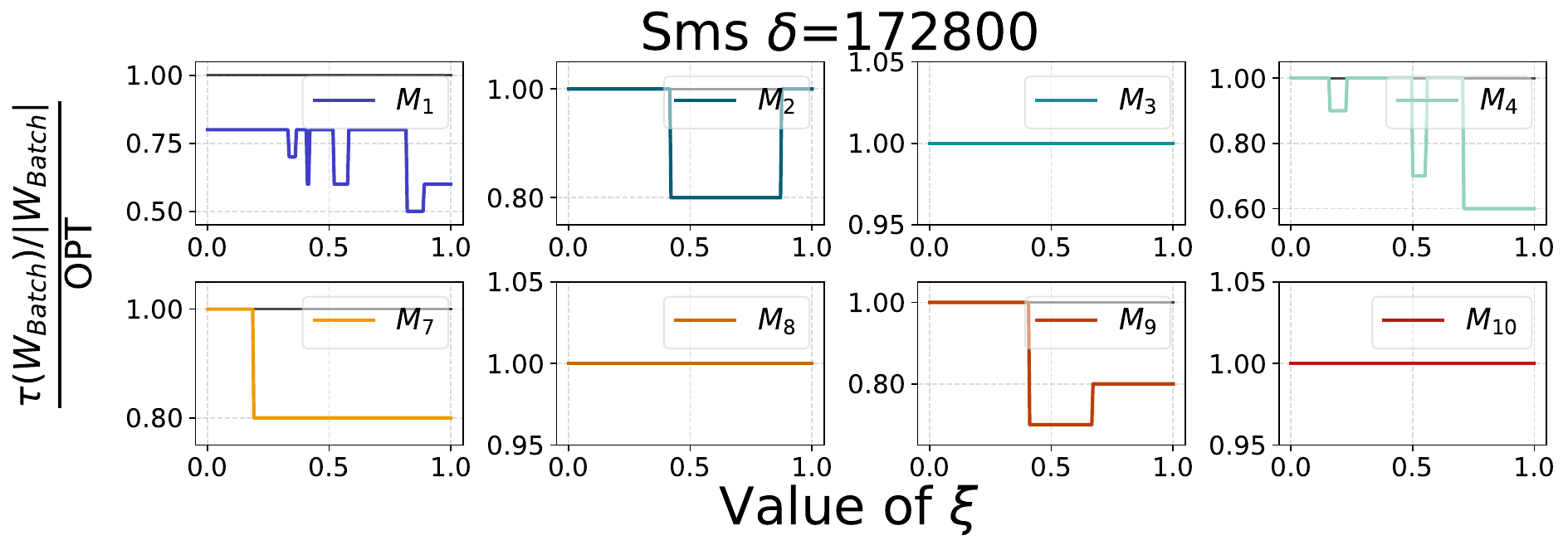}
	\end{tabular}
	\caption{Approximation ratio of~\algbatch for varying $\xi$.}
	\label{fig:xiincrease}
\end{figure}

\begin{figure}[t]
	\centering
	\begin{tabular}{cc}
		\includegraphics[width=0.43\linewidth]{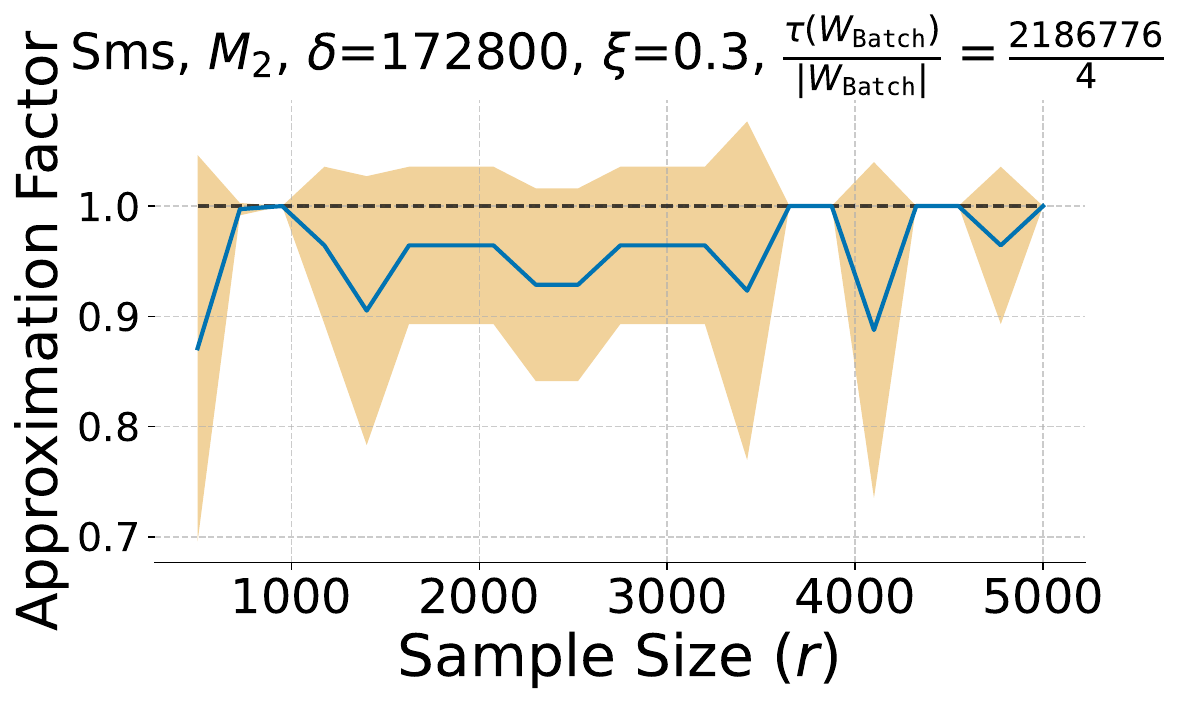} & \includegraphics[width=0.43\linewidth]{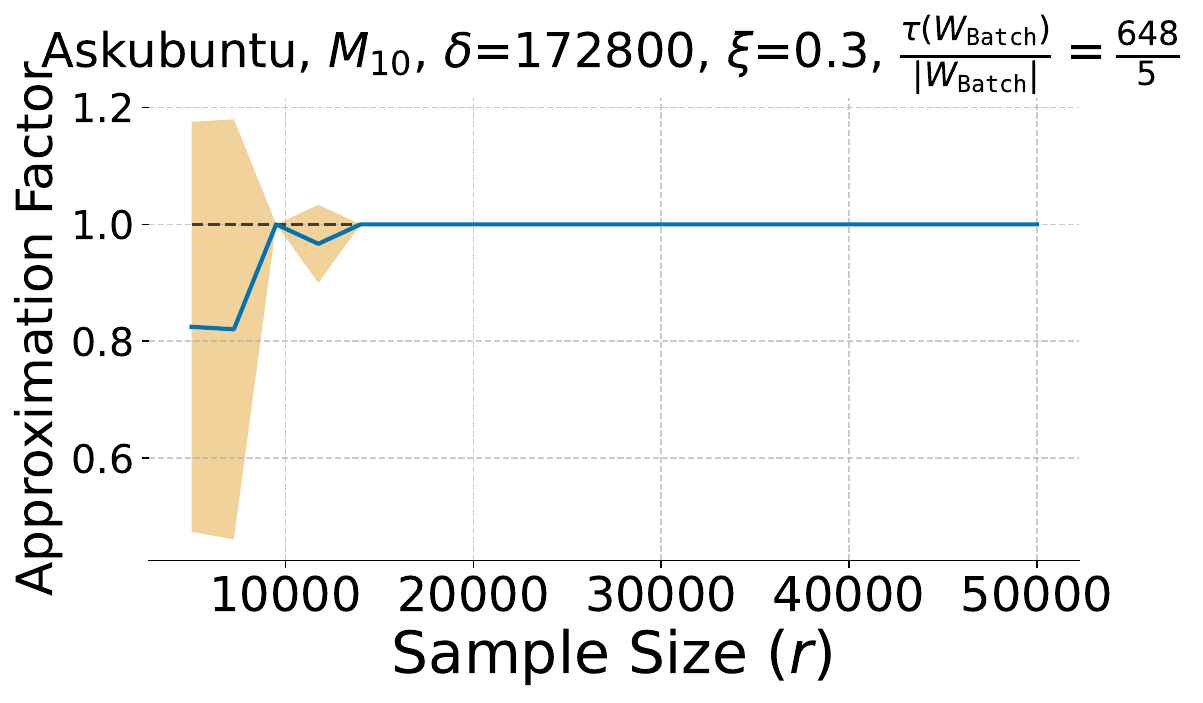}\\ \includegraphics[width=0.43\linewidth]{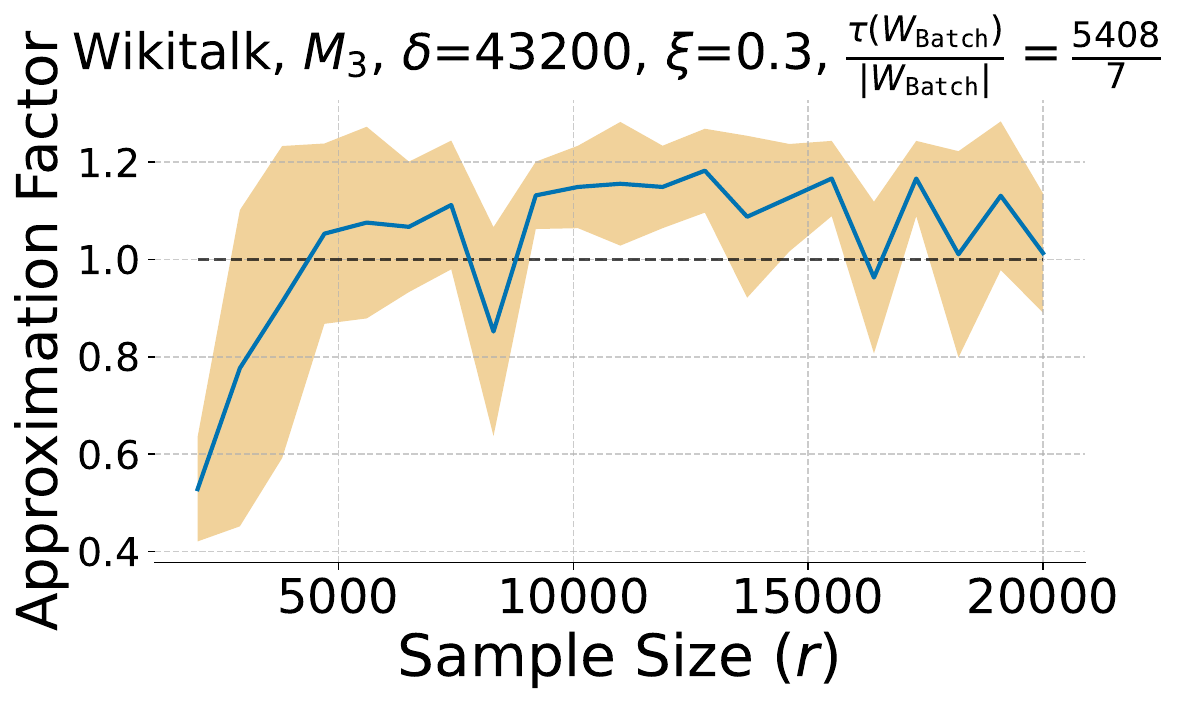} & \includegraphics[width=0.43\linewidth]{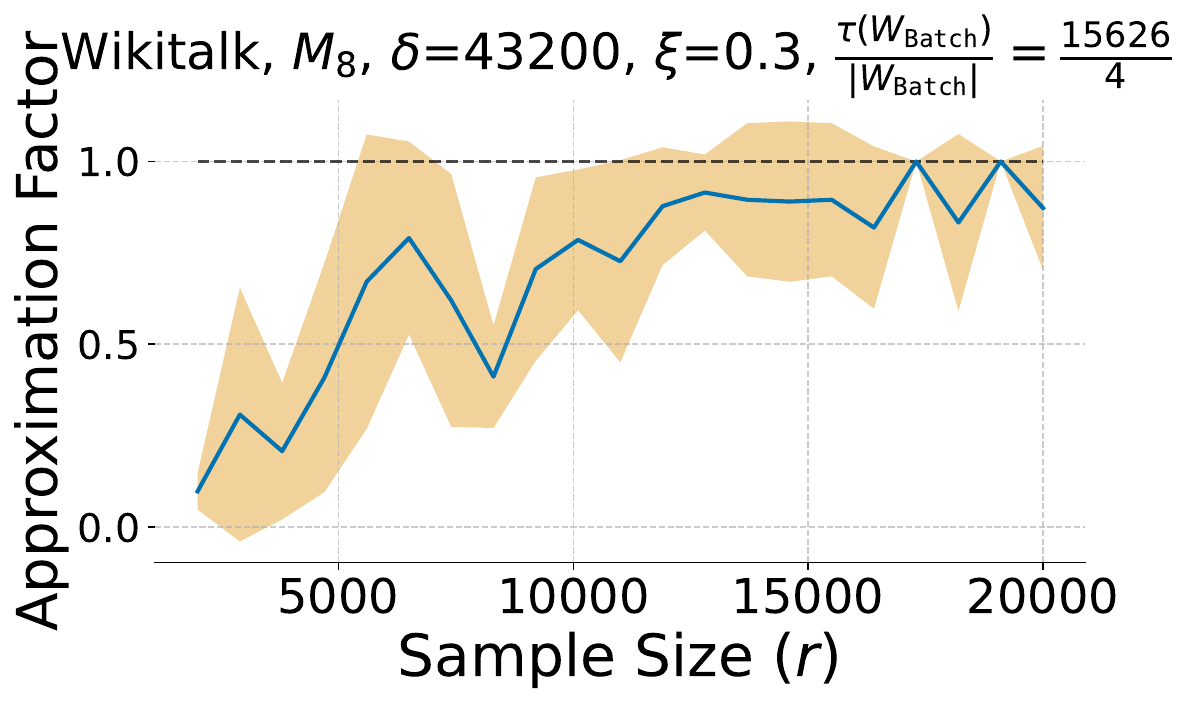}\\
	\end{tabular}
	\caption{Average approximation factor of~\algrand's solution compared to~\algbatch\ (i.e., $\fobj(W_{\mathrm{\algrand}})/\fobj(W_{\mathrm{\algbatch}})$), according to different values of the sample size $r$.}
	\label{fig:varyingR}
\end{figure}

\begin{figure*}[t]
	\centering
	\captionsetup[subfigure]{labelformat=empty}
	\begin{tabular}{ll}
		\multirow{2}{*}{\rotatebox[origin=c]{90}{Medium}} &\subfloat{\includegraphics[width=0.97\textwidth]{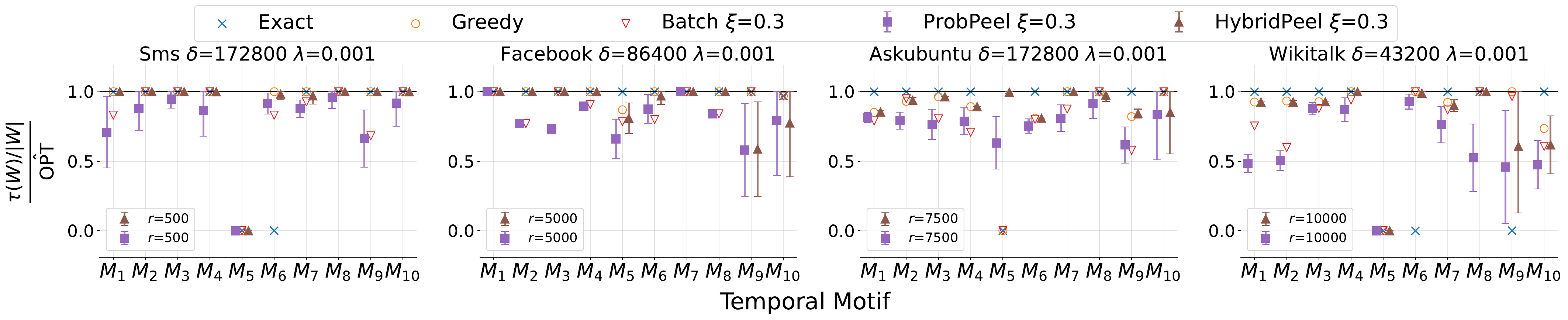}}\\[-4ex]
		&\subfloat{\includegraphics[width=0.97\textwidth]{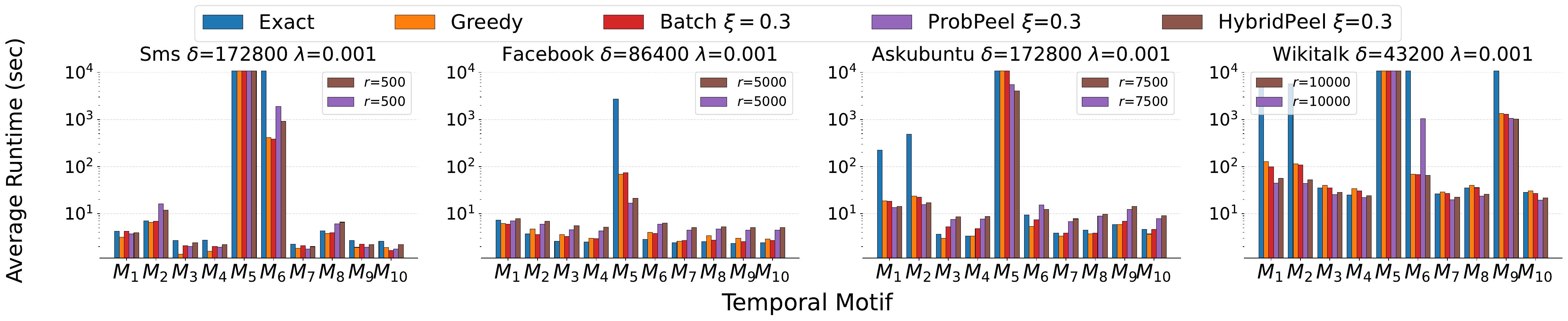}}\\ \cmidrule(l{15pt}r{15pt}){1-2}
		\multirow{2}{*}{\rotatebox[origin=c]{90}{Large}} &  \subfloat{\includegraphics[width=0.97\textwidth]{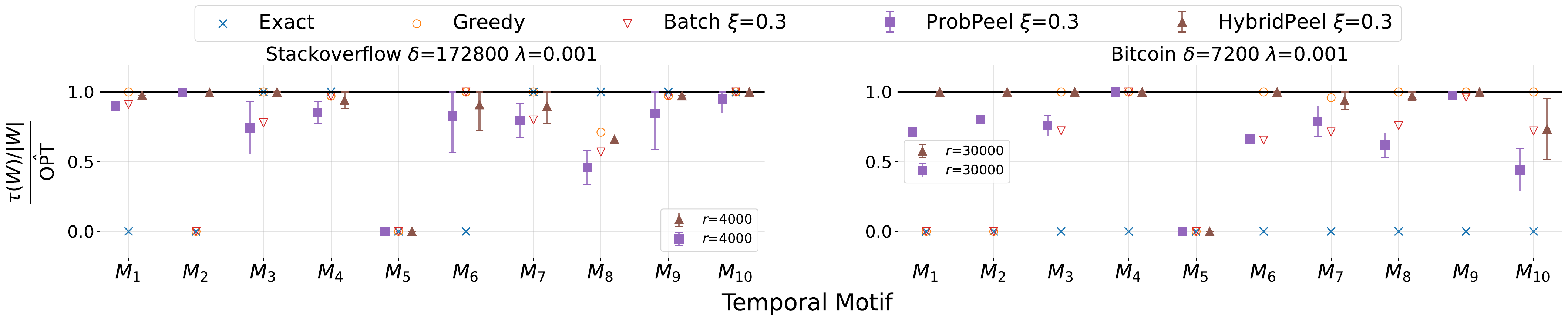}}\\[-4ex]
		&\subfloat{\includegraphics[width=0.97\textwidth]{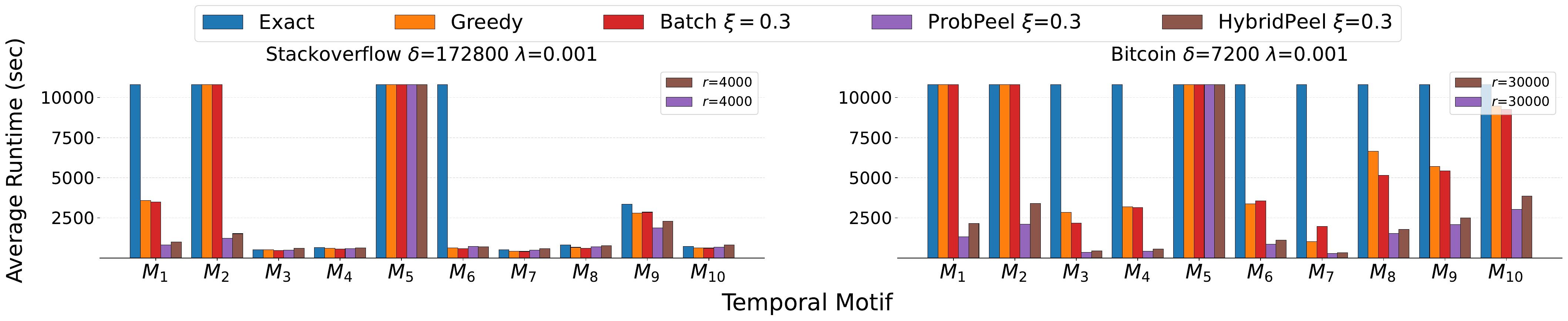}}\\[-4ex]
	\end{tabular}
	\caption{Results with exponential decaying function $\tau_d$ as from Equation~\eqref{eq:expDecay}, and $\lambda=0.001$, on medium and large datasets from Table~\ref{tab:datasets}. For each configuration we report (top): the quality of the solution compared to the best empirical solution (i.e., $\hat{\OPT}$), and (bottom): the running times to achieve such solutions.}
	\label{fig:testApproxandRTDecay}
\end{figure*}

\subsection{Sensitivity analysis}\label{subsec:sensitivity}
\subsubsection{Varying $\xi$} \label{subsubsec:varyingxi} 
We show how the solution obtained by~\algbatch\ varies according to the value of $\xi$, this setting is used to understand the best possible solution that can be obtained by \algrand according to various values of $\xi$ (i.e., when  $\varepsilon \approx 0$ then \algrand converges to the solution provided by \algbatch). We will use datasets Askubuntu and Sms considering the motifs of Fig.~\ref{fig:gridMotifs} (excluding $\mathsf{M}_5$ and $\mathsf{M}_6$ given their high running time). In particular, we start from $\xi=10^{-3}$ and increase it with a step of $5\cdot10^{-3}$ until it reaches 1
and we execute \algbatch with each different value of such parameter. We then compare each solution obtained by \algbatch to the optimal solution obtained by~\algex. The results are shown in Fig.~\ref{fig:xiincrease}. We see that in general smaller values of $\xi$ correspond to solutions with better approximation ratio obtained by \algbatch. On some instances by varying $\xi$ the solution may vary significantly (e.g., $\mathsf{M}_2$ on Askubuntu). In general the algorithm achieves satisfactory approximations on most configurations for $\xi<0.75$, given that for such range the value of the solution obtained by \algbatch is often within 80\% of the optimal solution found with~\algex. We also observe that in some settings the algorithm never outputs solutions with optimal densities (e.g., motif $\mathsf{M}_{2},\mathsf{M}_{8}$ on Askubuntu or $\mathsf{M}_1$ on Sms). This supports the design of \alghybrid, since for \algrand it may often be infeasible to converge to an optimal solution, as its guarantees are with respect to \algbatch (see Theorem~\ref{lemma:ratioSampling}).

\subsubsection{Varying $r$}
\label{subsubsec:varyingR} 
We show the convergence of the solution obtained by~\algrand\ to~\algbatch's solution by increasing the sample size~$r$. Recall that we set $r_i$ to a fixed value of $r$ for each iteration of~\algrand. For all the motifs of Fig.~\ref{fig:gridMotifs} (except $\mathsf{M}_5$) and each medium-size dataset of Table~\ref{tab:datasets} we vary the parameter $r$ as follows: we select a starting value $s$ and compute the approximation obtained by~\algrand\ on each motif, dataset and $\delta$ (as in Table~\ref{tab:datasets}) over five runs for each value of $r$ spanning the range $[s,10s]$. That is, we increase at each step $r$ starting from $s$ by $9s/30$ so that we compute the solutions over 30 different values of such parameter. We set $s$ to 500, 3000, 5000, 2000 for the medium-size datasets of Table~\ref{tab:datasets} following their order in the table. We only comment some representative behaviors (in Fig.~\ref{fig:varyingR}) illustrating possible convergence rates of~\algrand\ over different inputs (motifs and~datasets). 

We first note that, as suggested by Lemma~\ref{lemma:ratioSampling}, as the sample size~$r$ increases~\algrand\ converges to a solution value close to the one obtained by~\algbatch, and we observe such a behavior across \emph{all} the experiments we perform. Additionally, we observe different convergence behavior, such as fast convergence  ($\mathsf{M}_{10}$ on Askubuntu), convergence to a slightly better solution than with~\algbatch\ ($\mathsf{M}_3$ on Wikitalk), alternating solution, i.e., the algorithm switches between solution with similar values ($\mathsf{M}_2$ on Sms) and slow convergence, i.e., the algorithm requires a large sample to improve the solution quality ($\mathsf{M}_8$ on Wikitalk). These results support the theoretical guarantees obtained for \algrand\ and that the algorithm achieves tight convergence in most configurations. But also motivate the design of \alghybrid, combining such results with the one of Fig.~\ref{fig:xiincrease}, that is \algrand may often converge to a sub-optimal density even for small $\xi$.

\subsection{Results with exponential decaying function}\label{appsubsec:expdecay}
In this section we briefly evaluate the algorithms in \suitename against the baselines by using the exponential decaying weighting function $\tau_{d}$ that we introduced in Section~\ref{sec:prelims} and we recall that the \emph{decaying} function is defined as: 
\begin{equation}\label{eq:expDecay}
\tau_{d}: \tau_{d}(S) =\frac{1}{\ell-1} \sum_{j=2}^\ell \exp(-\lambda (t_j - t_{j-1}))
\end{equation}
where $\lambda > 0$ controls the  temporal decay over the temporal network. The setup is the same as described in Section~\ref{subsec:appxration}, where instead of using the constant function ($\tau_{c}$) that assigns a unitary weight to each $\delta$-instance in $\mathcal{S}_T$, we use the decaying weighting function with $\lambda=0.001$. As large datasets we will only consider Stackoverflow and Bitcoin since those are the ones where most algorithms (including \algkapp and \algbatch) except \suiteex were able to complete under the setting of Section~\ref{subsec:appxration} on most of configurations. The results are reported in Fig.~\ref{fig:testApproxandRTDecay}. We observe that the results follow the trends we discussed in Section~\ref{subsec:appxration}, among the baselines \algex is not able to conclude its execution on challenging configurations (i.e., datasets and motifs), and \algkapp outputs solutions with high approximation ratios (often over 0.8),  interestingly on some datasets such as Wikitalk and Askubuntu the approximation ratio achieved by \algkapp is smaller than the one achieved on the same configurations (dataset, motif and value of $\delta$) in Section~\ref{subsec:appxration}. This may be caused by the fact that with $\tau_{d}$ and $\lambda=0.001$ the vertices are peeled according to the weights of instances that are significantly smaller than those obtained with $\tau_{c}$. \algbatch achieves again satisfactory performances and trading off accuracy for efficiency, especially on large data (e.g. on Bitcoin).  Again our randomized algorithms in \suitename (\alghybrid and \algrand) become really efficient and scalable on large datasets saving significant running time on datasets Stackoverflow (e.g., more than $5\times$ faster than \algkapp and \algbatch on $\mathsf{M}_1$) and Bitcoin where we have at least a speedup of $\times2$, but more interestingly our algorithms \emph{scale} their computation on hard instances where existing baselines fail to compute a solution. Concerning the solution quality again \algrand matches the approximation quality provided by \algbatch while being more efficient and scalable on large data as captured by our analysis. And \alghybrid almost always matches \algkapp and \algex, and sometimes also improves over \algkapp (e.g., $\mathsf{M}_3$ on Askubuntu) achieving almost optimal results on most of the motifs considered, being therefore the best algorithm in terms of both runtime and approximation quality on most configurations.

Overall, the results discussed in this section confirm the findings of Section~\ref{subsec:appxration} with the small difference that, interestingly, some approximation algorithms (e.g., \algkapp) may slightly decrease their effectiveness when used with $\tau_{d}$ and some other maintain high-quality solutions in output (e.g., \alghybrid on large datasets) compared with the approximation qualities achieved under $\tau_{c}$.

\ifextended
\begin{figure}[t]
	\centering
	\subfloat{\includegraphics[width=1\linewidth]{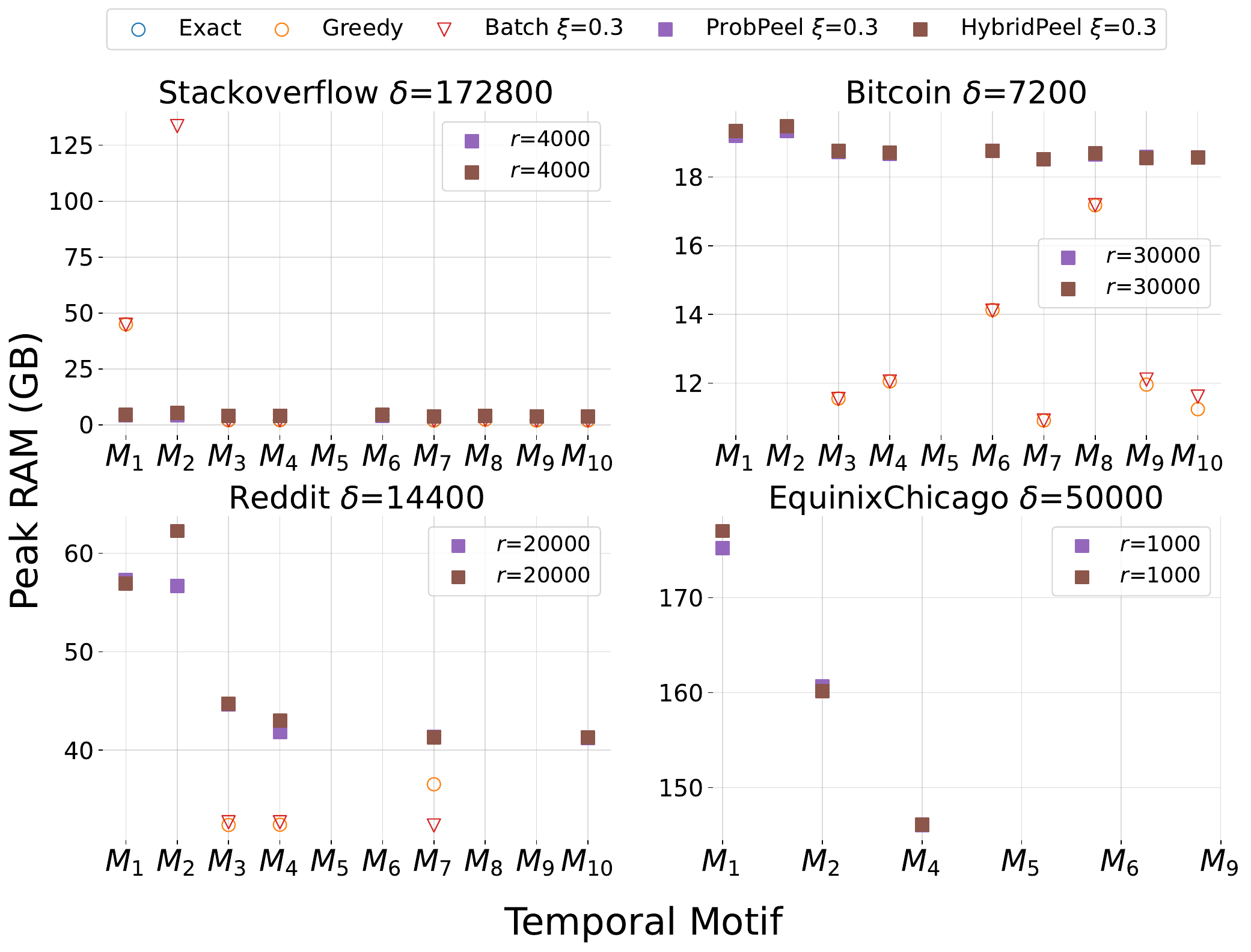}}
	\caption{Peak RAM memory in GB over one execution.}
	\label{fig:memory}
\end{figure}

\begin{figure}[t]
	\centering
	\begin{tabular}{l}
		\includegraphics[width=0.9\linewidth]{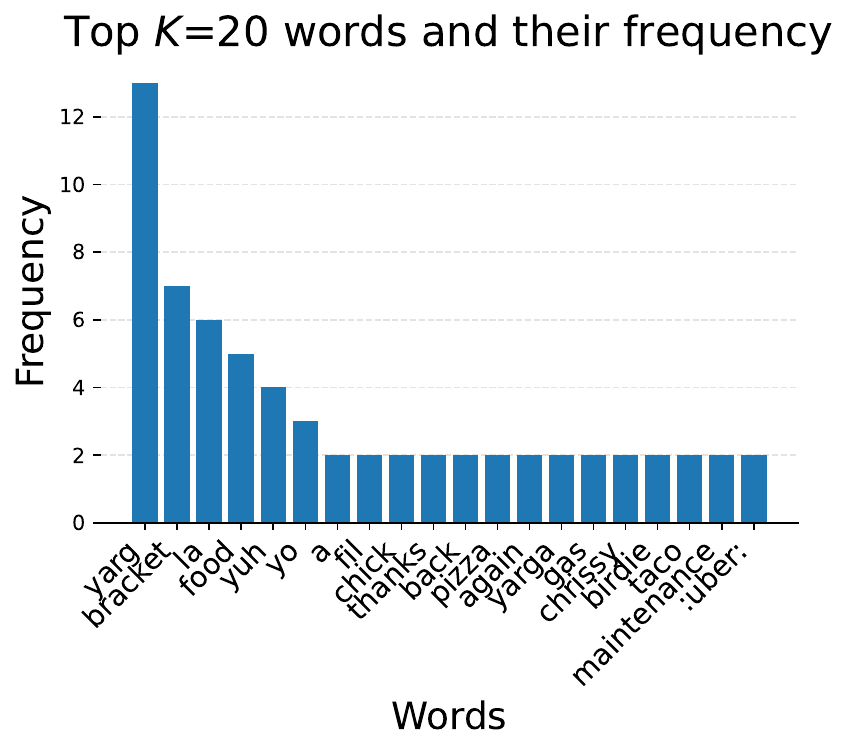} 
	\end{tabular}
	\caption{Histogram associated to the words inside messages collected on the edges of $T[W^*]$ on the Venmo dataset, for $\mathsf{M}_5$ and $\delta=172800$ (i.e, two days). See Section~\ref{subsec:casestudy} for more details, and Fig.~\ref{fig:venmoStarSols} for a representation of the optimal subnetwork.}
	\label{fig:vemoWordCloud}
\end{figure}

\subsection{Memory usage}
\label{subsubsec:memory} 
We measured the peak RAM memory over one single execution of each algorithm on the configurations of Section~\ref{subsec:appxration}. We show the results in Fig.~\ref{fig:memory} --- recall that the memory limit was set to 150\,GB on all datasets but EquinixChicago where the limit was set to 200\,GB. We do not report data for algorithms that do not finish within three hours.

Overall, the memory usage of the different algorithms strongly depends on the temporal motif considered, and in general motifs $\mathsf{M}_1,\mathsf{M}_2$ and $\mathsf{M}_5$ require much more memory that  other motifs on almost all datasets. We observe that on such motifs~\algrand and \alghybrid use much less memory compared to the baselines, saving on $\mathsf{M}_2$ more than 100\,GB on datasets Stackoverflow and Bitcoin and about 90\,GB on the Reddit dataset. This is due to the fact that, differently from all other algorithms, our randomized algorithms do not store the set $\mathcal{H}$ of $k$-CIS  (see Table~\ref{tab:ALDENTEalgs}).
\else
\fi
We finally note the randomized algorithms may use slightly more memory than the baselines on motifs requiring a small memory usage, since we store a copy of the temporal edges of $T$ to evaluate the solution of the returned vertex set on the original network, and such step can be easily be avoided (see Section~\ref{subsec:appxration}).

These results, coupled with those in Section~\ref{subsec:appxration}, show that the algorithms in \suitename, i.e., \algrand and \alghybrid, are the only practical algorithm to obtain high-quality solutions for the \ourproblem problem on massive networks, requiring a much smaller amount of memory on most of the challenging configurations.


\ifextended
\subsection{Case study}\label{appsubsec:caseStudy}
We report in Figure~\ref{fig:vemoWordCloud} the frequencies of top-20 words over the transactions associated to the optimal subnetwork from Figure~\ref{fig:venmoStarSols} (right). We observe terms related to social activities (i.e., food, pizza, etc.), and terms that may be related to gambling (e.g., bracket, yarg), we removed emojis from the histogram. 
\else
\fi

\fi

\end{document}